\documentclass[a4paper,12pt]{book}
\usepackage{blkarray}
\usepackage{amsfonts}
\usepackage{footmisc}
\usepackage{amsmath}
\usepackage{setspace}
\usepackage[top=1.5in, bottom=2in, left=1in, right=1in]{geometry}
\usepackage{graphicx}
\usepackage{amssymb}
\usepackage{fancyhdr}
\providecommand{\U}[1]{\protect\rule{.1in}{.1in}}
\newtheorem{theorem}{Theorem}

\newtheorem{definition}[theorem]{Definition}

\newtheorem{notation}[theorem]{Observation}

\newenvironment{proof}[1][Proof]{\noindent\textbf{#1.} }{\ \rule{0.5em}{0.5em}}

\def\fr{\frac}
\def\be{\begin{equation}}
\def\ee{\end{equation}}
\def\ba{\begin{eqnarray}}
\def\ea{\end{eqnarray}}
\def\pa{\partial}
\def\ra{\rightarrow}
\def\na{\nabla}
\def\s{\sigma}
\def\l{\lambda}
\def\e{\varepsilon}
\def\a{\alpha}
\def\b{\beta}
\def\g{\gamma}
\def\d{\delta}
\def\t{\tau}

\def\pt{\phantom{a}}
\def\w{\omega}

\def\f{\varphi}

\def\L{\Lambda^4}
\def\S{\mathbf{S}_\Lambda}
\def\f{\varphi}

\newcommand{\captionfonts}{\large}

\makeatletter  % Allow the use of @ in command names
\long\def\@makecaption#1#2{%
  \vskip\abovecaptionskip
  \sbox\@tempboxa{{\captionfonts #1: #2}}%
  \ifdim \wd\@tempboxa >\hsize
    {\captionfonts #1: #2\par}
  \else
    \hbox to\hsize{\hfil\box\@tempboxa\hfil}%
  \fi
  \vskip\belowcaptionskip}
\makeatother   % Cancel the effect of \makeatletter

\begin{document}

\title{Arrangement Field Theory \\ beyond Strings and Loop Gravity}
\author{Diego Marin\thanks{dmarin.math@gmail.com} \\ \small
with contributions by \\ \small Pangea Association\thanks{www.gruppopangea.com/?page\_ id=682\& lang=en}
$\qquad\qquad$ Fabrizio Coppola\thanks{fabrcop@aliceposta.it}, \\
\small Marcello Colozzo\thanks{extrabyte2000@yahoo.it} $\qquad\qquad$ Istituto Scientia\thanks{www.istitutoscientia.it, Via Ortola 65, 54100, Massa (MS), Italy}}
\date{}
%\titlepic{\includegraphics[width=0.8\textwidth]{spacetime.jpg}}

\maketitle
\large
\setcounter{page}{1}
\tableofcontents

\chapter{Abstracts}

This paper regroups all contributions to the arrangement field theory (AFT),
together with a philosophical introduction by Dr. Fabrizio Coppola.
AFT is an unifying theory which describes gravitational, gauge and fermionic
fields as elements in the super-symmetric extension of Lie algebra $Sp(12,\mathbf{C})$.

$\pt$

\noindent \textbf{Paper number 1}

\noindent We introduce the concept of \lq\lq non-ordered space-time'' and formulate a quaternionic field theory
over such generalized non-ordered space. The imposition of an order over a non-ordered space
appears to spontaneously generate gravity, which is revealed as a fictitious force. The same
process gives rise to gauge fields that are compatible with those of Standard Model. We suggest a
common origin for gravity and gauge fields from a unique entity called \lq\lq arrangement
matrix'' ($M$) and propose to quantize all fields by quantizing $M$. Finally we give a proposal for
the explanation of black hole entropy and area law inside this paradigm.

$\pt$

\noindent \textbf{Paper number 2}

\noindent In this work we apply the formalism developed in the previous paper (\lq\lq The arrangement field
theory'') to describe the content of standard model plus gravity. The resulting scheme finds an
analogue in supersymmetric theories but now all quarks and leptons take the role of gauginos for
$Sp(12,\mathbf{C})$ gauge fields. Moreover we discover a triality between \emph{Arrangement Field Theory},
\emph{String Theory} and \emph{Loop Quantum Gravity}, which appear as different manifestations of
the same theory. Finally we show as three families of fields arise naturally and we discover a new
road toward unification of gravity with gauge and matter fields.

$\pt$

\noindent \textbf{Paper number 3}

\noindent We show how antigravity effects emerge from arrangement field theory.
AFT is a proposal for an unifying theory which joins gravity with gauge
fields by using the Lie group $Sp(12,\mathbf{C})$. Details of theory have been
exposed in the papers number 1 and number 2.

\chapter[The philosophy of AFT]{The philosophy of arrangement field theory}
\label{intro}

\section{Classical Physics \label{classicalphy}}

In classical physics, space and time are fundamental entities, pro\-vi\-ding a preordained
structure in which interactions between physical objects can occur.
In short, space and time are \lq\lq absolute".
Moreover, the physical properties of a body or system are supposed to be objective and independent from a
possible observation.

In this paradigm, reality exists independently of classical measurements and is not significantly
influenced by measurements, unless these are particularly \lq\lq invasive". But even in such cases,
it is assumed that the observed systems had their own pre-existing cha\-ra\-cte\-ri\-stics.

These were obvious and implicit tenets in classical physics, which influenced whole science,
aimed to be purely objective.

\section{Space and time according to philosophers \label{philosophers}}

Despite the rapid and successful development of classical physics and science in
general, firmly based on the fixed concepts of space and time, between late 17th century and early
19th century respectable philosophers such as Locke, Hume, Leibniz, Kant and Schopenhauer,
conceptualized space and time not as objective and universal entities, but as concepts defined by
our own intellect, aimed to interpret the external reality perceived by our senses.

This idea was radically different from the founding conception of classical physics, based on
full objectivity, and appeared quite extravagant to several scientists at that time.
Nevertheless Kant, who had a scientific background, exposed his conception in a profound and rational way.

In 1781 Kant distinguished two main activities of conscious mind \cite{kant}: \lq\lq analytic
propositions" and \lq\lq synthetic propositions". In an oversimplied interpretation, \lq\lq analytic
propositions" are the elements of rational, logical reasoning, in which thoughts
proceed by deduction, starting from known facts and finding consequences which,
anyway, were implicit in the premises and only had to made explicit by reasoning.

\lq\lq Synthetic propositions", instead, are new, non-deductible informations, coming
from perceptions and sensations. For instance we can not deduce whether an apple is sweet,
or a radiator is hot, but we must check that through our senses.

Kant also proposed a distinction between \lq\lq a priori" propositions, meaning \lq\lq in advance",
ie \lq\lq before" an experience is performed; and \lq\lq a posteriori" propositions, meaning \lq\lq after"
an experience.

According to Kant, all analytic propositions are \lq\lq a priori". A trivial example is given by any sum,
such as 4 + 7 = 11.
This analytic proposition is true \lq\lq a priori": the result is already 11 before we make the calculation.
Kant states that no analytic proposition can be a \lq\lq posteriori". Synthetic propositions, on the other
side, are generally \lq\lq a posteriori", since perceptions come from experience.

Now, an interesting question remains: may \lq\lq a priori" synthetic propositions exist?
Kant answers that they do actually exist. Certain \lq\lq categories" that human mind applies to events,
such as the principle of \lq\lq cause and effect", are \lq\lq a priori". In fact we perceive events and
relate to each other according to a category, \lq\lq causality", which, according to Kant, already exists
in our intellect.

Kant states that \lq\lq space" and \lq\lq time" are also \lq\lq a priori" synthetic forms.
Even if space, time and causality are related to experience, Kant does not consider them as inherent
to the objective phenomena, but as subjective tools (even if they manifest themselves as universal)
that our intellect uses to \lq\lq order" the experiences.

After Kant's definitions, anyway, classical, mechanistic science continued to achieve extraordinary results.
However, in the early twentieth century, physics started to face unexpected problems and
contradictions, that forced scientists to formulate new principles and accept radical changes.

\section{Relativistic physics \label{relativistic}}

In 1632 Galileo had intuited and enunciated the \lq\lq principle of relativity", stating that
the laws of physics are the same in every inertial frame of reference \cite{galileo1}.

Later developments of physics, including several discoveries in optics and electromagnetism,
suggested instead that a privileged, steady, fundamental frame of reference should exist.
This issue especially afflicted electromagnetism, that was an excellent theory but
included certain unsolved inconsistencies.

In 1905 Einstein solved the whole problem, restarting from the Galileo's principle of relativity
and applying it to the new knowledge of electromagnetism and optics, thus developing an original,
consistent theory, \lq\lq special relativity"\cite{einstein1}.
His theory also accounted for the results of the Michelson-Morley experiment \cite{mmorley},
conducted in 1887, which had demonstrated that speed of light does not follow the classical laws
of velocity addition.
Einstein solved all the inconsistencies by proposing that the speed of light $c$ is independent from
the motion of the emitting body. The universal constant $c$ became an insurmountable speed limit
in physics.

Einstein's theory also implied new, counter-intuitive ideas:
for example, time flows differently in different inertial reference systems, and perception
of space also depends on the frame of reference of the observer.
In light of such new discoveries, Kant's ideas do not seem to be so extravagant anymore.

Space and time lose their absolute characteristics if considered independently
from each other, but, adequately considered as components (coordinates)
of four-dimensional points, remain \lq\lq absolute" (\lq\lq invariant") in a single entity,
\lq\lq space-time" or \lq\lq chronotope", ruled by a generalized geometrical entity including time as the fourth coordinate.
In 1908 such a four-dimensional structure was perfected and named \lq\lq Minkowski space" \cite{minko}.

In 1916 Einstein expanded the principle of relativity to non-inertial reference frames, thus
defining the new theory of \lq\lq general relativity", in which the four-dimensional geometry is
curved by the presence of the masses \cite{einstein3}. Hence, even the (linear) Minkowski space
had to be considered as an approximation, valid only in small regions of the (curved) universe.
In this perspective, \lq\lq gravitational forces" find their natural explanation
in geometrical terms, based on a specific concept of metric.

This approach also affected the interpretation of the principle of cause-effect, to the point that Einstein,
in paragraph $a2$, wrote:
\lq\lq The law of causality has not the signicance of a statement as to the world of experience,
except when observable facts ultimately appear as causes and effects" \cite{einstein3}.
Kant had exposed this \lq\lq extravagant" idea a long time before \cite{kant}.

In this paper we suggest a new step in the direction of \lq\lq relativization" (so to say),
by questioning the absolute ordering of the space-time points, that we believe is an imposition
made by our intellect, rather than a proper quality of Nature. Such conjecture might open new unexpected
perspectives for understanding the fundamental fields of physics, as we are going to see.

\section{Quantum limitation of objectivity \label{limitation}}

In 1900 Planck had proposed \lq\lq quantization" of energy to
explain the electromagnetic emission of a \lq\lq black body" \cite{planck}.
In 1905 quantization of energy was also applied by Einstein to explain the \lq\lq photoelectric effect" \cite{einstein2}.

The several discoveries that clarified the structure of the atom from
1905 to the 1930's included the Rutherford's experiment \cite{rutherford}
in 1911, and the consequent Bohr model \cite{bohr1} in 1913.
Bohr started from the results of the Rutherford's experiment, and imposed quantization
to the angular momentum of electrons, instead of quantizing energy directly. As a consequence,
energy also turned out to be quantized, and the calculated levels were in excellent
agreement with the experimental values. The agreement was nearly perfect in the case of hydrogen,
the simplest atom in Nature.

In the case of more complex and heavier chemical elements, the
mathematical frame was more difficult and the results were less precise. To solve
these problems, the complete theory of Quantum Mechanics (QM) was gradually
developed (mainly by the \lq\lq Copenhagen school" directed by Bohr himself
during the 1920's), which came out to be intuitively abstruse, offering no image of the
motion of the electrons around the atomic nucleus.

While developing QM, it began to emerge that the experiments
inevitably influenced the observed systems. Bohr, Heisenberg and other physicists
of the \lq\lq Copenhagen school" suspected that physical properties of
quantum systems could no longer be assumed to be completely predefined
and ontologically independent from observation.

In the first version of the \lq\lq Copenhagen interpretation"
they assumed that free will of the conscious observers played a decisive role in the collapse
of a quantum state into an eigenstate \cite{heisenberg1}.
This appeared as an unacceptable extravagance to many physicists, including
Einstein, because of the unexpected restrictions that the supposed objectivity of the
universe had to suffer, as a consequence of the new theory.

Quantum states evolve deterministically according to the Schr\"o\-din\-ger equation \cite{schr},
formulated in 1926, but remain devoided of certain characteristics, which can be revealed
(\lq \lq objectivated") only when the quantum state collapses into an \lq\lq eigenstate" of the measured
physical quantity. This is the main reason why physical quantities in QM are called \lq\lq observables".

QM \lq\lq works fine" only if it is accepted that such hidden properties are not objectively defined before
the measurement and are partly created by observation itself, when the state is reduced to an eigenstate.
The eigenvalues calculated according to QM are in excellent agreement with the possible outcomes
given by experiments, even though the theory can not predict which eigenvalue will come out: only the
respective probabilities can be calculated, as pointed out by Born \cite{born} in 1926.
This led in 1927 to the Heisenberg's \lq\lq uncertainty principle" \cite{heisenberg2}, which put
an end to the absolute determinism that was implicit in classical physics.

QM thus introduced a margin of \lq\lq uncertainty", in which Nature may reserve a small room for Her
non-predictable \lq\lq caprice" or \lq\lq willingness", according to Jordan \cite{jordan}, and secondarily \cite{heisenberg1} accepted
by Pauli, Wigner, Eddington, and von Neumann \cite{von}, and years later by Wheeler \cite{wheeler}, Stapp \cite{stapp1}, and other physicists.
For example, Stapp in 1982 defined human mental activity as \lq\lq creative",
because it only partially undergoes the course of causal mechanisms, having a margin for free choices \cite{stapp1}.

Another important consequence concerns the act of measurement, after which, the
subsequent course of the physical system under observation is unavoidably modified by
the measurement itself, so that observations inevitably imprint different directions to events.

In 1932 von Neumann, after reordering and formalizing QM into a consistent theory,
stated that a distinctive element was necessary to trigger the quantum \lq\lq collapse"or \lq\lq reduction",
and declared that the consciousness of an observer could be such an element, distinctive enough from the usual
physical quantities \cite{von}. In 2001 Stapp consistently explained this concept in detailed and clear terms \cite{stapp2}.

In 1935 the discussion about the interpretation of QM faced the problem introduced by the Einstein, Podolski and Rosen (EPR) paradox \cite{epr}, \cite{bridge},
that later, in 1951, was better defined by Bohm \cite{bohm}. In this well-known thought experiment,
two particles in quantum \lq\lq entanglement" but far away from each other, produce instant, non-local influences, in contradiction
with the upper limit set by relativity at the speed of light: E., P. and R. considered that as absurd and impossible.

Nevertheless, the experimental version that was defined by the Bell's theorem \cite{bell} in 1964, and implemented in 1982 by Aspect et al. \cite{aspect},
confirmed the existence of non-local influences due to the entanglement. Thus, a conflict seems to exist between special relativity
(that does not allow non-local influences) and QM (which includes and reveals such influences).
The subsequent theories have not been able to solve in a convincing way such a dissonance.
The conjecture exposed in this paper, however, may offer a new framework where such conflict can be finally overcome.

\begin{flushright}
Fabrizio Coppola, Istituto Scientia
\end{flushright}

\chapter{The arrangement field theory (AFT)}

\section{Introduction to formalism}
\label{sec:1}

The arrangement field paradigm describes the universe be means of a graph (ie an ensemble
of vertices and edges). However there is a considerable difference between this framework
and the usual modeling with spin-foams or spin-networks. The existence of an edge which
connects two vertices is in fact probabilistic. In this way we consider the vertices as
fundamental physical quantities, while the edges become dynamic fields.

In section \ref{reciprocal} we introduce the concept of non-ordered space-time, ie an
ensemble of vertices without any information on their mutual positions.
In section \ref{arrmatrix} we define the \lq\lq arrangement matrix'' ($M$), which is a matricial
field whose entries define the probability amplitudes for the existence of edges.
The arrangement matrix regulates the order of vertices in the space-time, determining
the topology of space-time itself. In the same section we extend the concept of derivative
on such non-ordered space-time.

In section \ref{ord} we define a simple \lq\lq toy-action'' for a quaternionic field in a
non-ordered space-time. We show how the imposition of an arrangement in such space-time
generates automatically a metric $h$ which is strictly determined by $M$.

In section \ref{local} we discover a low energy limit under which the \lq\lq toy-action''
becomes a local action after the arrangement imposition.

In section \ref{spin} we show that a new interpretation of spin nature arises spontaneously
from our framework. In the same section, the role of \lq\lq arrangement matrix'' is compared
to the role of an external observer.

In section \ref{symmetry} we anticipate some unpublished results regarding the
availment of our framework to describe all standard model interactions.

In section \ref{entropy} we apply a second quantization to the \lq\lq arrangement matrix'',
turning it in an operator which creates or annihilates edges. We show how this process can
give a new interpretation to black hole entropy and area law. We infer that quantization of
$M$ automatically quantizes $h$, apparently without renormalization problems.

\section{A non-ordered universe \label{unordered space}}

\subsection[Reciprocal relationship between space points]{Reciprocal relationship between space-time points \label{reciprocal}}

Every euclidean $4$-dimensional space can be approximated by a graph $\L$, that
is a collection of vertices connected by edges of length $\Delta$. We recover the
continuous space in the limit $\Delta \ra 0$. Moreover we can pass from the euclidean
space to the lorenzian space-time by extending holomorphically any function in the
fourth coordinate $x_4 \ra ix_4$ \cite{minko}.

In non commutative geometry, one can assume that a first vertex is connected to a second,
without the second is connected to the first. This means that connections between
vertices are made by two oppositely oriented edges, which we can represent by a couple of arrows.

We assume the vertices as fundamental quantities. Then we can select what couples of
vertices are connected by edges; different choices of couple generated different
graphs, which in the limit $\Delta \ra 0$ correspond to different spaces.

Our fundamental assumption is that the existence of an edge follows a probabilistic law,
like any other quantity in QM. We draw any pair of vertices, denoted by $v_{1}$ and $v_{2}$,
and we connect each other by a couple of arrows oriented in opposite directions.

Before proceeding, we extend the common definition of amplitude probability. Usually this is
a complex number, whose square module represents a probability and so is minor or equal to one.

We define instead the amplitude probability as an element in the division ring of quaternionic
numbers, commonly indicated with $\mathbf{H}$. Its square module represents yet a probability
and so is minor or equal to one. A quaternion $q$ have the form $q = a+ib+jc+kd$ with $a,b,c,d
\in \mathbf{R}$, $i^2 = j^2 = k^2 = -1$ and $ij = -ji = k$, $jk = -kj =i$, $ki = -ik =j$.

We write a quaternionic number near the arrow which moves from $v_1$ to $v_2$.
It corresponds to the probability amplitude for the existence of an edge which
connects $v_1$ with $v_2$. We do the same thing for the other arrow, writing the
probability amplitude for the existence of an edge which connects $v_2$ with $v_1$ \label{geometry}.

\begin{figure}[ptbh]
\centering\includegraphics[width=0.4\textwidth ]{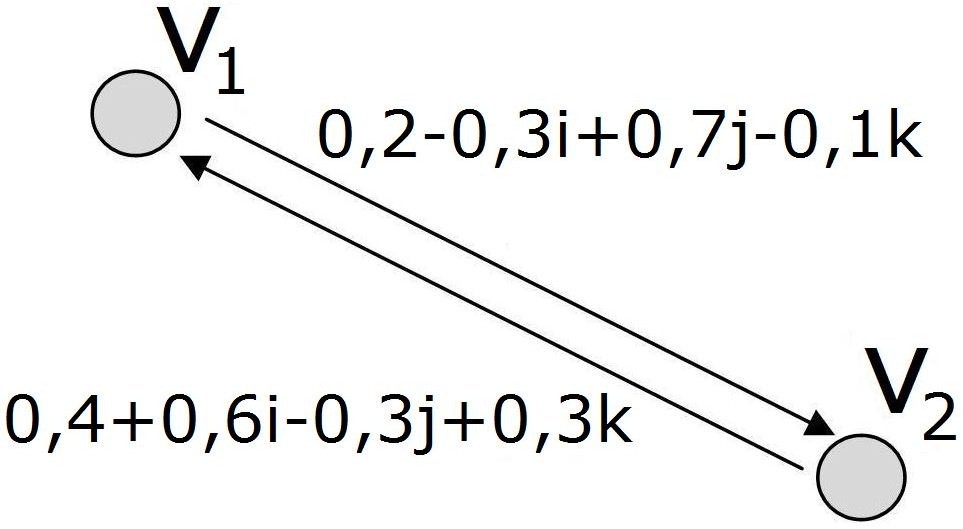}\end{figure}

A non-drawn arrow corresponds to an arrow with number $0$. In principle,
for every pair of vertices exists a couple of arrows which connect each other,
eventually with label $0$.

We can describe our universe by means of vertices connected by couple of arrows, with a
quaternionic number next to each arrow, as shown in figure \ref{eq: network}, below.

\begin{figure}[ptbh]
\centering\includegraphics[width=0.7\textwidth ]{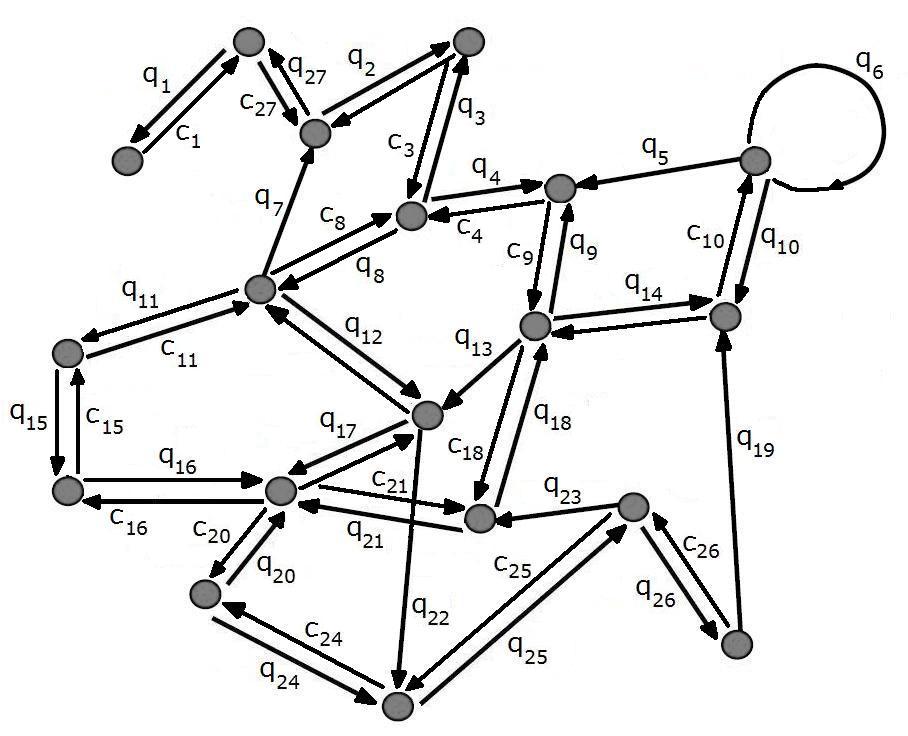}\caption{We can describe
our universe by means of vertices connected by couple of arrows, with a quaternionic
number next to each arrow.} \label{eq: network}
\end{figure}
What we are building is another variation of the Penrose's spin-network model \cite{spinnet}
or the Spin-Foam models \cite{spinfoam1}, \cite{spinfoam2} in Loop Quantum Gravity \cite{loopgravity},
which generalize Feynman diagrams.

\subsection{The Matrix relating couples of points \label{arrmatrix}}

Given a spin-network, like the one in figure \ref{eq: network}, we can move from picture to the
\lq\lq Arrangement Matrix'' $M$, which is a simple table constructed as follows. We enumerate all
the vertices in the graph at our will, provided we enumerate all of them.
Typically we think of indexing the vertices by the usual sequence of integers $1,2,3,4,5,\ldots.$

Thus we create such matrix, whose rows and columns are enumerated in the same way as the vertices in the graph.
Then we look at the vertices $v_{i}$ and $v_{j}$: in the entry $(i,j)$ we report the number situated near the
arrow which moves from $v_i$ to $v_j$. Similarly, in the entry $(j,i)$ we report the number written near the
opposite arrow. Remember that an absent arrow is an arrow with number $0$ and consider for the moment
$|M^{ij}|\leq 1$ for every $ij$.

\begin{figure}[h]
\centering\includegraphics[width=0.3\textwidth ]{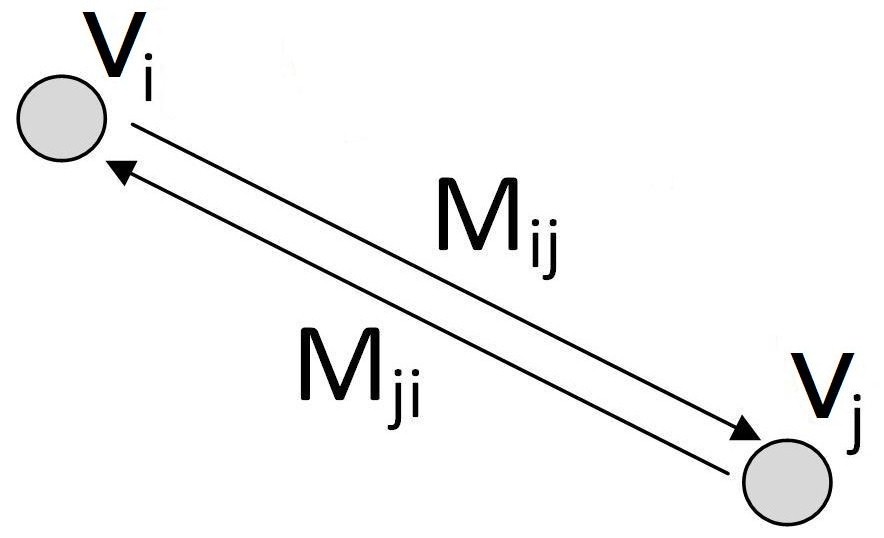}
\label{figuranuova}
\end{figure}

\noindent In principle, we can image an entry $M_{ij} \neq M_{ji}$, even with $|M_{ij}|^2 \neq |M_{ji}|^2$.
This means that $v_i$ may be connected to $v_j$ even if $v_j$ is not connected to $v_i$.
In that case, a non-commutative geometry is involved. The probability amplitude that $v_i$ and $v_j$
are mutually connected (we could talk about \lq\lq classical'' connection), is:

$$Cl.ampl. \propto M_{ij} M_{ji}$$
The probability amplitude for the vertex $v_i$ to be classically connected with any other
vertex (hence it will be not isolated) is:
$$Cl.ampl. \propto \sum_j M_{ij} M_{ji} = (M \cdot M)_{ii}$$
We can imagine our table with elements $M_{ij}$ as a machine which \lq\lq creates''
jointures between vertices, by connecting each other or closing a single vertex
onto itself through a loop. The loops are obviously represented by diagonal elements
of matrix, with the form $(i,i)$.

Now let's ask ourselves: is it necessary to know where the vertices are located?
Let's look at the Standard Model action: it is given by a sum (or more properly, an integral),
over $all$ the points of the universe, of locally defined terms. Any term is defined on a single point.
Since the terms are separated - a term for each point - and we integrate all of them,
we do not need to know where the points physically are.

However, there are terms which are not strictly local, ie those containing the derivative
operator $\partial$. The operator $\partial$, acting on a field $\varphi$\ in the point $v_{j}$,
calculates the difference between the value of $\varphi$ in a point immediately \lq\lq after''
$v_{j}$, and the value of $\varphi$ immediately \lq\lq\ before'' $v_{j}$.

In the discretized theory, the integral over points becomes a sum over
vertices of the graph. Similarly, the derivative becomes a finite difference.
Hence, for terms containing $\partial$, we need a clear definition of
\lq\lq before'' and \lq\lq after'', that is an arrangement of the vertices, as defined
by the matrix $M$.

We consider a scalar field but don't represent it with the usual function
(or distribution) $\varphi\left(  x\right) $. Instead we denote it
with a column of elements (an array) where each element
is the value of the field in a specific vertex of the graph. For example (with only 7 vertices):

\normalsize
\begin{spacing}{1.1}
\begin{equation}
\varphi=\left(
\begin{array}
[c]{c}%
\varphi\left(  p_{0}\right) \\
\varphi\left(  p_{1}\right) \\
\varphi\left(  p_{2}\right) \\
\varphi\left(  p_{3}\right) \\
\varphi\left(  p_{4}\right) \\
\varphi\left(  p_{5}\right) \\
\varphi\left(  p_{6}\right)
\end{array}
\right)
\end{equation}
\end{spacing}
\large

\begin{spacing}{1.5}
\noindent For simplicity, we start with a one-dimensional graph: it's easy to see how the
derivative operator is proportional to an antisymmetric matrix $\tilde{M}$ whose
elements are different from zero only immediately above the diagonal (where they count +1),
and immediately below (where they count -1). We can see this, for example, in a
\lq\lq toy-graph'' formed by only $12$ separated vertices (figure \ref{cerchio}).
The argument remains true while increasing the number of vertices.
\end{spacing}
\begin{spacing}{1.1}
\normalsize
$$
\!
\partial\varphi \! = \!\fr 1 {2\Delta} \!\left(
\begin{array}
[c]{cccccccccccc}%
0 & +1 & 0 & 0 & 0 & 0 & 0 & 0 & 0 & 0 & 0 & -1\\
-1 & 0 & +1 & 0 & 0 & 0 & 0 & 0 & 0 & 0 & 0 & 0\\
0 & -1 & 0 & +1 & 0 & 0 & 0 & 0 & 0 & 0 & 0 & 0\\
0 & 0 & -1 & 0 & +1 & 0 & 0 & 0 & 0 & 0 & 0 & 0\\
0 & 0 & 0 & -1 & 0 & +1 & 0 & 0 & 0 & 0 & 0 & 0\\
0 & 0 & 0 & 0 & -1 & 0 & +1 & 0 & 0 & 0 & 0 & 0\\
0 & 0 & 0 & 0 & 0 & -1 & 0 & +1 & 0 & 0 & 0 & 0\\
0 & 0 & 0 & 0 & 0 & 0 & -1 & 0 & +1 & 0 & 0 & 0\\
0 & 0 & 0 & 0 & 0 & 0 & 0 & -1 & 0 & +1 & 0 & 0\\
0 & 0 & 0 & 0 & 0 & 0 & 0 & 0 & -1 & 0 & +1 & 0\\
0 & 0 & 0 & 0 & 0 & 0 & 0 & 0 & 0 & -1 & 0 & +1\\
+1 & 0 & 0 & 0 & 0 & 0 & 0 & 0 & 0 & 0 & -1 & 0
\end{array}
\right) \!\!\! \left(
\begin{array}
[c]{c}%
\varphi\left(  0\right) \\
\varphi\left(  1\right) \\
\varphi\left(  2\right) \\
\varphi\left(  3\right) \\
\varphi\left(  4\right) \\
\varphi\left(  5\right) \\
\varphi\left(  6\right) \\
\varphi\left(  7\right) \\
\varphi\left(  8\right) \\
\varphi\left(  9\right) \\
\varphi\left(  10\right) \\
\varphi\left(  11\right)
\end{array}
\right) $$

\be = \left(
\begin{array}
[c]{c}%
\varphi\left(  1\right)  -\varphi\left(  11\right) \\
\varphi\left(  2\right)  -\varphi\left(  0\right) \\
\varphi\left(  3\right)  -\varphi\left(  1\right) \\
\varphi\left(  4\right)  -\varphi\left(  2\right) \\
\varphi\left(  5\right)  -\varphi\left(  3\right) \\
\varphi\left(  6\right)  -\varphi\left(  6\right) \\
\varphi\left(  7\right)  -\varphi\left(  5\right) \\
\varphi\left(  8\right)  -\varphi\left(  6\right) \\
\varphi\left(  9\right)  -\varphi\left(  7\right) \\
\varphi\left(  10\right)  -\varphi\left(  8\right) \\
\varphi\left(  11\right)  -\varphi\left(  9\right) \\
\varphi\left(  0\right)  -\varphi\left(  10\right)
\end{array}
\right)\label{eq: matrix_dev}
\ee
$$ \pt $$
\end{spacing}
\large
\begin{figure}[ptbh]
\centering\includegraphics[width=0.4\textwidth ]{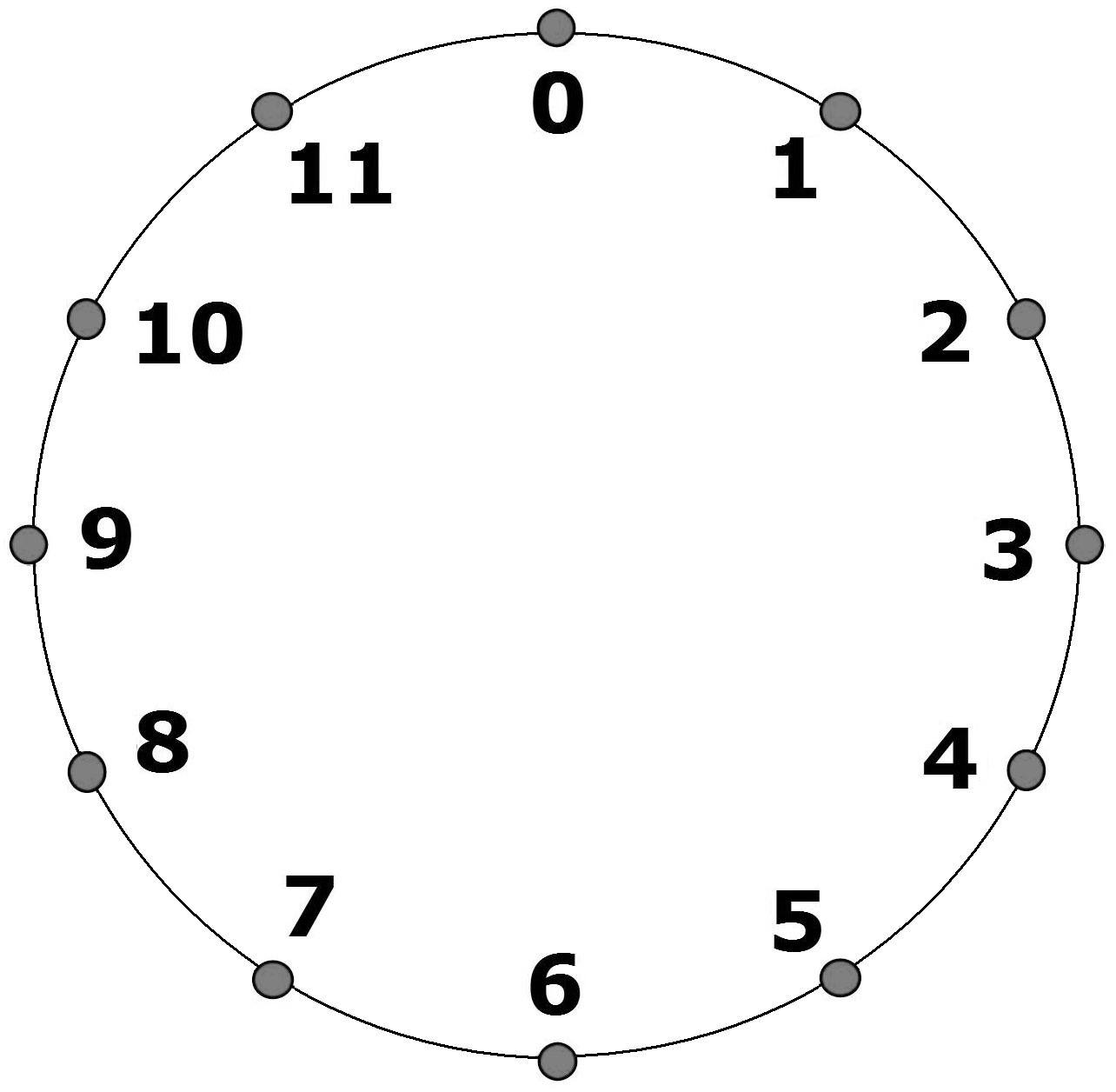}\caption{{A simple graph with $12$ vertices
which approximates a circular one-dimensional space.}}
\label{cerchio}
\end{figure}
%\end{spacing}
\noindent $\Delta$ is the length of graph edges. In the continuous limit, $\Delta \ra 0$
(that occurs in Hausdorff spaces, where matricial product turns into a convolution), we obtain

\ba
\pa \varphi (x) &=& \lim_{\Delta \ra 0}\fr 1 {2\Delta} \int \tilde{M}(x,y) \varphi (y) dy \nonumber \\
\pa \varphi (x) &=& \lim_{\Delta \ra 0}\fr 1 {2\Delta} \int \left[ \delta (y-(x+\Delta)) - \delta (y-(x-\Delta))\right] \varphi (y) dy \nonumber \\
\pa \varphi (x) &=& \lim_{\Delta \ra 0}\fr {\varphi(x+\Delta)-\varphi(x-\Delta)} {2\Delta} = \pa\varphi (x)
\ea
In this way our definition is consistent with the usual definition of derivative.

While increasing the number of points, a $(-1)$ still remains in the up right corner
of the matrix, and a $(+1)$ in the down left corner as well.
To remove those two non-null terms, it is sufficient to make them unnecessary,
by imposing boundary conditions that make the field null in the first and in the last point.

In fact we can describe an open universe (a straight line in one dimension),
starting from a closed universe (a circle) and making the radius to tend to infinity.
Hence we see that the conditions of null field in the first and in the last point become
the traditional boundary conditions for the Standard Model fields.

\begin{notation}
Note that in spaces with more than one dimension, a derivative matrix $\tilde{M}_\mu$
assumes the form (\ref{eq: matrix_dev}) only if we number the vertices progressively
along the coordinate $\mu$. However, two different numberings can be always related
by a vertices permutation.
\end{notation}

\section[A quaternionic field action in a non-ordered space]{A quaternionic field action in a non-ordered space-time \label{ord}}

\begin{definition}For any graph $\L$ we define its \emph{associated non-ordered space}
$\mathbf{S}_\Lambda$ as the ensemble of all its vertices. \end{definition} The graph
includes vertices plus edges (ordered connections between vertices), while the
\emph{associated non-ordered space} contains only vertices. In some sense, $\S$ doesn't
know where any vertex is.

Consider a \emph{numbering function} $\pi$, that is whatever bijection from
$X \subset\mathbf{N}$ to the non-ordered space.

\ba \pi &:& X\subset\mathbf{N} \longrightarrow \S \nonumber \\
        &&  i \longrightarrow v_i = \pi(i) \nonumber \ea

\noindent In this way, every vertex $v_i$ in $\S$ is one to one with an integer
$i \in X\subset\mathbf{N}$. This means that the ensemble of vertices has to be at
most numerable.

We consider a generic invertible matrix $M$ and interpret any entry $M^{ij}$ of $M$
as the probability amplitude for the existence in $\L$ of an edge which connects
$\pi(i)$ with $\pi(j)$. Remember that a couple of vertices can be connected by at most
two oriented edges with different orientations. $M^{ij}$ defines the probability amplitude
for the edge which moves from $\pi(i)$ to $\pi(j)$, while $M^{ji}$ defines the probability
amplitude for the edge which moves from $\pi(j)$ to $\pi(i)$.

Take care that in four dimensions we have to number the vertices by
elements $(i,j,k,l)$ in $\mathbf{N}^4$ before taking the limit $\Delta \ra 0$.
In this way $\sum_{(i,j,k,l)}\Delta^4$ becomes $\int dx^0 dx^1 dx^2 dx^3$.
If, as we have suggested, the vertices have been already numbered
with elements of $\mathbf{N}$, we can change the numbering by using
the natural bijection $\vartheta$ between $\mathbf{N}$ and $\mathbf{N}^4$, with
$(i,j,k,l) = \vartheta (a)$, $(i,j,k,l) \in \mathbf{N}^4$ and $a \in \mathbf{N}$.

\begin{definition}[covariant derivative] Given any skew hermitian matrix $A_\mu$,
with entries in $\mathbf{H}$, and a skew hermitian matrix $\tilde{M}_\mu$, which assumes
the form (\ref{eq: matrix_dev}) when the vertices are numbered along the coordinate $\mu$,
their associated covariant derivative is
\end{definition}
\begin{equation}
\nabla_{\mu}=\tilde{M}_{\mu}+ A_{\mu}.
\end{equation}

\begin{definition}[arrangement] We indicate with $n$ the number of elements inside $X \subset \mathbf{N}$.
Given a normal matrix $\hat{M}$ and four covariant derivatives $\nabla_\mu$ ($\mu = 0,1,2,3$) with dimensions $n \times n$,
an \emph{arrangement} for $\hat{M}$ is a quadruplet of couples $(\hat{D}^\mu,\hat{U})$, with $\hat{D}^\mu$
diagonal and $\hat{U}$ hyperunitary, such that

\begin{equation}
\hat{M} = \sum_\mu \hat{U}\hat{D}^\mu \nabla_\mu \hat{U}^{\dag}.
\end{equation}
\end{definition}
We require that covariant derivative will be form-invariant under the action of a transformation
$V \in U(n,\mathbf{H})$ which acts both on $\tilde{M}_\mu$ and $A_\mu$. We explicit $V\nabla_{\mu}V^{\dag}$:

\begin{align}
V\nabla_{\mu}V^{\dag}  &  = V\left(  \tilde{M}_{\mu}+A_{\mu}\right)  V^{\dag}\\
&  = \underset{=1}{\underbrace{VV^{\dag}}}\tilde{M}_{\mu}+V\left[  \tilde
{M}_{\mu},V^{\dag}\right]  +VA_{\mu}V^{\dag}\nonumber .
\end{align}
Setting

\begin{equation}
A_{\mu}^{\prime}= V\left[  \tilde{M}_{\mu},V^{\dag}\right]  +VA_{\mu}V^{\dag},
\end{equation}
we obtain

\begin{equation}
V\nabla_{\mu}V^{\dag}=\tilde{M}_{\mu}+A_{\mu}^{\prime}\overset{def}{=}%
\nabla_{\mu}^{\prime} \label{eq: der_A}
\end{equation}
that means

$$V\nabla_{\mu}[A]V^{\dag} = \nabla_\mu [A'].$$

\noindent Hence the transformation law for the matrix $A_{\mu}$ is like we expect:

\begin{equation}
A_{\mu}\rightarrow {A'}_\mu =
V\left[  \tilde{M}_{\mu},V^{\dag}\right]  +VA_{\mu}V^{\dag} . \label{eq: trasforma_A}
\end{equation}

\noindent We observe that (\ref{eq: trasforma_A}) preserves the hermiticity of $A_{\mu}$. In fact

\ba
{A'}^\dag_\mu &=& (V[\tilde{M}_\mu,V^\dag] + VA_\mu V^\dag)^\dag \nonumber \\
            &=& (V\tilde{M}_\mu V^\dag - \tilde{M}_\mu + VA_\mu V^\dag)^\dag \nonumber \\
            &=& V\tilde{M}_\mu^\dag V^\dag - \tilde{M}_\mu^\dag + V A_\mu^\dag V^\dag \nonumber \\
            &=& - V \tilde{M}_\mu V^\dag + \tilde{M}_\mu - V A_\mu V^\dag \nonumber \\
            &=& -(V[\tilde{M}_\mu,V^\dag] + VA_\mu V^\dag) = -{A'}_\mu \ea
It's easy to see that (\ref{eq: trasforma_A}) reduces to the usual transformation for a gauge field
${A'}_{\mu} = V\partial_{\mu}V^{\dag}+VA_{\mu}V^{\dag}$ in the limit $\Delta \ra 0$.

\begin{theorem}
For every invertible normal matrix $\hat{M}$ and every covariant derivative $\nabla[A]_\mu$ which is
invertible (in the matricial sense), there exist
\begin{enumerate}
\item A new quadruplet of covariant derivatives ${\nabla'}_\mu = \nabla[A']_\mu$ such that
$D^\mu {\nabla'}_\mu = 1$ for some diagonal matrix $D^\mu$, where $A^\prime_\mu$ is the gauge
transformed of $A_\mu$ for some unitary transformation $U$;
\item An arrangement $(\hat{D}^\mu,\hat{U})$ between $\hat{M}$ and $\nabla^\prime_\mu$.
\end{enumerate}\label{existence}\end{theorem}

\begin{proof}
According to spectral theorem, $\forall \hat{M} \in\mathbb{M}^{(N)}$ $\exists \hat{U}$ hyperunitary
such that $\hat{U}\hat{M}\hat{U}^{\dag}= K$ with $K$ diagonal. $\hat{M}$ is invertible, so the same is
true for $K$. Setting $\hat{D}= K^{-1}$:

\begin{align}
\hat{U}\hat{M}\hat{U}^{\dag}\hat{D} = K\hat{D} = KK^{-1} = 1\label{ordinamento}\\
\hat{D}\hat{U}\hat{M}\hat{U}^{\dag}= \hat{D}K = K^{-1}K = 1 .\nonumber
\end{align}

At this point we choice a covariant derivative $\nabla_\mu$ (which is also a normal matrix) and
we reason as we did above for $\hat{M}$, putting

\begin{equation}
1 =D^{\mu}U\nabla_{\mu}U^{\dag}=U\nabla_{\mu}U^{\dag}D^{\mu}
\label{eq: riduzioneB}
\end{equation}
for some $D^\mu$ diagonal and $U$ unitary.
No sum over repeated indices is implied.

A well known theorem states that $U$ can be chosen in such a way that $D^\mu$ takes values in $\mathbf{C}$.
Moreover we can always find a quaternion $s$ with $|s|=1$ such that, if $D^\mu$ takes values
in $\mathbf{C} = \mathbf{R} \oplus i\mathbf{R}$, then $s^* D^\mu s$ will take values
in $\mathbf{C} = \mathbf{R} \oplus (ri+tj+pk)\mathbf{R}$, with fixed $r,t,p \in \mathbf{R}$ and $r^2+t^2+p^2 =1$.
Every $s$ with $|s|=1$ describes in fact a rotation in the $3$ dimensional space with base elements $i,j,k$.

Introducing such $s$, the equation (\ref{eq: riduzioneB}) becomes

\be
s1s = s^*D^{\mu}ss^{*}U\nabla_{\mu} U^{\dag}s.
\ee
Now we note that $s^*U$ is another hyperunitary transformation. Redefining $s^* D_\mu s \ra D_\mu$, $s^* U \ra U$
we obtain newly

\be
1 = D^{\mu}U\nabla_{\mu} U^{\dag}.
\ee
In this way we can always choose in what complex plane is $D_\mu$. In the following we call this
propriety \lq\lq $s$-invariance''. Using (\ref{eq: der_A}) into (\ref{eq: riduzioneB}):

\begin{equation}
1=D^\mu \nabla_{\mu}^{\prime}=\nabla_{\mu}^{\prime}D^\mu \Longrightarrow\left[
\nabla_{\mu}^{\prime},D^\mu \right]  = 0 .
\end{equation}
Taking into account (\ref{ordinamento}):

\begin{align}
\hat{D}\hat{U}\hat{M} \hat{U}^{\dag}  &  =D^\mu \nabla_{\mu}^{\prime}\\
\hat{U}\hat{M} \hat{U}^{\dag}\hat{D}  &  =\nabla_{\mu}^{\prime}D^\mu \nonumber .
\end{align}
Summing on $\mu$ we obtain:

\begin{align}
4\hat{D}\hat{U}\hat{M} \hat{U}^{\dag}  &  = \sum_\mu D^\mu \nabla_{\mu}^{\prime}\\
4\hat{U}\hat{M} \hat{U}^{\dag}\hat{D}  &  = \sum_\mu \nabla_{\mu}^{\prime}D^\mu . \nonumber
\end{align}

\noindent Solving for $\hat{M}$:

\be
\hat{M}  = \fr 14 \sum_\mu \hat{U}^{\dag} \hat{D}^{-1} D^{\mu} \nabla_{\mu}^{\prime} \hat{U}
   = \fr 14 \sum_\mu \hat{U}^{\dag} \nabla_{\mu}^{\prime} D^{\mu}\hat{D}^{-1} \hat{U}.
\ee
Defining $\hat{D}^\mu$ as $\fr 14 \hat{D}^{-1} D^\mu $

\be
\hat{M}  = \sum_\mu \hat{U}^{\dag} \hat{D}^{\mu} \nabla_{\mu}^{\prime} \hat{U} \label{finale-dim}
\ee
QED \end{proof}

Note that in general $\hat{M} \neq \sum_\mu \hat{U}^{\dag} \nabla_{\mu}^{\prime} \hat{D}^{\mu}\hat{U}$
because $\hat{D}^{-1} D^\mu \neq D^\mu \hat{D}^{-1}$ for the non commutativity of quaternions.

\begin{theorem} For every invertible matrix $M$ with entries in $\mathbf{H}$, a normal matrix $\hat{M} = U_M M$ exists,
where $U_M$ is unitary and $\hat{M}$ is neither hermitian nor skew hermitian. \label{normal}
\end{theorem}

\begin{proof} Given an invertible matrix $M$, a unique choice of matrices $U$ and $P$ always exists,
with $U$ unitary and $P$ hermitian positive, such that $UM = P$. Moreover, a well known theorem states
that, for every hermitian matrix $P$ with entries in $\mathbf{H}$, there exist $I,J,K$ skew hermitian
unitary matrices which commute with $P$. Moreover $I,J,K$ achieve the same algebra of quaternionic
imaginary unities $i,j,k$.
\begin{spacing}{1.5}
Consider then the unitary matrix $p = exp((bI + cJ + dK)P)$, with $b,c,d \in \mathbf{R}$.
It's easy to see that $[p,P]=0$. Moreover the matrix $\hat{M} = pP$ is normal and it is neither
hermitian or skew hermitian. In fact
\end{spacing}
\normalsize
$$(pP)^\dag = p^{\dag} P = p^{-1} P = \neq \pm pP$$
$$(pP)(pP)^\dag = (Pp)(Pp)^\dag = Ppp^\dag P^\dag = PP = Pp^\dag pP = P^\dag p^\dag pP = (pP)^\dag (pP)$$
\large

\noindent Moreover

$$\hat{M} = pUM = U_M M\qquad U_M = pU\,\,\, unitary.$$
\end{proof}

\begin{definition}[associated normal matrix] $\pt\!\!\!\!$ For every invertible matrix $M$, we define an \emph{associated normal matrix} as a normal matrix obtained trough the construction above. We indicate it with $\hat{M}$ and use the notation $U_M$ for the unitary transformation which transforms
$M$ in $\hat{M} = U_M M$. \end{definition}

\begin{theorem}
For every $n \times n$ invertible matrix $M$ with entries in $\mathbf{H}$ and every quadruplet of covariant derivatives $\nabla[A]_\mu$ which are invertible (in the matricial sense), there exist
\begin{enumerate}
\item An associated normal matrix $\hat{M} = U_M M$ with $U_M$ unitary;
\item A new quadruplet of covariant derivatives ${\nabla'}_\mu = \nabla[A']_\mu$ such that $D^\mu {\nabla'}_\mu = 1$ for some diagonal matrix $D^\mu$, where $A^\prime_\mu$ is the gauge transformed of $A_\mu$ for some unitary transformation $U$;
\item An arrangement $(\hat{D}_\mu,\hat{U})$ between $\hat{M}$ and $\nabla^\prime_\mu$ such that
\end{enumerate}

\ba
&& S = (M\phi)^\dag \cdot(M\phi) =\sum_{i=1}^n \sum_{\mu,\nu} \sqrt{\left\vert h\right\vert }
h^{\mu\nu}(x_i) ({\nabla'}_{\mu}\phi^{\prime}(x_i))^* ({\nabla'}_{\nu}\phi^{\prime}(x_i))  .\nonumber \\
&& \pt \label{azione-sca}
\ea

\noindent Here $\phi$ is a one-component quaternionic field, while

\ba
&& x_i \equiv \pi(i) \nonumber \\
&& \phi^{\prime}(x_{i})= {\phi^{\prime}}^{i}(x)= \sum_{j} \hat{U}^{ij}\phi^{j}(x) = \sum_{j} \hat{U}^{ij}\phi(x_{j})\nonumber \\
&& \sqrt h h^{\mu\nu}(x_i) = \fr 12 d^\mu d^{*\nu}(x_i) +c.c.\qquad \hat{D}_\mu^{ij} = d^\mu(x_i)\d_{ij}. \nonumber \\ && \label{transf-scalar}\ea
\end{theorem}

\begin{proof} The existence of $\nabla^\prime_\mu = \nabla[A']_\mu$ follows from the proof of theorem \ref{existence}, while the existence of an associated normal matrix $\hat M = U_M M$ descends from theorem \ref{normal}. Hence we see that the first action in (\ref{azione-sca}) is invariant for transformations $(U_1,U_2)$ in $U(n,\mathbf{H}) \otimes U(n,\mathbf{H})$ which send $M$ in $U_2 M U_1^\dag$ and $\phi$ in $U_1\phi$. In fact
\begin{spacing}{1.3}
\ba S[\phi] &=& \phi^\dag M^\dag M \phi \nonumber \\
        &\ra& \phi^\dag U_1^\dag (U_2 M U_1^\dag)^\dag (U_2 M U_1^\dag) U_1 \phi \nonumber \\
        &=& \phi^\dag U_1^\dag U_1 M^\dag U_2^\dag U_2 M U_1^\dag U_1 \phi \nonumber \\
        &=& \phi^{\dag} M^\dag M \phi = S[\phi]\ea
If we set $U_1 = 1$ and $U_2 = U_M$ we have

\be S[\phi] \ra \phi^\dag M^\dag U_M^\dag U_M M \phi = \left\{ \begin{array}[l]{ll} = \phi^\dag M^\dag M \phi = S[\phi] \\
    = \phi^\dag \hat{M}^\dag \hat{M} \phi \qquad .\end{array} \right. \label{sca-nuova}  \ee
We substitute (\ref{finale-dim}) in (\ref{sca-nuova}) with $\hat{M}$ in place of $M$.
\end{spacing}
\ba
    S[\phi] &=& {\displaystyle\sum\limits_{\mu,\nu}}
    \left(  \hat{U}^{\dag}\hat{D}^{\mu}\nabla_{\mu}^{\prime}\hat{U}\phi\right)^\dag  \left(
    \hat{U}^{\dag}\hat{D}^{\nu}\nabla_{\nu}^{\prime}\hat{U}\phi\right) \nonumber \\
   &=& {\displaystyle\sum\limits_{\mu,\nu}}
    \left(  \phi^{\dag}\hat{U}^{\dag}\nabla_{\mu}^{\prime\dag}\hat{D}^{\mu\dag}\hat{U}
    \hat{U}^{\dag}\hat{D}^{\nu}\nabla_{\nu}^{\prime}\hat{U}\phi\right) \nonumber\\
   &=& {\displaystyle\sum\limits_{\mu,\nu}}
    \left( \phi^{\dag}\hat{U}^{\dag}\nabla_{\mu}^{\prime\dag}\hat{D}^{\mu\dag}
    \underset{=1}{\underbrace{\hat{U}\hat{U}^{\dag}}}\hat{D}^{\nu}\nabla_{\nu}^{\prime}
    \hat{U}\phi\right) \nonumber\\
   &=& {\displaystyle\sum\limits_{\mu,\nu}}
    \left(  \phi^{\dag}\hat{U}^{\dag}\nabla_{\mu}^{\prime\dag}\hat{D}^{\mu\dag}
    \hat{D}^{\nu}\nabla_{\nu}^{\prime}\hat{U}\phi\right) \nonumber\\
   &=& {\displaystyle\sum\limits_{\mu,\nu}}
    \left( \phi^{\prime\dag}\nabla_{\mu}^{\prime\dag} \hat{D}^{\mu\dag}\hat{D}^{\nu}
    \nabla_{\nu}^{\prime}\phi^{\prime}\right) \nonumber .
\ea
In the last step we have taken in account the definition
(\ref{transf-scalar}). Finally

\be
S  = \fr 12 \sum_{\mu,\nu} \phi^{\prime^\dag} \nabla_{\nu}^{\prime^\dag}
\left( \hat{D}^{\mu\dag}\hat{D}^{\nu} +c.c. \right) \nabla_{\mu}^{\prime}\phi^{\prime}.
\ee
It is remarkable that $\hat{D}_{\mu}$ is diagonal:
\begin{equation}
\hat{D}^{\mu}_{ij}=d^{\mu}\left(  x_{i}\right)  \delta_{ij} .
\end{equation}
We can set

\begin{equation}
\sqrt{\left\vert h\right\vert }h^{\mu\nu}\left(  x_{i}\right)  = \fr 12 d^{\mu *}d^{\nu}\left(  x_{i}\right)  +c.c.
\end{equation}
and then

\begin{equation}
S= {\displaystyle\sum\limits_{i,\mu,\nu}}
\sqrt{\left\vert h\right\vert }h^{\mu\nu}(x_{i})\left(  \nabla_{\mu}^{\prime
}\phi^{\prime}\right)^{*i} \left(  \nabla_{\nu}^{\prime}
\phi^{\prime}\right)^{i} . \label{ultima2}
\end{equation}
\end{proof}
QED.

The action of a transformation $(U_1, U_2)$ on $\nabla'$ follows from its action on $M$.
We can always use the invariance under $U(n,\mathbf{H}) \otimes U(n,\mathbf{H})$ to put $M$
in the form $M =\sum_\mu \hat{D}^\mu \nabla'_\mu$. Starting from this we have

\be U_2 M U_1^\dag = \sum_\mu U_2 \hat{D}^\mu \nabla'_\mu U_1^\dag = \sum_\mu U_2 \hat{D}^\mu U_1^\dag U_1\nabla'_\mu U_1^\dag.\label{transf}\ee
We define ${\nabla''}_\mu = U_1\nabla'_\mu U_1^\dag$ the transformed of $\nabla'$ under $(U_1, U_2)$ and
$\hat{D}^{\prime\mu} = U_2 \hat{D}^\mu U_1^\dag$ the transformed of $\hat{D}^\mu$. We assume that $A^\prime_\mu$
inside $\nabla^\prime_\mu$ transforms correctly as a gauge field, so that

$$\nabla^\prime [A^\prime]_\mu \phi' = \nabla^\prime [A^\prime]_\mu U_1^\dag \phi'' = U_1^\dag \nabla'' [A^\prime]_\mu \phi'' = U_1^\dag \nabla^\prime [A^\prime_{U1}]_\mu \phi''$$
$$\phi'' = U_1 \phi' .$$

\noindent We want $\hat{D}^{\prime\mu}$ remain diagonal and $h' = h[\hat{D}'] = h[\hat{D}]$.
In this case there are two relevant possibilities:
\begin{enumerate}
\item $\hat{D}$ is a matrix made by blocks $m \times m$ with $m$ integer divisor of $n$
and every block proportional to identity. In this case the residual symmetry is $U(1,\mathbf{H})^n \times U(m,\mathbf{H})^{n/m}$ with
elements $(sV, V)$, $s$ both diagonal and unitary, $V \in U(m,\mathbf{H})^{n/m}$;
\item $h$ is any diagonal matrix. The symmetry reduces to $U(1,\mathbf{H})^n \otimes U(1,\mathbf{H})^n$
which is local $U(1,\mathbf{H}) \otimes U(1,\mathbf{H}) \sim SU(2) \otimes SU(2) \sim SO(4)$.
\end{enumerate}

\noindent In this way, if we keep fixed the metric $h$ and keep diagonal $\hat{D}$, the new action will
be invariant at least under $U(1,\mathbf{H})^n \otimes U(1,\mathbf{H})^n$ which doesn't modify $h$.

Note however that the action (\ref{ultima2}) is highly non local, because the fields $A_\mu(x^a, x^b)$ with
$a \neq b$ can relate couples of vertices very far each other. In fact the transformations
in $U(n,\mathbf{H})$ mix all the vertices in the universe independently from their
position. In the next section we'll discover in what limit (besides $\Delta \ra 0$) the
(\ref{ultima2}) becomes a local action. Let us now pause on the metric $h^{\mu\nu}$.

\begin{notation}
We observe how the metric $h$ has appeared from nowhere.
We get the \lq\lq impression'' that the metric does not exist \lq\lq a priori'',
but is generated by the matrices $\hat{D}$. In other words:
the metric is simply the result of our desire to see an ordered universe at any cost.
\end{notation}

\begin{notation}
Note that we have chosen the matrix $\nabla$ between skew hermitian matrices, so that the gauge fields
$AR_i$ have real eigenvalues, corresponding to effectively measurable quantities\footnote{The operator
$R_i$ acts on any array $\psi$ as $R_i \psi = \psi i$.}.
Conver\-se\-ly, $\hat{M}$ must remain generically normal.
In fact, if $\hat{M}$ was (skew) hermitian, the fields $d$ would become (imaginary) real, and there would
not be enough degrees of freedom to construct the metric $h$.
\end{notation}

\noindent We focus on the relationship:

\begin{equation}
\sqrt{\left\vert h\right\vert }h^{\mu\nu}\left(  x_{i}\right)  = \fr 12 d^{\mu *}d^{\nu}\left(  x_{i}\right)  +c.c. \label{eq: metrica_ord}
\end{equation}

\noindent We set:

\be
d =\left(
\begin{array}
[c]{c}%
a_{0}+ib_{0}+jc_0+kd_0 \\
a_{1}+ib_{1}+jc_1+kd_1 \\
a_{2}+ib_{2}+jc_2+kd_2 \\
a_{3}+ib_{3}+jc_3+kd_3
\end{array}
\right)
\ee

It's easy to see how $s$-invariance permits us to choose the $D_\mu$ in such a way that the real
vectors
\begin{spacing}{1.3}
\be \left(
\begin{array}
[c]{c}%
a_{0}\\
a_{1}\\
a_{2}\\
a_{3}
\end{array}
\right), \quad
\left(
\begin{array}
[c]{c}%
b_{0}\\
b_{1}\\
b_{2}\\
b_{3}
\end{array}
\right), \quad
\left(
\begin{array}
[c]{c}%
c_0 \\
c_1 \\
c_2 \\
c_3
\end{array}
\right), \quad
\left(
\begin{array}
[c]{c}%
d_0 \\
d_1 \\
d_2 \\
d_3
\end{array}
\right)\ee
\end{spacing}
\newpage
\noindent will be linearly independent.

\ba
\sqrt{\left\vert h\right\vert }h^{-1}  &=& \left(
\begin{array}
[c]{cc}
a_{0}^{2}+b_{0}^{2}+c_0^2+d_0^2 & a_{0}a_{1}+b_{0}b_{1}+c_0 c_1+d_0 d_1 \\
a_{1}a_{0}+b_{1}b_{0}+c_1 c_0+d_1 d_0 & a_1^2+b_1^2+c_1^2+d_1^2 \\
a_2 a_0+b_2 b_0+c_2 c_0 +d_2 d_0 & a_2 a_1+b_2 b_1 +c_2 c_1+d_2 d_1 \\
a_3 a_0+b_3 b_0+c_3 c_0+d_3 d_0 & a_3 a_1+b_3 b_1+c_3 c_1+d_3 d_1
\end{array}
\right. \nonumber \\
&& \left. \begin{array}
[c]{cc}
a_{0}a_{2}+b_{0}b_{2}+c_0 c_2+d_0 d_2 & a_{0}a_{3}+b_{0}b_{3}+c_0 c_3+d_0 d_3 \\
a_1 a_2 + b_1 b_2+c_1 c_2+d_1 d_2 & a_1 a_3 +b_1 b_3+c_1 c_3+d_1 d_3 \\
a_2^2+b_2^2+c_2^2+d_2^2 & a_2 a_3 + b_2 b_3 + c_{2}c_{3}+d_{2}d_{3} \\
a_3 a_2+b_3 b_2 +c_3 c_2+d_3 d_2 & a_3^2+b_3^2+c_3^2+d_3^2
\end{array}\right) \nonumber
\ea

Note that we have $10$ independent metric components as it should be.
What would have happened if the entries of $M$ were been simply complex numbers?

In that case we could always take a one-form $X_\nu$ such that $X_\nu (Im$ $d^\nu) = X_\nu (Re\,d^\nu) = 0$.
The contraction of $X_\nu$ with the metric would be

$$\sqrt h h^{\mu\nu} X_\nu = d^{*\mu} (d^\nu X_\nu) + d^\mu (d^{*\nu} X_\nu) = 0 .$$

Hence the metric would be degenerate. For $d^\mu \in \mathbf{H}$ this can't happen,
because no one-form can be orthogonal to $4$ vectors linearly independent in a $4$-dimensional space.
Moreover a such one-form exists in spaces with dimension $>4$. For this reason our theory hasn't
meaning in presence of extra dimensions.

\section{A local action from the quaternionic field action}
\label{local}
Here we expose how to get a local action from the quaternionic field action in the limit of low energy.
We can add to action quadratic $\sim M^2$ and quartic $\sim M^4$ terms, provided they are gauge invariant.
In general we obtain a non-trivial potential of form $\alpha M^4 - \beta M^2$.
We suppose that a minimum for such potential breaks the symmetry $U(n,\mathbf{H}) \otimes U(n,\mathbf{H})$
and provides a mass to gauge fields $A_\mu$.
To view it is sufficient to rewrite $M$ as a function of $A_\mu$ and consider a quartic term:

\begin{equation}
h^{\mu\alpha}A_{\mu}A_{\alpha}h^{\nu\beta}A_{\nu}A_{\beta} .
\label{eq: quartico}
\end{equation}

For a minimum of $M$ there is a minimum of $A$ which gives sense to the expansion:

\begin{equation}
A_{\mu}=A_{\mu}^{\min}+\delta A_{\mu} .
\end{equation}
Therefore the (\ref{eq: quartico}) generates a factor:

\begin{equation}
m\left(  x\right)  ^{2}h^{\nu\beta}A_{\nu}A_{\beta}
\end{equation}

\begin{equation}
m\left(  x\right)  ^{2} = h^{\mu\alpha}A_{\mu}^{\min}A_{\alpha}^{\min}
\end{equation}
Hence the gauge fields acquire a mass, varying from point to point
in the universe and essentially dependent on the metric.

\begin{theorem}
Given a potential for $M$, which is both hermitian and invariant for $U(n,\mathbf{H})\otimes U(n,\mathbf{H})$,
his minimum configurations are always invariant at least for $U(1,\mathbf{H})^n \otimes U(1,\mathbf{H})^n$,
that is a local $U(1,\mathbf{H}) \otimes U(1,\mathbf{H})$.
\end{theorem}

\begin{proof}
A such potential contains only terms of type $tr((MM^\dag)^j)$, $j \in \mathbf{N}$.
All we can measure are eigenvalues of hermitian operators, and a hermitian operator has
only real eigenvalues $q$ which are invariant under $U(1,\mathbf{H})^n$, ie $sqs^* = qss^* =q$ for $|s| =1$.
The simpler hermitian operators made by $M$ are $MM^\dag$ and $M^\dag M$, whose eigenvalues are invariant under

$$M \ra s_1 M s_2^*\qquad (s_1, s_2) \in U(1,\mathbf{H}) \otimes U(1,\mathbf{H})$$
$$MM^\dag \ra s_2 M s_1^* s_1 M^\dag s_2^* = s_2 MM^\dag s_2^*$$
$$M^\dag M \ra s_1 M^\dag s_2^* s_2 M s_1^* = s_1 M^\dag M s_1^*$$
\end{proof}

\noindent In this manner we have always $m =0$ for diagonal fields
$A_\mu (x^a, x^a) \overset{!}{=} A_\mu (x^a)$.

A transformation $(s_1,s_2) \in U(1, \mathbf{H})\otimes U(1, \mathbf{H})$ acts inside action in the expected way
(see formula (\ref{transf}))

$$ \phi' \ra s_1 \phi' \overset{!}{=} \phi'' $$
$$ {\nabla'}[A']_\mu \ra s_1 {\nabla'}[A']_\mu s_1^* = {\nabla'}[A^\prime_{s1}]_\mu $$
$$ d^\mu \ra s_2 d^\mu s_1^* $$

\be S[\phi', A'] = S'[\phi'', A^\prime_{s1}] = \sum_{\mu\nu}(s_2 d^\mu s_1^* \nabla^\prime [A^\prime_{s1}]_\mu \phi'')^\dag (s_2 d^\nu s_1^*
 \nabla^\prime [A^\prime_{s1}]_\nu \phi'') \label{locale1}\ee

\noindent We use the natural correspondence

\be (1, i, j, k) \longleftrightarrow i(\s^0, \s^1, \s^2, \s^3), \qquad \s^0 = -i\mathbf{1}, \label{correspondence}\ee
and define the complex field $\hat{\phi}$ as a complex $2 \times 2$ matrix,

$$\hat{\phi}^a = \left( \begin{array}[c]{cc}
\phi_1^a +i\phi_2^a & \phi_3^a+i\phi_4^a  \\
-\phi_3^a +i\phi_4^a & \phi_1^a -i\phi_2^a \\
\end{array} \right)$$
with $\phi'' = \phi_1 + i \phi_2 +j\phi_3 +k \phi_4$ and $\phi_1, \phi_2, \phi_3, \phi_4 \in \mathbf{R}$.
Every term between parenthesis becomes

\be W_2 (i\s^k) d_k^\mu W_1^\dag {\nabla'} [A^\prime_{s1}]_\mu \hat{\phi}  \label{SU2}\ee
where $\s$ are Pauli matrices and $(W_1,W_2) \in SU(2)\otimes SU(2)$.

\begin{theorem}
For every $SO(4)$ transformation $\Lambda$, a transformation $(W_1,W_2) \in SU(2) \otimes SU(2)$
exists, such that for every vector $d_j \in R^4$ we find

$$ \Lambda_i^{\pt j} d_j \sigma^i = d_i W_2 \sigma^i W_1^\dag .$$
\end{theorem}

\begin{proof}
We write $W_1 = U^{\prime\dag}_1 U_1$ and $W_2 = U^\prime_1 U_1$. In this manner we decompose $SU(2)\otimes SU(2)$ in $SU(2)_{rot} \otimes SU(2)_{boosts}$. $SU(2)_{rot}$ is generated by the couples $(U_1,U_1)$, while $SU(2)_{boosts}$ by the couples $(U^{\prime\dag}_1, U^{\prime}_1)$. After a wick rotation, the first one describes rotation in $R^3$, while the second one describes boosts.

A generic vector $d = \left( \begin{array}[c]{cccc} d_0 & d_1 & d_2 & d_3 \end{array} \right)$ gives

$$d_i (i\sigma^i) = \left( \begin{array}[c]{cc} d_0 + id_3 & id_1+d_2 \\ id_1-d_2 & d_0-id_3
\end{array} \right)$$
with $|d|^2 = det\,d_i (i\sigma^i)$. A transformation in $SO(4)$ doesn't change the norm $|d|$. Moreover, for every
$d$ exists a transformation in $SO(4)$ which put it in the normal form

$$d = \left( \begin{array}[c]{cccc} |d| & 0 & 0 & 0 \end{array} \right) .$$
The same properties have to be true for $SU(2) \otimes SU(2)$.
The first one is banally verified because $det\,W_1 = det\,W_2 = 1$ and then $det\, d_i(i\sigma^i) = det\,d_i(W_2 i\sigma^i W_1^\dag)$.
Being $d_i(i\sigma^i)$ normal, we can use a transformation in $SU(2)_{rot}$ to put it in a diagonal form

$$U_1 d_i (i\sigma^i) U_1^\dag = \left( \begin{array}[c]{cc} d_0 + id_3 & 0 \\ 0 & d_0-id_3
\end{array} \right) .$$
Define now the matrix $U^\prime_1$ as

$$U^\prime_1 = \fr 1 {\sqrt{|d|}}\left( \begin{array}[c]{cc} \sqrt{d_0 + id_3} & 0 \\ 0 & \sqrt{d_0-id_3}
\end{array} \right) .$$
It's easy to verify that $U^\prime_1 U^{\prime\dag}_1 = 1$ and $det\,U^\prime_1 = 1$.
Applying to $U_1 d_i (i\sigma^i) U_1^\dag$ this transformation in $SU(2)_{boosts}$ we obtain

$$U_1^\prime U_1 d_i (i\sigma^i) U_1^\dag U_1^\prime = \left( \begin{array}[c]{cc} |d| & 0 \\ 0 & |d|
\end{array} \right) .$$
So, for every $d$, a transformation in $SU(2) \otimes SU(2)$  exists, which puts it in the normal form.
In this way, $d$ transforms exactly as a vielbein field in the Palatini formulation of General Relativity,
giving then the correspondence

\ba \Lambda_i^{\pt j} d_j (i\sigma^i) &=& U^\prime_1 U_1 (i\sigma^j) U_1^\dag U^\prime_1 d_j \nonumber \\
    \Lambda_i^{\pt j} d_j \sigma^i &=& W_2 \sigma^j W_1^\dag d_j \nonumber \\
    \fr 12 tr (\Lambda_i^{\pt j} d_j \sigma^i \sigma^k) &=& \fr 12 tr (W_2 \sigma^j W_1^\dag \sigma^k) d_j \nonumber \\
    \Lambda_k^{\pt j} d_j &=& \fr 12 tr (W_2 \sigma^j W_1^\dag \sigma^k) d_j \nonumber \\
    \Lambda_k^{\pt j} &=& \fr 12 tr (W_2 \sigma^j W_1^\dag \sigma^k) .\ea
So, at every $\Lambda \in SO(4)$ corresponds a couple $(U_1, U_2) \in SU(2) \otimes SU(2)$.
\end{proof}

Applying this to (\ref{SU2}), it becomes

\be W_2 (i\s^k) d_k^\mu W_1^\dag {\nabla'} [A^\prime_{s1}]_\mu \hat{\phi} =  \Lambda_k^{\pt i} d_i^\mu \s^k {\nabla'}[A^\prime_{s1}]_\mu \hat{\phi} \label{SU2-second} .\ee

Note that if we write $\hat{\phi} = (\hat{\phi}_1\,\,\,\hat{\phi}_2)$, with $\hat{\phi}_1, \hat{\phi}_2$
complex column arrays $1 \times 2$, then $\hat{\phi}_2 = i\s_2 \hat{\phi}_1^*$. This implies that
the column array $1 \times 4$ $\left( \begin{array}[c]{c}
\hat{\phi}_1 \\
\hat{\phi}_2 \\
\end{array} \right)$ transforms under $SO(4)$ as a Majorana spinor.

Applying newly the correspondence (\ref{correspondence}) to (\ref{SU2-second}), we obtain

$$ s_2 d^\mu s_1^* \nabla^\prime [A^\prime_{s1}]_\mu \phi'' = \Lambda d^\mu {\nabla'} [A^\prime_{s1}]_\mu \phi''.$$
Inserting it in the action (\ref{locale1})

\ba S'[A^\prime_{s1},\phi''] &=& \sum_{\mu\nu}({\nabla'} [A^\prime_{s1}]_\nu \phi'')^\dag d^{*\nu} \Lambda^\dag \Lambda\, d^\mu
({\nabla'} [A^\prime_{s1}]_\mu \phi'') \nonumber \\
      &=& \sum_{\mu\nu}({\nabla'} [A^\prime_{s1}]_\nu \phi'')^\dag d^{*\nu} d^\mu ({\nabla'} [A^\prime_{s1}]_\mu \phi'') \nonumber \\
      &=& \sum_{\mu\nu}(d^{\nu}{\nabla'} [A^\prime_{s1}]_\nu \phi'')^\dag (d^\mu{\nabla'} [A^\prime_{s1}]_\mu \phi'') \nonumber \\
      &=& S[A^\prime_{s1},\phi''] .\ea

The diagonal gauge field $A(x^a)$ compensates the action of $SU(2)\otimes SU(2)$
inside $\nabla'$. Moreover we have just demonstrated that the field $d^\mu$ transforms under this group
as a vielbein field in the Palatini formulation of General Relativity. This implies $A(x^a)$ is a
gravitational spin-connection. Consequently, every purely imaginary quaternion defines a spin operator
$\vec{S}$ via the correspondence $(i,j,k) \leftrightarrow 2i(S_1, S_2,$ $S_3)$. In fact, each element in
$U(1,\mathbf{H})$ is the exponential of a purely imaginary quaternion, in the same  way as an element in
$SU(2)$ is the exponential of $i\vec{\a} \cdot \vec{S}$ for some real vector $\vec{\a}$.

Note that a majorana spinor in an euclidean space can't distinguish if $s_2$ belongs to
$SU(2)_{rot}$, $SU(2)_{boosts}$ or if it is a mixed combination. Only after the wick rotation
it feels a difference, because the generator of $SU(2)_{boosts}$ moves from $i\sigma^i$ to $\sigma^i$,
while $SU(2)_{rot}$ remains unchanged.

Someone can infer that, if $\phi$ transforms as a majorana spinor, our action
has not the standard form. We don't care this now: what exposed is only a toy model. In another
work (under review) we show explicitly how to get the correct Dirac action for these and all the
others fields (both fermions and bosons).

To finish, we suppose that masses of other fields ($A(x^a, x^b)$ with $a \neq b$) are sufficiently large,
so that the experimental physics of nowadays is unable to locate them.
For the same reason, in the low energy approximation, they can be omitted from the action.
Neglecting the \lq\lq ultra-massive'' fields, the scalar field action becomes a local action

\begin{equation}
S=\sum_{i=1}^n \sum_{\mu,\nu}
\sqrt{\left\vert h\right\vert }h^{\mu\nu}\left(  x_i\right)  \left(
\overset{G}{\nabla}_{\mu}\phi(x_i)\right)^*  \left(  \overset{G}{\nabla}_{\nu}\phi(x_i)\right)
\end{equation}
where $\overset{G}{\nabla}$ are standard gravitational covariant derivatives.

\section{The origin of spin}
\label{spin}
Consider the spin operator $S_{3}$

\begin{equation}
\hat{S}_{3} = \frac{\hslash}2\left(
\begin{array}
[c]{cc}%
1 & 0\\
0 & -1
\end{array}
\right)
\end{equation}
and calculate the normalized eigenvectors and eigenstates.

\begin{align}
\left\vert \uparrow\right\rangle  &  =e^{i\phi}\left(
\begin{array}
[c]{c}%
1\\
0
\end{array}
\right)  \text{, \ with eigenvalue }\lambda_{1}=+\frac{1}{2}\text{ \ (in unit  }%
\hslash=1\text{)}\label{eq: autketS3}\\
\left\vert \downarrow\right\rangle  &  =e^{i\phi}\left(
\begin{array}
[c]{c}%
0\\
1
\end{array}
\right)  \text{, \ with eigenvalue }\lambda_{2}=-\frac{1}{2}\text{ \ }
\end{align}
where $\phi$ is an arbitrary phase. The eigenvectors completeness guarantees that the field $\hat{\phi}_1$,
which appears in the precedent section, can be always decomposed in a sum of such eigenstates.

The projectors on a single eigenstate of $S_{3}$ are

\begin{equation}
\hat{\pi}^{+} =\frac{1}{2}\left(
\begin{array}
[c]{cc}
1 & 0 \\
0 & 0
\end{array}
\right)  ,
\end{equation}

\begin{equation}
\hat{\pi}^{-} =\frac{1}{2}\left(
\begin{array}
[c]{cc}
0 & 0 \\
0 & 1
\end{array}
\right) .
\end{equation}

We see that $\hat{\pi}^{\pm}$ are idempotent, while $\hat{\pi}^{+}\hat{\pi
}^{-}=0$, as it should be. A rotation by an angle $\theta$
around the axe $1$ is represented by the unitary matrix:

\begin{equation}
U_{1}\left(  \theta\right)  =\left(
\begin{array}
[c]{cc}
\cos(\theta/2) & -i\sin(\theta/2)  \\
-i\sin(\theta/2) & \cos(\theta/2)
\end{array}
\right)
\end{equation}

\noindent where

$$U_1(\theta)\hat{\phi}_1 = \widehat{(s(U_1)\phi)}_1\qquad \hat{\phi}_1=\widehat{(\phi)}_1$$

\noindent for some quaternion $s(U)$ with $|s|=1$. In the special case of a rotation by $\pi$:

\begin{equation}
U_{1}\left(  \pi\right)  =\left(
\begin{array}
[c]{cc}
0 & -i \\
-i & 0
\end{array}
\right)  \label{eq: rotazione}.
\end{equation}

We suppose now that the system is in the eigenstate $\left\vert
\uparrow\right\rangle $; following a rotation around the axis $1$ the
state will be:

$$
\left\vert \uparrow\right\rangle _{R}= U_{1}(\theta)\left\vert \uparrow
\right\rangle .
$$
For $\theta=\pi$:

\begin{equation}
\left\vert \uparrow\right\rangle _{R}= U_{1}\left(  \pi\right)  \left\vert
\uparrow\right\rangle = -i\left\vert \downarrow\right\rangle =e^{-i\pi
/2}\left\vert \downarrow\right\rangle \rightarrow\left\vert \downarrow
\right\rangle \text{,} \label{eq: scambio_spin}%
\end{equation}
\begin{equation}
\left\vert \downarrow\right\rangle _{R}= U_{1}\left(  \pi\right)  \left\vert
\downarrow\right\rangle = -i\left\vert \uparrow\right\rangle =e^{-i\pi
/2}\left\vert \uparrow\right\rangle \rightarrow\left\vert \uparrow
\right\rangle \text{,} \label{eq: scambio_spin2}%
\end{equation}
since the state is defined up to an inessential phase factor.
We observe that a rotation by $\pi$ around the axe $1$ is equivalent to
exchange $\left\vert \uparrow\right\rangle $ with $\left\vert
\downarrow\right\rangle $, as we have just verified by
(\ref{eq: scambio_spin}) and (\ref{eq: scambio_spin2}).

Surely we can expand the matrix $M$ as follows

\begin{equation}
M\left(  x^{a},x^{b}\right)  = M'\left(  x^{a},x^{b}\right) + |s(x^a)|e^{r(x^a)} \d^{ab}
\end{equation}

\noindent with $M'\left(  x^{a},x^{b}\right) = 0$ for $a=b$.

The element $r(x^a) = arg[s(x^a)]$ is a purely imaginary quaternion: when it acts on $\phi$, it determines
uniquely the result of a spin measure, exchanging the states $\left\vert \uparrow\right\rangle $ - $\left
\vert\downarrow\right\rangle $. This seems to suggest an identification between the arrangement field $M$
and the observer who performs the measurement.

Indeed the operator $M$ can simulate a measurement operation when it presents
the form $M^{ab} = u^a w^b$:

\begin{eqnarray}
 M^{ab} &=& u^a w^b \overset{continuous}{\longrightarrow} M(x,y) = \psi(x) \psi^{\ast}(y) \nonumber \\
 M^{ab}\varphi_b &=& u^a (w^b \varphi_b) \overset{continuous}{\longrightarrow} \int dy \, M(x,y) \varphi(y) \nonumber \\
        &=& \psi(x)\int dy \, \psi^{\ast}(y)\varphi(y) = \psi(x)(\psi,\varphi) \nonumber
\end{eqnarray}
$\psi\left(  x\right)$ is any eigenstate, while $\left(  \psi
,\varphi\right)$ denotes the scalar product between $\psi$ and $\varphi$. We see that $M$ projects
$\varphi$ along the eigenstate $\psi$, and in quantum mechanics a measurement is just a projection.

The latter argument gives also an indication about the spin nature. Consider the entries of $M$
closest to the diagonal: they are the $M^{ij+1}$ and $M^{ij-1}$ which compose $\tilde{M}$.
Moreover, they represent the probability amplitudes for the existence of connections between (numerically) consecutive vertices.
In the limit $\Delta \ra 0$, $\tilde{M}$ becomes $\pa$, which is proportional to $i\pa$, an operator which
acts on a wave function $\psi(x)$ and returns the momentum $p$ of the corresponding particle:

$$i\pa \psi(x) = p\psi(x) .$$

In this way, the entries of $\tilde{M}$ represent both a momentum and a probability amplitude for
connections between (numerically) consecutive vertices. In a certain sense, $\tilde{M}$ draws
continuous paths and measures the momentums along these paths (figure \ref{percorso_continuo}).

If we describe a particle with a wave function $\phi$, its spin is determined by diagonal
components of $M$: in fact, $exp(r)$ acts on $\phi$ as a rotation in the tangent space.
Consequently, if $r$ is applied to $\phi$, it returns the spin of the associated particle.

The diagonal components of $M$ represent also the probability amplitudes for a connection between a vertex and itself.
Reasoning in analogy with the components of $\tilde{M}$, we associate at every such \lq\lq pointwise'' loop a
circumference $S^1$: we interpret the spin as the rotational momentum due to the motion along these
circumferences (figure \ref{Loop_1}).

\begin{figure}[ptbh]
\centering\includegraphics[width=0.6\textwidth ]{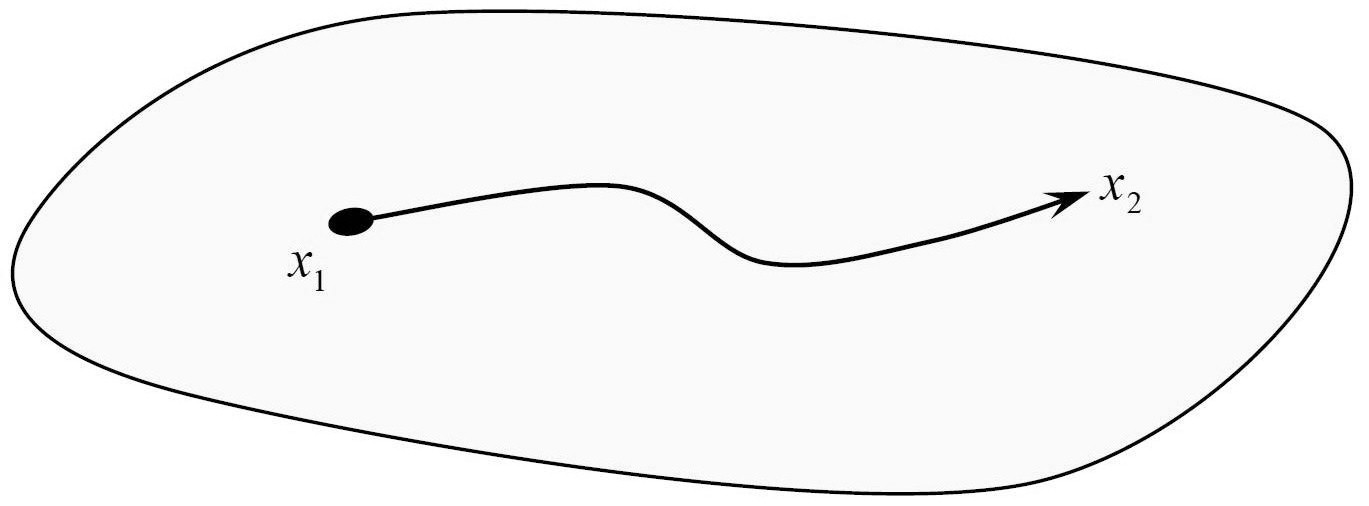}\caption{$\tilde{M}$ behaves as a derivative, that is proportional to a momentum operator.
The non-empty entries of $\tilde{M}$ represent both a momentum and a probability amplitude for
connections between (numerically) consecutive vertices.
In a certain sense, $\tilde{M}$ draws continuous paths and measures the momentums along them.}
\label{percorso_continuo}
\end{figure}
\begin{figure}[ptbh]
\centering\includegraphics[width=0.6\textwidth ]{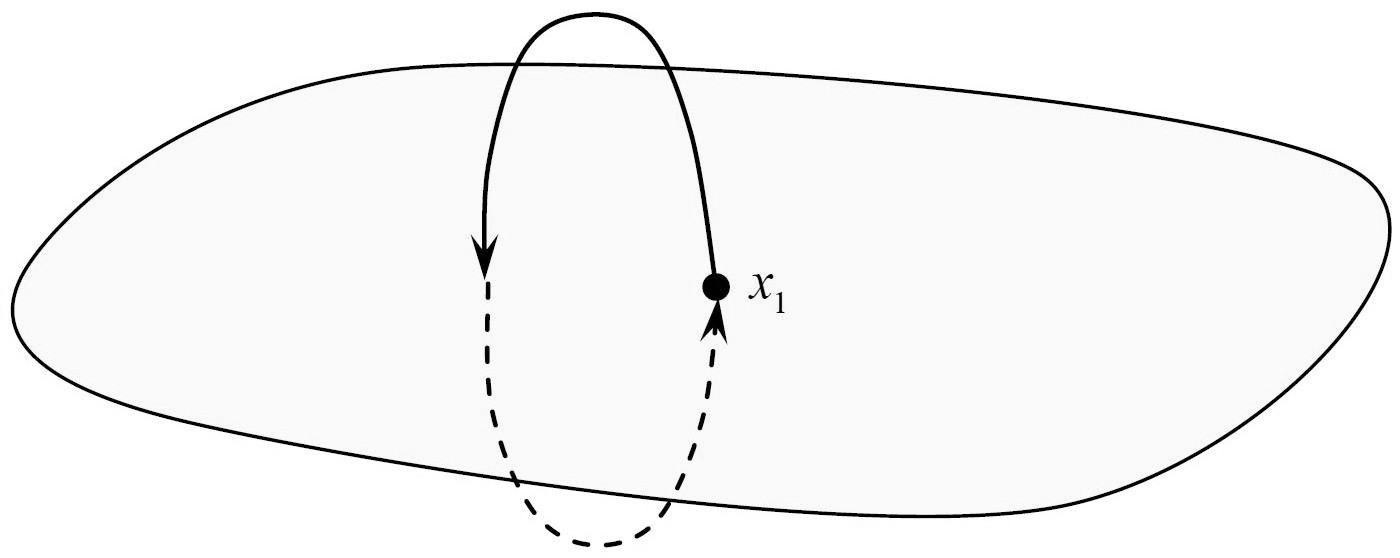}\caption{Each diagonal component of $M$ represents the probability amplitude for a connection between a vertex and itself.
The spin is a momentum along such pointwise loops.}
\label{Loop_1}
\end{figure}

It is remarkable that there exist two types of pointwise loops: the one in figure \ref{Loop_1}, where a particle assumes the same aspect after a complete rotation,
and the one in figure \ref{Loop_2}, where a particle assumes the same aspect after two complete rotations. The first case suggests a relationship with gauge fields of spin $1$,
the second with fermionic fields of spin $1/2$.

\begin{figure}[ptbh]
\centering\includegraphics[width=0.6\textwidth ]{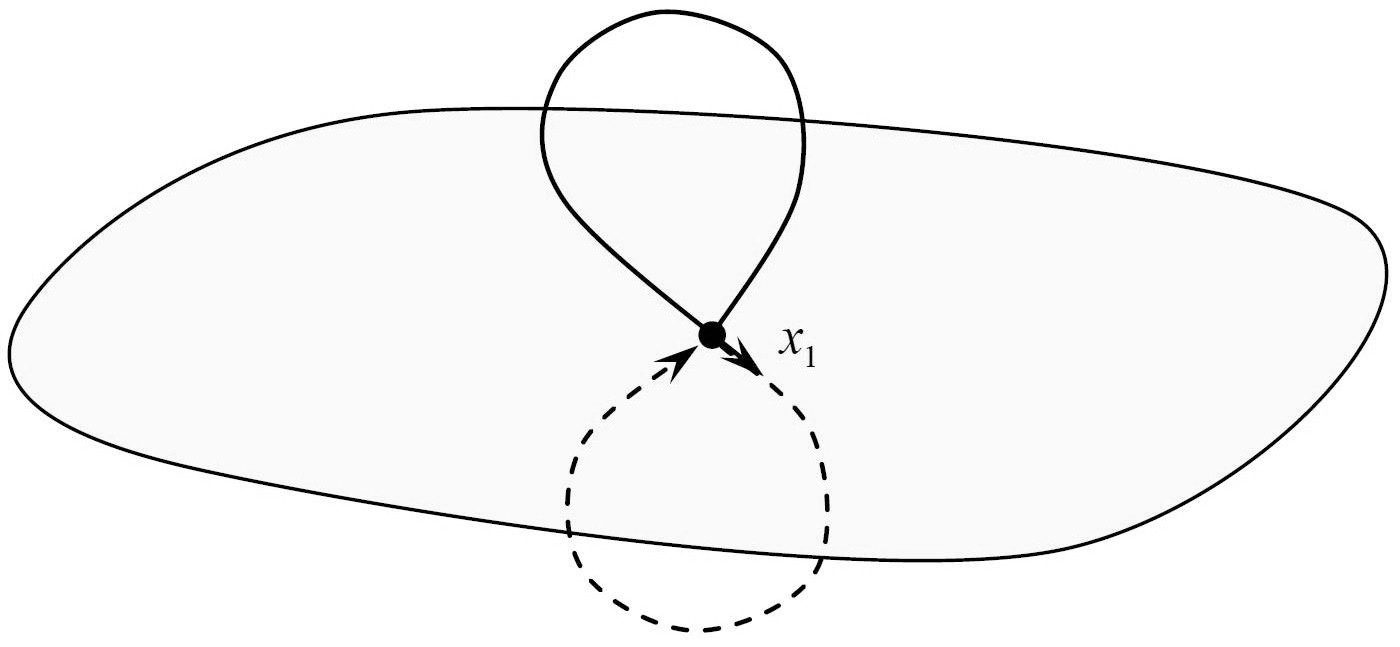}\caption{Pointwise loop associable with fermionic field.}
\label{Loop_2}
\end{figure}

\section{Symmetry breaking}
\label{symmetry}
We imagine that the symmetry breaking of $U(n,\mathbf{H})\otimes U(n,\mathbf{H})$ is not complete, but
a residual symmetry remains for transformations in $U(1,\mathbf{H})^n \times U(m,\mathbf{H})^{n/m}$.
Here $m$ is an integer divisor of $n$.

In this case, it is possible to regroup the $n$ points into $n/m$ ensembles $\mathcal{U}^a$, with $a = 1, 2, \ldots, n/m$.

$$ \mathcal{U}^a = \mathcal{U}^a (x^a_1, x^a_2, \ldots, x^a_m)$$

$$ \varphi = (\varphi (x^a_i)) = \left( \begin{array}[c]{ccccc}
\varphi(x^1_1) & \varphi(x^1_2) & \varphi(x^1_3) & \ldots & \varphi(x^1_m) \\
\varphi(x^2_1) & \varphi(x^2_2) & \varphi(x^2_3) & \ldots & \varphi(x^2_m) \\
\varphi(x^3_1) & \varphi(x^3_2) & \varphi(x^3_3) & \ldots & \varphi(x^3_m) \\
\vdots & \vdots & \vdots & \vdots & \vdots\\
\varphi(x^{n/m}_1) & \varphi(x^{n/m}_2) & \varphi(x^{n/m}_3) & \ldots & \varphi(x^{n/m}_4)
\end{array} \right) $$

$$ A = (A^{ab}_{ij}) = (A(x^a_i, x^b_j)) .$$

Now the indices $a,b$ of $A$ act on the columns of $\varphi$, while the indices $i,j$ act on the rows.
The fields $A^{ab}_{ij}$ with $a=b$ maintain null masses and so they continue to behave as gauge fields for $U(m,\mathbf{H})^{n/m}$.
Every $U(m,\mathbf{H})$ term in $U(m,\mathbf{H})^{n/m}$ acts independently inside a single $\mathcal{U}^a$. So, if we consider
the ensembles $\mathcal{U}^a$ as the real physical points, we can interpret $U(m,\mathbf{H})^{n/m}$ as a local $U(m,\mathbf{H})$.
It's simple to verify:

\ba
h^{\mu\nu}(x^a_i) &=& h^{\mu\nu}(x^a_j) \qquad \forall x_i, x_j \in \mathcal{U}^a \nonumber \\
h(x^a) &\overset{!}{=}& h(\mathcal{U}^a) = h^{\mu\nu}(x^a_i) \qquad \forall x^a_i \in \mathcal{U}^a \nonumber \\
A(x^a_i, x^a_j) &=& Tr\,\left[ A(x^a) T^{(ij)} \right], \quad where \nonumber \\
A(x^a) &=& \sum_{ij} A(x^a_i, x^a_j) T^{(ij)},
\ea
with $T^{(ij)}$ generator of $U(m,\mathbf{H})$.
Using these relations, in the next work we'll show how the terms $tr\,(MM^\dag)$ and $tr\,(MM^\dag MM^\dag)$
ge\-ne\-ra\-te respectively the Ricci scalar and the kinetic term for gauge fields.
Extending $M$ to grassmanian elements we have (up to a generalized $U(n, \mathbf{H})$ transformation)

$$M = \theta(\pa + \psi) + d^\mu (\pa_\mu + A_\mu)$$
$$M^\dag = (\pa^\dag +\psi^\dag)\theta^\dag + (\pa_\mu^\dag + A^\dag_\mu)d^{*\mu} .$$

$\theta, \theta^\dag$ are at the same time grassmanian coordinates and grassmanian equivalents of $d, d^*$.
$\pa, \pa^\dag$ are grassmanian derivatives and $\psi, \psi^\dag$ grassmanian fields (ie fermions).

Our final action will be

$$S = tr\,\left(\fr {MM^\dag}{16\pi G}-\fr 1 4 MM^\dag MM^\dag \right)$$
This action resembles the action of a $\lambda \phi^4$ theory. Some preliminary results
suggest that we can treat it by means of Feynman graphs, apparently without renormalization
problems.

We will see how the quartic term includes automatically the kinetic terms for gauge fields of
$SO(4) \otimes SU(3)\otimes SU(2) \otimes U(1)$ and the dirac action for exactly three fermionic families.

\section{Second quantization and black hole entropy}
\label{entropy}
It is remarkable that in our model the gauge fields and the gravitational fields have different
origins, although they are both born from $M$. The gravitational field in fact appears as
a multiplicative factor for moving from $M$ to the covariant derivative $\nabla'$.
The gauge fields are instead some additive elements in $\nabla'$.
This could be the reason for which the gravitational field seems non quantizable in the standard way.
On the other side, quantizing the gauge fields is equivalent to quantize a partial piece of $M$ in a flat space.
But a similar equivalence does not exist for the gravitational field.
In our framework this doesn't create problems, since we will quantize $M$ directly,
rather than gravitational and gauge fields.

What does it mean \lq\lq to quantize'' $M$? It's true that a matrice $M$
is a quantum object from its birth, as they are quantum objects the
wave functions which describe particles.

However, we will impose commutation relations on $M$,
in the same way we impose commutation relations on the wave functions.
This is the so called \lq\lq second quantization''.

The wave functions, which first had described the probability amplitude
to find a particle, then have become operators which create or annihilate
particle. Similarly, $M$ describes first the probability amplitude for
the existence of connections between vertices. After the second quantization
it will become an operator which creates or annihilate connections. In particular,
the operator $M(x^a,x^b)$ creates a connection between the vertices $x^a$ and $x^b$.

$M$ corresponds to $D^\mu\nabla^\prime_\mu$ (by invariance respect $U(n,\mathbf{H})$):
so it contains the various fields $A_\mu$ and $h^{\mu\nu}$. If we second
quantize $M$, then, indirectly, we quantize the other fields, including
the gravitational field.

To quantize $M$ we put $[M^{ij}, M^{kl\dag}]=\d^{ik}\d^{jl}$.
Here the symbol $\dag$ indicates the adjoint operator respect only
the scalar product between states in the Fock space.
The condition $[M^{ij}, M^{kl\dag}]=\d^{ik}\d^{jl}$ means that every entry
$M^{ij}$ expands in a sum of $4$ operators

$$M^{ij} = a + i(b_1 +b_2+ b_3)\qquad\quad b_1^\dag = b_1,\,\,\,b_2^\dag = b_2,\,\,\,b_3^\dag = b_3$$
The $b$'s realize the $SU(2)$ algebra implicit in the imaginary part of quaternions.
$$[b_1,b_2]=b_3 ;\qquad [b_2,b_3] = b_1;\qquad [b_3,b_1]=b_2$$
The operators $a^\dag$ and $b^\dag = b_1 + ib_2$ create an edge which connects the vertex $i$ with the vertex $j$.
The number operator is

$$N^{ij} = M^{ij\dag} M^{ij} = a^\dag a + |\vec{b}|^2\qquad \text{no sum on} \pt ij$$

\noindent $a^\dag a$ has eigenvalues $q \in \mathbf{N}$ with multiplicity $1$. Moreover the
eigenvalues of $|\vec{b}|^2$ are in the form $j(j+1)$ for $j \in \mathbf{N}/2$, with
multiplicity $(2j+1)$.
How about $N > 1$? We can consider a surface immersed into the graph.
Its area is $\Delta^2$ times the number of edges which pass through it.
If we admit the possibility for the creation of many superimposed edges,
we can interpret this superimposition as a \lq\lq super-edge'' which carries
an area equal to $N \Delta^2$.

\begin{notation}
Regarding diagonal components, we suggest a sli\-gh\-tly different interpretation:
$a^\dag$ could create loops, while $b^\dag$ could create perturbations which
travel through the loops (ie particles with spin $j$). This suggest a duality
between a loop on vertex $v_i$ and a closed string (as intended in \emph{String Theory})
situated approximately on the same vertex. Note that the two interpretations
can be accommodated if we consider quanta of area as non-local perturbations.
\end{notation}

The only Black Hole information detectable from the exterior, is the information coded in the Horizon.
So, the only distinguishable states of a Black Hole are distinguishable states of its horizon.
For the Black Hole horizon we consider all the edges which pass through it, oriented only from the interior to the exterior.

If the horizon is crossed only by edges with $N = q+j(j+1)$ and $a^\dag a =q$, the number of its distinguishable states is

$$num_S = \left( 2j+1 \right)^{A /(q+j(j+1))} .$$

We suppose now a generic partition with $A = \sum_{j,q} A_{j,q}$, where an area $A_{j,q}$ is
crossed only by edges with $N=q+j(j+1)$ and $a^\dag a =q$. The number of distinguishable states becomes

$$num_S = \sum_{\{A_{j,q}\}}\prod_{j,q} \left( 2j+1 \right)^{A_{j,q} /((q+j(j+1))\Delta^2)}$$

where the sum is over all the possible partitions of $A$.
The \lq\lq classical'' contribution comes from $j=0$ and gives $num_S =1$ (We call it \lq\lq classical''
because it is the only one with $N =1$). This implies no entropy and is related to the fact that
$tr \,M_H^\dag M_H \sim \int_H \sqrt {h_H} R(h_H) = 2\pi \chi_H$, where $M_H$ is the restriction of $M$
to the edges which cross the horizon, $h_H$ is the induced metric
on the horizon and $\chi$ is the Euler characteristic.

The dominant contribution comes from
$q=0$ and $j = 1/2$, which gives

$$num_S = 2^{4A_{1/2,0} /(3\Delta^2)}$$

So we can define entropy as

$$S = k_B\,log\,2^{4A/{3\Delta^2}} = \fr {4 \,log\,2 \, k_B A }{3\Delta^2}.$$

Our approach gives thus a proposal for the explanation of area law.
Indeed our entropy formula corresponds to the one given by Bekenstein and Hawking if $3\Delta^2 = 16 G\,log\,2$.

What is our interpretation of black hole radiation? The proximity between vertices is probabilistic:
we can have a high probability of receiving two vertices as \lq\lq neighbors'', but never a certainty. We look at
a large number of vertices for a long time: some vertex, which first seems to be adjacent to some other,
suddenly can appears far away. For this reason, some internal vertices in a Black Hole may happen to be found
outside, so that the Black Hole slowly evaporates.

We can consider also the contribution from ($q = j = 0$). If it exists, clearly it is the dominant one.
Indeed, an horizon means absence of connections between the exterior and the interior. For an external
observer, the universe finishes with the horizon. In fact, respect the coordinate system of a statical
observer infinitely distant from the horizon, every object, falling in the black hole, sits on the horizon
for an infinite time. In relation to the proper time of the statical far away observer, the object never
surpasses the horizon. If nothing surpasses the horizon, this means that the Hawking radiation comes
from the deposit of all the objects fallen in the black hole, ie from the horizon. This resolves the
information paradox proposed by Hawking.

Someone can infer that absence of connections is only illusory, because the horizon singularity is
of the type called \lq\lq apparent'': it doesn't exist in several coordinate systems, as the system
comoving with a free falling object.

We reply that it's true, because also the absence of connections depends strictly from the state on which
the number operator acts. Every state can be associated to a particular coordinates system and, if we
change coordinate system, we have to change the state. In this way, the connections can exist for an observer
and not exist for some others.

It's the same which happens for the particles. The same particle can exist in a coordinate system and not exists in an
another system (see Unruh effect). This is because the same number operator acts on different states.

Calculate now $num_S$ for $q = 0, j \ra 0$. It is

\ba num_S &=& \lim_{j\ra 0} \left( 2j+1 \right)^{A /(j(j+1)\Delta^2)} \nonumber \\
          &=& \lim_{j\ra 0} \left( 1+2j \right)^{A /(j\Delta^2)} \nonumber \\
          &=& \lim_{j\ra 0} \left( 1+2j \right)^{2A /(2j\Delta^2)} \nonumber \\
          &=& \lim_{x\ra \infty} \left( 1+\fr 1x \right)^{2Ax/\Delta^2} \nonumber \\
          &=& e^{2A/\Delta^2} \ea

The entropy becomes

$$ S =k_B \,log\,e^{2A/\Delta^2} = \fr {2k_B A}{\Delta^2}$$

This corresponds to the Bekenstein-Hawking result for $\Delta^2 = 8G$.

\section{Conclusion}

In this paper we have abandoned the preconceived existence of an order in the space-time structure,
taking a probabilistic approach also to its topology and its homology.

This framework gives new suggestions about the origin of space-time metric and particles spin.
At the same time it hints a possible emersion of all fields from an unique entity, ie the arrangement
matrix, after the imposition of an order.

Unfortunately, there isn't space here to post an explicit calculation of terms $tr\,(M^\dag M)$ and $tr\,(M^\dag M M^\dag M)$.
We have already said that they generate the Ricci scalar, the kinetic terms for gauge fields and the Dirac actions for
exactly three fermionic families.

In a next future we'll show how several phenomena can find a possible explanation inside this paradigm, as we
have seen earlier for black hole entropy. These deal with the galaxy rotation curves, the inflation, the
quantum entanglement, the values of matrices CKM and PMNS and the value of Newton constant $G$.

Here we have given a simple example by using a one-component field. Nevertheless, a potential for $M$
causes a symmetry breaking which gives mass to gauge fields without need of Higgs mechanism.
In the end, the one-component field action results unnecessary.

\section*{Acknowledgements}
I thank professor Valter Moretti, Dr. Fabrizio Coppola and Dr. Marcello
Colozzo for the useful discussions and suggestions.

\chapter{The arrangement field theory (AFT). Part 2}

\section{Introduction}

The arrangement field paradigm describes universe by means of a graph, ie an ensemble
of vertices and edges. However there is a considerable difference between this framework
and the usual modeling with spin-foams or spin-networks. The existence of an edge which
connects two vertices is in fact probabilistic. In this framework the fundamental quantity
is an invertible matrix $M$ with dimension $n \times n$, where $n$ is the number of vertices.
In the entry $ij$ of such matrix we have a quaternionic number which gives the probability
amplitude for the existence of an edge connecting vertex $i$ to vertex $j$.
In the introductory work \cite{Arrangement} we have developed a simple scalar field theory
in this probabilistic graph (we call it \lq\lq non-ordered space''). We have seen that a
space-time metric emerges spontaneously when we fix an ensemble of edges. Moreover, the
quantization of metric descends naturally from quantization of $M$ in the non-ordered space.
In section \ref{formalism} we summarize these results.

In section \ref{ricciscalar} we express Ricci scalar as a simple quadratic function of $M$.
We discover how the gravitational field emerges from diagonal components of $M$, in contrast to
gauge fields which come out from non-diagonal components.

In section \ref{kinetic} we define a quartic function of $M$ which develops a Gauss Bonnet
term for gravity and the usual kinetic term for gauge fields.

In section \ref{string} we discover a triality between \emph{Arrangement Field Theory},
\emph{String Theory} and \emph{Loop Quantum Gravity} which appear as different manifestations
of the same theory.

In section \ref{electroweak} we show that a grassmanian extension of $M$ generates automatically
all known fermionic fields, divided exactly in three families. We see how gravitational field
exchanges homologous particles in different families. The resulting scheme finds an analogue
in supersymmetric theories, with known fermionic fields which take the role of gauginos for known
bosons.

In the subsequent sections we explore some practical implications of arrangement field theory,
in connection to inflation, dark matter and quantum entanglement. Moreover we explain how
deal with theory perturbatively by means of Feynman diagrams.

\begin{center}
We warmly invite the reader to see introductory work \cite{Arrangement} before proceeding.
\end{center}

\section{Formalism}
\label{formalism}

In paper \cite{Arrangement} we have considered an euclidean $4$-dimensional space represented
by a graph with $n$ vertices. In this section we retrace the fundamental results of that work,
moving to Lorentzian spaces in the next section. Since now we assume the Einstein convention,
summing over repeated indices.

In proof of \textbf{theorem 8} in \cite{Arrangement} we have demonstrated the equivalence between the following actions:

\be S_1 = (M \varphi)^\dag (M \varphi) \label{iniziale}\ee
\be S_2 = \sum_{i=1}^n \sqrt{|h|} h^{\mu\nu} (x^i) (\nabla_\mu \varphi^i)^*(\nabla_\nu \varphi^i).
\label{action-1}\ee

\noindent $M$ is any invertible matrix $n \times n$ while the field $\varphi$ is represented
by a column array $1\times n$, with an entry for every vertex in the graph:

\be \varphi = \left( \begin{array}[c]{c}
\varphi(x^1) \\ \varphi(x^2) \\ \varphi(x^3) \\ \vdots \\ \varphi(x^N) \end{array} \right) .\ee

\noindent The entries of both $M$ and $\phi$ take values in the division ring of quaternions,
usually indicated with $\mathbf{H}$. The first action considers the universe as an abstract
ensemble of vertices, numbered from $1$ to $n$, where $n$ is the total number of space vertices.
The entry $(ij)$ in the matrix $M$ represents the probability amplitude for the existence of an
edge which connects the vertex number $i$ to the vertex number $j$. We admit non-commutative
geometries, which in this framework implies a possible inequivalence $|M^{ij}| \neq |M^{ji}|$.
More, the first action is invariant under transformations $(U_1,U_2) \in U(n,\mathbf{H})
\otimes U(n,\mathbf{H})$ which send $M$ in $U_2 M U_1^\dag$.

In action (\ref{action-1}) a covariant derivative for $U(n,\mathbf{H}) \otimes U(n,\mathbf{H})$
appears, represented by a skew hermitian matrix $\nabla$ which expands according to $\nabla_\mu =
\tilde{M}_\mu + A_\mu$. Here $\tilde{M}_\mu$ is a linear operator such that $lim_{\Delta \ra 0}
\tilde{M}_\mu = \pa_\mu$, where $\Delta$ is the graph step. If we number the space vertices along direction $\mu$, $\tilde{M}_\mu$
becomes

\be \tilde{M}^{ij}_\mu = \fr 1 {2\Delta} \left[ \delta^{(i+1)j} - \delta^{(i-1)j} \right]\label{dderiv}\ee
$$ \sum_j \tilde{M}^{ij} \varphi^j = \fr 1 {2\Delta} \sum_j \delta^{(i+1)j} \varphi^j - \delta^{(i-1)j} \varphi^j = \fr {\varphi(i+1) - \varphi(i-1)} {2\Delta} .$$
The gauge fields $A$ act as skew hermitian matrices too:

$$ A = (A^{ij}) = (A(x^i, x^j)) $$
$$ (A\phi)^i = A^{ij}\phi^j .$$
In proof of \textbf{theorem 5} we have discovered that for every normal matrix $\hat{M}$, which is
neither hermitian nor skew hermitian, four couples $(U_1,D^\mu)$ exist, with $U_1$ unitary and $D^\mu$
diagonal, such that

\be U_1^\dag D^\mu \nabla_\mu U_1 = \hat{M} \label{fondamentale}\ee
\be \sqrt{|h|} h^{\mu\nu}(x^i) = \fr 1 2 d^{*\mu}_i d_i^\nu +c.c. \qquad D^{ij}_\mu = d^\mu_i \delta^{ij} .\ee
Here $h$ is a non degenerate metric while the first relation determines uniquely the values of gauge fields.
The matrices $\nabla_\mu, U_1, D^\mu$ act on field arrays via matricial product and the ensemble of four
couples $(U_1, D^\mu)$ is called \lq\lq space arrangement''.

Further, in proof of \textbf{theorem 6}, we have seen that for every invertible matrix $M$ we can always find
an unitary transformation $U_M$ and a normal matrix $\hat{M}$, which is neither hermitian nor skew hermitian,
such that $M = U_M \hat{M}$. If we define $U_2 = U_1 U_M^\dag$, we have

\be M^\dag M = \hat{M}^\dag \hat{M} \label{matemat} \ee
\be U_2^\dag D^\mu \nabla_\mu U_1 = M .\label{fond} \ee

\noindent It's sufficient to substitute (\ref{fond}) in (\ref{iniziale}) to verify its equivalence with (\ref{action-1}).
We have called $\hat{M}$ the \lq\lq associated normal matrix'' of $M$.

The action of a transformation $(U_1, U_2)$ on $\nabla$ follows from its action on $M$. We can always use the
invariance under $U(n,\mathbf{H}) \otimes U(n,\mathbf{H})$ to put $M$ in the form $M = D^\mu \na_\mu$. Starting
from this we have

$$U_2 M U_1^\dag = U_2 D^\mu \na_\mu U_1^\dag = U_2 D^\mu U_1^\dag U_1\na_\mu U_1^\dag.$$
We define $\nabla' = U_1\na_\mu U_1^\dag$ the transformed of $\nabla$ under $(U_1, U_2)$ and $D^{\prime\mu} =
U_2 D^\mu U_1^\dag$ the transformed of $D^\mu$. We assume that $A_\mu$ inside $\nabla_\mu$ transforms correctly
as a gauge field, so that

$$\na [A]_\mu \phi = \na [A] U_1^\dag \phi' = U_1^\dag \na [A_{U1}]_\mu \phi'$$
$$\phi' = U_1 \phi .$$
We want $D^{\prime\mu}$ remain diagonal and $h' = h[D'] = h[D]$. In this case there are two relevant possibilities:
\begin{enumerate}
\item $D$ is a matrix made by blocks $m \times m$ with $m$ integer divisor of $n$ and every block proportional to
identity. In this case the residual symmetry is $U(1,\mathbf{H})^n \times U(m,\mathbf{H})^{n/m}$ with elements
$(sV, V)$, $s$ both diagonal and unitary, $V \in U(m,\mathbf{H})^{n/m}$;
\item $h$ is any diagonal matrix. The symmetry reduces to $U(1,\mathbf{H})^n \otimes U(1,\mathbf{H})^n$ which is
local $U(1,\mathbf{H}) \otimes U(1,\mathbf{H}) \sim SU(2) \otimes SU(2) \sim SO(4)$.
\end{enumerate}

\noindent In this way, if we keep fixed the metric $h$ and keep diagonal $D$, the action (\ref{action-1})
will be invariant at least under $U(1,\mathbf{H})^n \otimes U(1,\mathbf{H})^n$ which doesn't modify $h$.

We have supposed that a potential for $M$ breaks the $U(n,\mathbf{H})\otimes U(n,\mathbf{H})$ symmetry in
$U(1,\mathbf{H})^n \otimes U(m,\mathbf{H})^{n/m}$ where $m$ is an integer divisor of $n$. We'll see in fact
that the more natural potential has the form $tr\,(\a M^\dag M - \b M^\dag MM^\dag M)$, known as \lq\lq
mexican hat potential''. This potential is a very typical potential for a spontaneous symmetry breaking. In
this way all the vertices are grouped in $n/m$ ensembles $\mathcal{U}^a$:

$$\mathcal{U}^a = \{x^a_1, x^a_2, x^a_3, \ldots, x^a_m\}$$

$$ \varphi = (\varphi (x^a_i)) = \left( \begin{array}[c]{ccccc}
\varphi(x^1_1) & \varphi(x^1_2) & \varphi(x^1_3) & \ldots & \varphi(x^1_m) \\
\varphi(x^2_1) & \varphi(x^2_2) & \varphi(x^2_3) & \ldots & \varphi(x^2_m) \\
\varphi(x^3_1) & \varphi(x^3_2) & \varphi(x^3_3) & \ldots & \varphi(x^3_m) \\
\vdots & \vdots & \vdots & \vdots & \vdots\\
\varphi(x^{n/m}_1) & \varphi(x^{n/m}_2) & \varphi(x^{n/m}_3) & \ldots & \varphi(x^{n/m}_4)
\end{array} \right) $$

$$ A = (A^{ab}_{ij}) = (A(x^a_i, x^b_j)) .$$
Now the indices $a,b$ of $A$ act on the columns of $\varphi$, while the indices $i,j$ act on the rows.
The fields $A^{ab}_{ij}$ with $a=b$ maintain null masses and then they continue to behave as gauge fields for $U(m,\mathbf{H})^{n/m}$.
Every $U(m,\mathbf{H})$ term in $U(m,\mathbf{H})^{n/m}$ acts independently inside a single $\mathcal{U}^a$. So, if we consider the ensembles
$\mathcal{U}^a$ as the \lq\lq real'' physical points, we can interpret $U(m,\mathbf{H})^{n/m}$ as a local $U(m,\mathbf{H})$.

It's simple to verify:

$$h^{\mu\nu}(x^a_i) = h^{\mu\nu}(x^a_j) \qquad \forall x^a_i, x^a_j \in \mathcal{U}^a$$
$$h^{\mu\nu}(x^a) \overset{!}{=} h^{\mu\nu}(\mathcal{U}^a) = h^{\mu\nu}(x^a_i) \qquad \forall x^a_i \in \mathcal{U}^a$$
$$A_{ij} (x^a) \overset{!}{=} Tr\,\left[ A(x^a)T^{ij} \right],\quad where$$
\ba A(x^a) &=& \sum_{ij} A(x^a_i, x^a_j) T^{ij},\quad \text{with $T^{ij}$ generator of $U(m,\mathbf{H})$}\nonumber \\
&& \phantom{l} \label{definitions}\ea

\section{Ricci scalar in the arrangement field paradigm}
\label{ricciscalar}

\subsection{Hyperions}

In this subsection we define an extension of $\mathbf{H}$ by inserting
a new imaginary unit $I$. It satisfies:

$$I^2 = -1 \qquad I^\dag = -I$$
$$[I,i] = [I,j] = [I,k] = 0$$

\noindent In this way a generic number assumes the form

$$v = a + Ib+ ic + jd + ke + iIf+ jIg + kIh, \qquad a,b,c,d,e,f,g,h \in \mathbf{R}$$
$$v = p + Iq, \qquad p,q \in \mathbf{R}$$

\noindent We call this numbers \lq\lq Hyperions'' and indicate their ensemble with $Y$.
It's straightforward that such numbers are in one to one correspondence with even
products of Gamma matrices. Explicitly:

$$1 \Leftrightarrow \g_0 \g_0 = 1 \qquad I \Leftrightarrow \g_5 = \g_0 \g_1 \g_2 \g_3$$
$$i \Leftrightarrow \g_2 \g_1 \qquad iI \Leftrightarrow \g_0 \g_3$$
$$j \Leftrightarrow \g_1 \g_3 \qquad jI \Leftrightarrow \g_0 \g_2$$
$$k \Leftrightarrow \g_3 \g_2 \qquad kI \Leftrightarrow \g_0 \g_1$$

\noindent Note that imaginary units $i,j,k,iI,jI,kI$ satisfy the Lorentz algebra,
with $i,j,k$ which describe rotations and $iI, jI, kI$ which describe boosts.

\begin{definition}[bar-conjugation]
The bar-conjugation is an operation which exchanges $I$ with $-I$ (or $\g_0$ with
$-\g_0$ in the $\g$-representa\-tion). Explicitly, if $v = a + Ib+ ic + jd + ke + iIf+
jIg + kIh$ with $a,b,c,d,e,f,g,h \in \mathbf{R}$, then $\bar{v} = a - Ib+ ic + jd +
ke - iIf - jIg - kIh$.
\end{definition}

\begin{definition}[pre-norm]
The pre-norm is a complex number with $I$ as imaginary unit (we say \lq\lq I-complex number'').
Given an hy\-pe\-rion $v$, its pre-norm is $|v| = (\bar{v}^\dag v)^{1/2}$. If $v \in \mathbf{H}$,
its pre-norm coincides with usual norm $(v^\dag v)^{1/2}$.
\end{definition}

Note that every hyperion $v$ can be written in the polar form

$$v = |v|e^{ia+jb+kc+iId+jIe+kIf} \qquad a,b,c,d,e,f$$
$$|v|^2 = \bar{v}^\dag v = |v|e^{-(ia+jb+kc+iId+jIe+kIf)} |v|e^{ia+jb+kc+iId+jIe+kIf}= |v|^2.$$

\noindent If $M$ takes values in $\mathbf{Y}$, the probability for the
existence of an edge $(ij)$ can be defined as $||M^{ij}||$, which is the
norm of pre-norm.

\begin{notation}[Spectral theorem in $\mathbf{Y}$] $\pt$\newline

\noindent The fundamental relation (\ref{fondamentale}) descends uniquely from spectral theorem
in $\mathbf{H}$. You can see from work of Yongge Tian \cite{Tian} that spectral theorem
is still valid in $Y$ in the following form: \lq\lq Every normal matrix $M$ with entries
in $\mathbf{Y}$ is diagonalizable by a transformation $U \in U(n,\mathbf{Y})$ which sends
$M$ in $UM\bar{U}^\dag$''. Here $U(n,\mathbf{Y})$ is the exponentiation of $u(n, \mathbf{Y})
= u(n, \mathbf{H})\cup Iu(n, \mathbf{H})$ and $M$ satisfies a generalized normality condition.
Explicitly, $\bar{U}^\dag = U^{-1}$ and $\bar{M}^\dag M = M \bar{M}^\dag$. This implies that
(\ref{fondamentale}) is valid too in the form

$$\bar{U}^\dag D^\mu \na_\mu U = M$$

\noindent Matrix $\na$ is now in $u(n, \mathbf{Y})$ and then it satisfies $\bar{\na}^\dag = -\na$.
Accordingly, its diagonal entries belong to Lorentz algebra (they don't comprise real and
$I$-imaginary components).

To conclude, we don't know if an associate normal matrix exists for any invertible
matrix with entries in $\mathbf{Y}$. Fortunately, in lorentzian spaces there is
no reason for using such machinery and we can start from the beginning with a
normal arrangement matrix.
\end{notation}

\begin{notation}[gauge fixing] \label{sinv}$\pt$\newline

\noindent It follows from spectral theorem that eigenvalues $\lambda$ of $M$ are equivalence
classes

$$\lambda \sim s\lambda \bar{s}^\dag \qquad s \in \mathbf{Y}, \bar{s}^\dag s = 1.$$

\noindent As a consequence, we can choose freely the diagonal matrix $D$ inside the equivalence
class $SD\bar{S}^\dag$, where $S$ is both diagonal and unitary ($\bar{S}^\dag = S^{-1}$).
This choice does't affect the metric $\sqrt h h^{\mu\nu} = Re\,(\bar{D}^{\dag\mu} D^\nu)$,
granting for the persistence of a symmetry $U(1,Y)^n = SO(1,3)^n$, ie local $SO(1,3)$.
Clearly this is a reworking of the usual gauge symmetry which acts on the tetrads, sending
$e^\mu_a$ in $\Lambda_a^{\pt b} e^\mu_b$ via the lorentz transformation $\Lambda$.
In what follows we exploit $SO(1,3)$-symmetry to satisfy two conditions:

\ba && tr\,(\{\bar{\na}^\dag_\mu, \na_\nu\} \bar{D}^{\dag\mu} D^\nu) = 0 \label{scond} \\
&& tr\,(D^\b \{\na_\b, \bar{\na}^\dag_\mu\} \bar{D}^{\dag\mu} D^\nu \na_\nu \bar{\na}^\dag_\a \bar{D}^{\dag\a}) = 0 \nonumber \ea

\noindent Note that these are global conditions because operator \emph{tr}
is analogous to a space-time integration.

\end{notation}

\subsection{Ricci scalar with hyperions}

In this subsection we simplify the form of Ricci scalar by means of
hyperions, in order to make it suitable for the arrangement field
formalism. Given a gauge field $\w_\mu$ in $so(1,3)$ and a complex
tetrad $e^\mu$, we define

\be A_\mu = \w^{ab}_\mu \g_{a}\g_{b} \qquad\qquad h^{\mu\nu} = Re\,(e^{\dag \mu}_a e^\nu_b \eta^{ab})\label{gaugey} \ee
$$d^\mu = \sqrt e e^{\mu a} \g_0 \g_a \qquad e = \left[ det(- e^{\dag \mu}_a e^\nu_b \eta^{ab}) \right]^{-1/2} \in \mathbf{R^+}$$
$$\bar{d}^\mu = d^\mu (\g_0 \rightarrow-\g_0)$$
$$\Rightarrow \bar{d}^{\dag \mu} d^{\nu} = ee^{\dag \mu a}e^{\nu b} \g_{a} \g_{b} \quad \Rightarrow \sqrt h h^{\mu\nu} = \fr 14 Re\,\left[ tr(\bar{d}^{\dag \mu} d^{\nu})\right] $$

\noindent Note that our definitions are the same to require $\bar{A}^\dag =-A$
in the hyperions framework. The Ricci scalar can be written as

$$\sqrt h R(x) = -\fr 18 tr\left(\left(\pa_\mu A_\nu - \pa_\nu A_\mu + [A_\mu, A_\nu]\right) \bar{d}^{\dag\mu} d^\nu\right)$$

\noindent To verify its correctness we expand first the commutator

\ba [A_\mu, A_\nu] &=& \w^{ab}_\mu \w^{cd}_\nu \left( \g_a\g_b\g_c\g_d - \g_c \g_d \g_a \g_b \right) \nonumber \\
                &=& \fr 12 \w^{ab}_\mu \w^{cd}_\nu \left( \g_a\{\g_b,\g_c\}\g_d - \g_c \{\g_d, \g_a\} \g_b \right)  + \nonumber \\
                && +\fr 1 {2}\w^{ab}_\mu \w^{cd}_\nu \left( \g_a [\g_b,\g_c ]\g_d - \g_c [\g_d, \g_a] \g_b \right)\nonumber \ea

\ba  &=& \left(\w^{ab}_\mu \w_{b\nu}^{\pt d} - \w^{ab}_\nu \w_{b\mu}^{\pt d} \right)\left( \g_a \g_d \right)  + \nonumber \\
                && +\fr 1 {4!}\w^{ab}_\mu \w^{cd}_\nu \left( \e_{abcd} \e^{efgh} \g_e \g_f \g_g \g_h \right)\nonumber \\
                &=& [\w_\mu, \w_\nu]^{ab}\g_a\g_b + \w^{ab}_\mu \w^{(D)}_{ab\nu} \,\g_5 \ea

\noindent In the last line we have defined $\w^{(D)}_{ab \nu} = \e_{abcd} \w_\nu^{cd}$. Hence

\ba R(x) &=& -\fr 18 tr(\g_a \g_b \g_c \g_d)\left( \pa_\mu \w^{ab}_\nu - \pa_\nu \w^{ab}_\mu + [\w_\mu,\w_\nu]^{ab}\right)e^{\dag c\mu} e^{d\nu} - \nonumber \\
        && -\fr 18 tr(\g_5 \g_b \g_c) \w^{ab}_\mu \w^{(D)}_{ab\nu} e^{\dag c\mu} e^{d\nu} \ea

\noindent Consider now the relations

$$\fr 14 tr(\g_a \g_b \g_c \g_d) = \eta_{ab}\eta_{cd} - \eta_{ac}\eta_{bd} + \eta_{ad}\eta_{bc}$$
$$tr(\g_5 \g_b \g_c) = 0$$

\noindent We obtain

$$R(x) = \left( \pa_\mu \w^{ab}_\nu - \pa_\nu \w^{ab}_\mu + [\w_\mu,\w_\nu]^{ab}\right)e^{\dag \mu}_a e^{\nu}_b $$
which is the usual definition.

\begin{notation}
A complex tetrad implies that tangent space is the complexification of
Minkowsky space (usually indicated with $\mathcal{CM}$). This fact gives
a strict connection with theory of \textbf{twistors}\cite{twistor}, where massless particles
move on trajectories which have an imaginary component proportional
to helicity.
\end{notation}

\noindent We can move freely from matrices $\g$ to hyperions,
substituting $tr$ with $4$. In this way

\ba \sqrt h R(x) &=& -\fr 12 \left(\pa_\mu A_\nu - \pa_\nu A_\mu + [A_\mu, A_\nu]\right) \bar{d}^{\dag\mu} d^\nu \nonumber \\
        &=& -\fr 12 [\na_\mu,\na_\nu]\bar{d}^{\dag\mu} d^\nu \nonumber \ea
$$\na_\mu =\pa_\mu + A_\mu \qquad\qquad A_\mu, d^\mu \in \mathbf{Y}$$

$$e^\mu_a = Re\,e^\mu_a + I\,Im\,e^\mu_a$$
\ba d^\mu &=& Re\,e^{\mu 0} + i I\, Re\,e^{\mu 3} + j I\, Re\,e^{\mu 2} + k I\, Re\,e^{\mu 1} + \nonumber \\
 && + I\,Im\,e^{\mu 0} - i\, Im\,e^{\mu 3} - j\, Im\,e^{\mu 2} - k\, Im\,e^{\mu 1} \nonumber \ea

\subsection{Ricci scalar in the new paradigm}

We try to define Hilbert-Einstein action as

$$S_{HE} = tr\,(\bar M^\dag M).$$

\noindent We insert in $S_{HE}$ the usual expansion $M = UD^\mu \na_\mu \bar{U}^\dag$, obtaining

\ba S_{HE} &=& tr\,[(\bar U \bar{D}^{\dag\mu} \bar{\na}_\mu U^\dag)^\dag (U D^\nu \na_\nu \bar{U}^\dag)] \nonumber \\
           &=& tr\,[ U \bar{\na}^\dag_\mu \bar{D}^{\dag\mu} \bar{U}^\dag U D^\nu \na_\nu \bar{U}^\dag] \nonumber \\
           &=& tr\,[\na_\nu \bar{\na}^\dag_\mu \bar{D}^{\dag\mu} D^\nu].\label{HEfirst}\ea

\noindent Now we can impose the first condition in (\ref{scond}) which gives

\ba S_{HE} &=& \fr 12 tr\,\{[\na_\nu, \bar{\na}^\dag_\mu ] \bar{D}^{\dag\mu} D^\nu \} \label{espansione} .\ea
Expanding the covariant derivatives we obtain

\ba S_{HE} &=& \fr 12 \sum_{a,b,c} \{ \pa^\dag_\mu A_\nu (x^a, x^b)-\pa_\nu \bar{A}_\mu^\dag (x^a, x^b) + \nonumber \\
    && \qquad\qquad\qquad\qquad + [\bar{A}^\dag_\mu, A_\nu](x^a, x^b)\}\bar{d}^{\dag\mu}(x^b)\d^{bc} d^{\nu}(x^c)\d^{ca} \nonumber \\
    &=& \fr 12 \sum_a \{ \pa^\dag_\mu A_\nu (x^a)-\pa_\nu \bar{A}^\dag_\mu(x^a) + \nonumber \\
        && \qquad\qquad\qquad\qquad +[\bar{A}^\dag_\mu, A_\nu](x^a, x^a)\}\bar{d}^{\dag\mu}(x^a)d^{\nu}(x^a) \nonumber \\
        &=& \fr 12\sum_{a,b\neq a} \{ \pa^\dag_\mu A_\nu (x^a)-\pa_\nu \bar{A}_\mu^\dag(x^a) + [\bar{A}_\mu^\dag (x^a), A_\nu(x^a)] + \nonumber \\
        && \qquad\qquad\qquad\qquad +[\bar{A}_\mu^\dag (x^a,x^b),A_\nu (x^b,x^a)]\} \cdot \bar{d}^{\dag\mu}(x^a)d^{\nu}(x^a) \nonumber \\  && \ea

\noindent Consider now a symmetry breaking with residual group $U(m,\mathbf{Y})^{n/m}$ which regroups
vertices in ensembles $\mathcal{U}^a = \{x^a_1, x^a_2,\ldots,x^a_m\}$.
We assume that fields $A(x^a_i,x^b_j)$ with $a \neq b$ acquire big masses and thus we can neglect them.
The symbol $\sum_a$ becomes $\sum_{a,i}$, while $\sum_{a,b\neq a}$ becomes $\sum_{a,i,b,j|(a,i)\neq (b,j)}$. After
neglecting heavy fields, the last one is simply $\sum_{a,i,j\neq i}$.

\ba S_{HE} &=& \fr 12\sum_{a} \{ \pa^\dag_\mu tr\,A_\nu (x^a)-\pa_\nu tr\,\bar{A}^\dag_\mu(x^a) +
[tr\,\bar{A}^\dag_\mu (x^a), tr\,A_\nu(x^a)] + \nonumber \\
        && +\sum_{i,j\neq i}[\bar{A}_\mu^{\dag ij}(x^a) A^{ji}_\nu(x^a) - A_\nu^{ij}(x^a) \bar{A}^{\dag ji}_\mu](x^a)\} \cdot \bar{d}^{\dag\mu}(x^a)d^{\nu}(x^a) \nonumber \\
        && \ea

\noindent For what follows we write $S_{HE} = \fr 12 \sum_a R^{ik}_{\mu\nu} \d^{ik}\bar{d}^{\dag\mu}d^\nu $ with

\ba R^{ik}_{\mu\nu} &=& \pa^\dag_\mu tr\,A_\nu (x^a)-\pa_\nu tr\,\bar{A}^\dag_\mu(x^a) + [tr\,\bar{A}^\dag_\mu (x^a), tr\,A_\nu(x^a)] + \nonumber \\
        && \qquad\qquad +\sum_{i,j\neq i,k\neq j}[\bar{A}_\mu^{\dag ij}(x^a) A^{jk}_\nu(x^a) - A_\nu^{ij}(x^a) \bar{A}^{\dag jk}_\mu](x^a) .\nonumber \\ && \label{curvature}\ea
$R^{ik}_{\mu\nu}$ is a generalization of curvature tensor.
We have indicated with $tr\,A$ the track on $ij$, ie $\d^{ij}A^{ij}(x^a) = \d^{ij}A(x^a_i,x^a_j)$.
Note that $[\bar{A}^{\dag ii},A^{jj}]$ is equal to zero when $i \neq j$ and then

$$\sum_{i} [\tilde{A}^{^\dag ii}_\mu, A^{ii}_\nu] = \sum_{ij}[\bar{A}^{^\dag ii}_\mu, A^{jj}_\nu]= [tr\,\bar{A}^\dag_\mu, tr\,A_\nu].$$

\noindent Consider now any skew hermitian matrix $W_\mu$ with elements $W_\mu^{ij} = A_\mu^{ij}$ for $i \neq j$ and $W_\mu^{ij} = 0$
for $i = j$. It belongs to the subalgebra of $u(m,\mathbf{Y})$ made by all null track generators.
This means that commutators between null track generators are null track generators too. In this way

$$ \sum_{i,i\neq j} [\bar{A}^\dag_\mu (x^i,x^j),A_\nu (x^j,x^i)] = tr [\bar{W}^\dag_\mu, W_\nu] = 0 .$$

\noindent Hence we can delete the mixed term in $S_{EH}$.

\ba S_{HE} &=& \fr 12 \sum_{a} \{ \pa^\dag_\mu tr\,A_\nu (x^a)-\pa_\nu tr\,\bar{A}^\dag_\mu(x^a) +
               [tr\,\bar{A}^\dag_\mu (x^a), tr\,A_\nu(x^a)]\}\cdot \nonumber \\
            &&\qquad\qquad\qquad\qquad\qquad\qquad\qquad\qquad\qquad \cdot \bar{d}^{\dag\mu}(x^a)d^{\nu}(x^a)\nonumber \ea

\noindent In the arrangement field paradigm, the operator $\dag$ transposes also rows with
columns in matrices which represent $\pa$ and $A$. As we have seen, the fields
$A$ which intervene in $R$ are only the diagonal ones, so the transposition of rows with
columns is trivial. Note that $\na$ satisfies a generalized condition of skew-hermiticity
($\bar{\na}^\dag = -\na$) and then its diagonal components belong to
lorentz algebra. This implies $tr\,\bar{A}^\dag = -tr\,A$, matching exactly with our request
in (\ref{gaugey}). Finally, if we consider the matrix which represents $\pa$ (we have called
it $\tilde M$), we note that $\bar{\pa}^\dag = \pa^T = -\pa$. Explicitly

$$\na_\nu^\dag = (\pa_\nu + tr\,\bar{A}_\nu)^\dag = \pa_\nu^\dag + tr\,\bar{A}_\nu^\dag = -\pa_\nu - tr\,A_\nu = -\na_\nu .$$

\noindent Applying this to $S_{HE}$,

\ba S_{HE} \!\! &=& \!\! -\fr 12 \sum_{a} \{ \pa_\mu tr\,A_\nu (x^a)-\pa_\nu tr\,A_\mu(x^a) +
               [tr\,A_\mu (x^a), tr\,A_\nu(x^a)]\}\cdot \nonumber \\
            &&\qquad\qquad\qquad\qquad\qquad\qquad\qquad\qquad\qquad \cdot \bar{d}^{\dag\mu}(x^a)d^{\nu}(x^a) \nonumber \\
            &=& -\fr 12 [\overset{G}{\na}_\mu,\overset{G}{\na}_\nu]\bar{d}^{\dag\mu}(x^a)d^{\nu}(x^a) \nonumber \\
           &=& \sum_a \sqrt{h}R(x^a) \ra \int d^4 x \, \sqrt{h}R(x).\ea

\noindent Here $\overset{G}{\na}$ is the gravitational covariant derivative $\overset{G}{\na} = \pa + tr\,A$.
It's very remarkable that gauge fields in $R$ are only the diagonal ones. First, this is the
unique possibility to obtain ${\overset{G}{\na^\dag}_\nu} = -\overset{G}{\na}_\nu$. Moreover, while gauge fields in $R$
are tracks of matrices $(A_{ij})(x^a)$, we'll see as the other gauge fields in Standard Model
correspond to non diagonal components.

\section{The kinetic term}
\label{kinetic}
Until now we have obtained no terms which describe gauge interactions.
In this section we find a such term, with the condition that it hasn't
to change Einstein equations. One option is as follows:

\ba
S_{GB}  &=& -tr\,(\bar{M}^\dag M \bar{M}^\dag M) \label{eq: opz}\\
        &=& -tr\,\left[ U \bar{\na}^\dag_\mu \bar{D}^{\dag\mu} \bar{U}^\dag U D^\nu \na_\nu \bar{U}^\dag U \bar{\na}^\dag_\a \bar{D}^{\dag\a} \bar{U}^\dag U D^\b \na_\b \bar{U}^\dag \right]  \nonumber \\
        &=& -tr\,\left[ \bar{\na}^\dag_\mu \bar{D}^{\dag\mu} D^\nu \na_\nu \bar{\na}^\dag_\a \bar{D}^{\dag\a} D^\b \na_\b \right] \nonumber \ea

\noindent We assume a residual symmetry under $U(m,\mathbf{Y})^{n/m}$. This means that
$D^\mu$ are matrices made of blocks $m \times m$ where every block is a
hyperionic multiple of identity. We use newly the correspondence between
$(1,I,i,j,k,iI,jI,kI)$ and gamma matrices:

\ba S_{GB}  &=& -\fr 14 tr(\g_a\g_b\g_c\g_d\g_e\g_f\g_g\g_h)\,\left[ \na_\b^a \bar{\na}^{b\dag}_\mu \bar{D}^{\dag
c\mu} D^{d\nu} \na^e_\nu \bar{\na}^{f\dag}_\a \bar{D}^{\dag g\a} D^{h\b} \right] \nonumber \ea

\noindent We use letters $a,b,c,d$ for indices which run on Gamma matrices, $\a,\b,\mu,\nu$
for spatial coordinates indices and $ijk$ for gauge indices (ie indices which run
inside a single $\mathcal{U}^a$). Pay attention to not confuse the index $a$ in the first group with the index
$a$ which runs over the vertices like in $x^a_i$.

We will see that physical fields arise in three families, determined by the
choice of a subspace inside $Y$. This is true both for fermionic and bosonic
fields. Thus the indices with letters $a,b,c,d$ run over the three families.

We proceed by imposing the second condition in (\ref{scond}), in such a way to
ignore terms proportional to $\{\na_\b, \bar{\na}^\dag_\mu\}$ inside $S_{GB}$.
We take

$$ S_{GB} = \sum_a L_{GB} (x^a) $$
Then

\ba L_{GB} &=& R^{ij}_{ab\mu\b} R^{ab ji}_{\nu\a} \bar{d}^{\dag\mu}_c d^{c\nu} \bar{d}^{\dag \a}_d d^{d\b} -
4R^{ij}_{ac\mu\b} \bar{d}^{\dag a\mu} R^{cb ji}_{\nu\a} d^{\a}_b d^{d\b} \bar{d}^{\dag \a}_d  + \nonumber   \\
        && + R^{ij}_{ac\mu\b} \bar{d}^{\dag a\mu} d^{c\b} R^{cb ji}_{\nu\a} \bar{d}^{\dag \nu}_c d^{\a}_b   \nonumber \\
        &=& h R^{ij}_{ab\mu\b} R^{ab ji \mu\b} - 4 h R^{ij}_{c \b} R^{c ji \b} + h R^{ij}R^{ji} \ea

\noindent $R^{ij}_{\b\mu}$ was defined in (\ref{curvature}), while $\sqrt[4] h R^{ij}_\mu
= R^{ij}_{\b\mu} d^{\b}$ and $\sqrt h R^{ij} = R^{ij}_{\b\mu} d^{\b} d^{*\mu}$.
You understand in a moment that for $i \neq j$ we have $R^{ij}_{ac\b\mu}
R^{ji ac}_{\nu\a} h^{\mu\a}h^{\nu\b} = tr\,\sum_{(ac)} F^{(ac)}_{\mu\nu} F^{(ac) \mu\nu}$.
The index $(ac)$ runs over three fields families and $F_{(ac)\mu\nu}$
is a strength field tensor. In this way the terms $R^{ij \nu}_\b R^{ji \b}_\nu$
and $R^{ij}R_{ji}$ are terms which mix families.

The trouble with $S_{GB}$ is that it generates a factor $h$ instead of $\sqrt{h}$.
However, we can solve the problem imposing the gauge condition $h=1$. Note that for
$i=j$ we have

$$L_{GB} = R_{ac\b\mu} R^{ac\b\mu} + R^2 - 4 R^{\a}_\mu R^{\mu}_\a $$
which is a topological term and it doesn't change the Einstein equations.

\begin{notation}[symmetry breaking]
The combination of $S_{HE}$ and $S_{GB}$ gives to gravitational
gauge field $\overset{G}{A}$ a potential with form

$$\overset{G}{A^2} - \overset{G}{A^4}.$$
This potential has non trivial minimums which imply a non-trivial
expectation value for $\overset{G}{A}$. Moreover, inside $S_{GB}$ we
find the following kind of terms for other fields $A$:

$$\langle \overset{G}{A^2} \rangle A^2 - A^4.$$
In this way we have a mass for gauge fields $A$ and another
potential with non-trivial minimums. Therefore, also gauge fields
$A$ have non-trivial expectation values. Finally, such expectation values
give mass to fermionic fields via terms

$$\psi^\dag \langle A \rangle \psi.$$
There is no need for a scalar Higgs boson.
\end{notation}

\section{Connections with Strings and Loop Gravity}\label{string}
We have seen in \cite{Arrangement}, at \textbf{Remark 13}, that some similarities exist between
diagonal components of $M$ (loops) and closed strings in string theory. Now we have discovered
that such diagonal components describe a gravitational field. Is then a case that the lower energy
state for closed string is the graviton? We think no. Moreover, we have seen that gauge fields
correspond to non-diagonal components of $M$, ie open edge in the graph. This finds also a connection
with open strings, whose lower energy states are gauge fields. We have shown that a symmetry $U(m,\mathbf{Y})$
arises when vertices are grouped in ensembles $\mathcal{U}^a$ containing $m$ vertices. This seems to
represent a superimposition of $m$ universes or branes. Gauge fields for such symmetry correspond to open
edge which connect vertices in the same $\mathcal{U}^a$. Is then a case that the same symmetry arises
in open strings with endpoints in $m$ superimposed branes? We still think no.
Until now we have supposed that open edges between vertices in the same $\mathcal{U}^a$ have length zero, so
that we haven't to introduce extra dimensions. However, by $T-duality$ such edges correspond to open
strings with $U(m,\mathbf{Y})$ Chan-Paton which moves in an infinite extended extra dimension. This happens because
an absente extra dimension is a compactified dimension with $R = 0$ and $T-duality$ sends $R$ in $1/R$.
Regarding edges between vertices in different $\mathcal{U}^a$, we see that they have a mass proportional to
separation between endpoints. This is true both in our model and string theory.

$\pt$

\noindent The following two theorems emphasize a triality between \emph{Arrangement Field Theory},
\emph{String Theory} and \emph{Loop Quantum Gravity}. We can see as they are different manifestations
of the same theory.

\begin{theorem}\label{Loop}
Every element $M^{ij}$ in the arrangement matrix can be written as a state in the
Hilbert space of \emph{Loop Quantum Gravity}, ie an holonomy for a
$SO(1,3)$ gauge field\footnote{In Loop Gravity the gauge field appears usually in the
form $iA$ with $A$ hermitian. We incorporate the $i$ inside $A$ so that $A^{ab}\g_a\g_b$ corresponds
to a hyperionic number.}. In this way, every field (gauge or gravitational) becomes
a manifestation of only gravitational field.
\end{theorem}

\begin{proof}
An element $M^{ij}$ can always be written in the following form:

\be M^{ij} = |M^{ij}| exp \left(\int_{x_i}^{x_j} A_\mu dx^\mu\right) \label{defM} \ee
with $\mu = 1,2,3$ and

$$|M^{ij}|= exp \left(\int_{x_i}^{x_j} A_0 dx^0\right).$$
Here $A_\mu$ is a $SO(1,3)$ connection and $A_0$ is an $I$-complex field.
Obviously, we take $A_\mu$ hyperionic by using the usual correspondence
with Gamma matrices. In this way $A_\mu$ is purely imaginary.
The integration is intended over the edge which goes from vertex $i$ to vertex $j$, parametrized
by any $\t \in [0,1]$. If you look (\ref{defM}), you see on the left a discrete space (the graph) with
discrete derivatives and fields which are defined only on the vertices. On the right you find
instead a Hausdorff space with continuous paths, continuous derivatives and fields which are
defined everywhere. Applying eventually a transformation in $U(n,\mathbf{Y})$,
we have

$$ M^{ij} = D^{ik\mu}\na^{kj}_\mu = D^{ii\mu} \na^{ij}_\mu = d^\mu(x_i) \na^{ij}_\mu. $$
In the following we introduce a real constant $\lambda$, with length dimensions, in order to make $M$ dimensionless:

\be M^{ij} = \lambda D^{ik\mu}\na^{kj}_\mu = \lambda D^{ii\mu} \na^{ij}_\mu = \lambda d^\mu(x_i) \na^{ij}_\mu. \label{scompose}\ee
In \emph{Loop Quantum Gravity} we consider any space-time foliation defined by some temporary
parameter and then we quantize the theory on a tridimensional slice. The simpler choice is
a foliation along $x_0$: in this case the metric on the slice is simply the spatial block $3\times 3$
inside the four dimensional metric when it's taken in temporary gauge. In such framework we have
$d^0 = \mathbf{1}$ and $[d^\mu(x), A_\nu(x')] = G\d^\mu_\nu \d^3(x-x')$ with $\mu,\nu = 1,2,3$.
We deduce the relation $d^\mu(x) = G\d / {\d A_\mu(x)}$ and apply it to (\ref{scompose}) when vertices $i$
and $j$ sit on the same slice. We obtain

\be d^\mu(x_i) \na^{ij}_\mu = G\fr \d {\d A_\mu (x_i)} \na^{ij}_\mu = \fr 1 \lambda |M^{ij}|
exp \left({\int_{x_i}^{x_j} A_\mu dx^\mu}\right)\label{appl}\ee
with $\mu = 1,2,3$. Note that $x_0 (x_i) = x_0 (x_j)$ when $i$ and $j$ sit on the same slice. Hence

$$|M^{ij}|= exp \left(\int_{x_i}^{x_j} A_0 dx^0\right) = exp \left(\oint A_0 dx^0\right).$$
Consider now the following relation:

\be exp \left({\int_{x_i}^{x_j} A_\mu dx^\mu}\right) = \fr \d {\d A_\nu} \int_\Omega d^2s \, n_\nu exp \left({\int_{x_i}^{x_j} A_\mu dx^\mu}\right)\label{relaz}\ee
with
$$n_\nu = \fr 12 \e_{\nu\mu\a} \fr {\pa x^\mu}{\pa s^a} \fr {\pa x^\a}{\pa s^b}\e^{ab}.$$
$\Omega$ is a two dimensional surface parametrized by coordinates $s^a$
with $a=1,2$ and $\int_\Omega d^2 s = G$. We assume that $\Omega$ contains
the vertex $x_i$ and no other point which is a vertex or sits along an edge.
Substituting (\ref{relaz}) in (\ref{appl}) we obtain

$$ \fr \d {\d A_\nu^i} \na^{ij}_{\nu} = \fr 1 {\lambda G}\fr \d {\d A_\nu}
\int_\Omega d^2s \, n_\nu |M^{ij}| exp \left({\int_{x_i}^{x_j} A_\mu dx^\mu}\right)$$
and then

\ba \na^{ij}_{\nu} &=& \fr 1 {\lambda G}\int_\Omega d^2s \, n_\nu |M^{ij}| exp \left({\int_{x_i}^{x_j} A_\mu dx^\mu}\right) +K_\nu (x_i, x_j) \nonumber \\
                &=& \fr 1 {\lambda G} \int_\Omega d^2s \, n_\nu exp \left({\int_{x_i}^{x_j} A_\mu dx^\mu}\right) + K_\nu (x_i, x_j).\nonumber \ea
$K_\nu$ is any function of $x_i$ and $x_j$ independent from $A_\mu$.
In the se\-cond line we have taken $\mu =0,1,2,3$. For diagonal components this becomes

\be A_\nu^{ii} = \fr 1 {\lambda G}\int_\Omega d^2s \, n_\nu exp \left({\oint A_\mu dx^\mu}\right)+K_\nu (x_i).\label{una}\ee
We have used $\pa^{ii} = 0$ because the matrix which represents the discrete derivative is
null along diagonal. We choose loops and surfaces $\Omega$ in such a way to have

$$n_\nu \oint A_\mu dx^\mu = \lambda A_\nu (x_i) + O(\lambda^2).$$
Applying this into (\ref{una}), it becomes

\ba A^{ii}_\nu &=& \fr 1 {\lambda G} \int_\Omega d^2s \, n_\nu \left( 1 + \oint A_\mu dx^\mu + O(\lambda^2) \right) +K_\nu (x_i)\nonumber \\
                        &=& \fr 1 {\lambda G} (G n_\nu + G \lambda A_\nu (x_i)+G\cdot O(\lambda^2)) + K_\nu (x_i)\nonumber \\
                        &=& \fr 1 \lambda (n_\nu + \lambda A_\nu (x_i)+ O(\lambda^2))+ K_\nu (x_i).\ea
If we set $K_\nu (x_i) = -n_\nu (x_i) / \lambda$, we obtain

$$A^{ii}_\nu = A_\nu (x_i) + O(\lambda).$$
This verifies the consistence of our definition and proves the theorem.

Note that $\l$ could be taken equal to $\Delta$ because $M$ contains a factor $\Delta^{-1}$
from definition (\ref{dderiv}) of $\tilde{M}$. In such case we obtain

$$A^{ii}_\nu = A_\nu (x_i)$$
in the continuous limit.
\end{proof}
\begin{notation}[third quantization]
Note that canonical quantization of gauge fields implies
\normalsize
$$\left[ \pa_0 A^{ij}_\a (x_a), A^{ij}_\nu (x_b) \right] =
\left[\left(\int d^4 x \pa_0 A_\mu (x) \fr {\d \na^{ij}_\a}
{\d A_\mu (x)}\right)(x_a), \na^{ij}_\nu(x_b)\right] = \d_{\a\nu}\d^3(x_a -x_b).$$
\large
\noindent Integration in the first factor is over continuous coordinates of Hausdorff space.
Conversely, the argument $x_a$ indicates simply to what ensemble $\mathcal{U}^a$
the edge $(ij)$ belongs. Here we have used $\pa^{ij} = 0$, which holds not only for
$i = j$ but also for $x_i$ and $x_j$ in the same ensemble $\mathcal{U}^a$. This implies
$\na^{ij} = A^{ij}$. Moreover $\na^{ij}$ is a state in the Hilbert space of Loop
Quantum Gravity and hence we have a sort of third quantization which applies on
gravitational states and creates gauge fields:

$$\left[\left( \int d^4 x \dot{A}_\mu(x) \fr {\d \Psi[\Lambda,A]}{\d A_\mu(x)}\right), \Psi^\dag[\Lambda',A]\right] = \d (\Lambda - \Lambda').$$
$$\left[\left( \int d^4 x \dot{A}_\mu(x) \fr {\d \Psi[\Lambda, A]}{\d A_\mu(x)}\right), \Psi^\dag[\Lambda, A']\right] = \d (A - A').$$
This implies
\normalsize
$$\Psi[A] = \int D[d^\mu]\, a(d)\, exp\left(\fr 1 G{\int d^4 x \,d^\mu A_\mu}\right) + b^\dag(d) \, exp\left(\fr 1 G{\int d^4 x \,d^{\dag\mu} A^\dag_\mu}\right)$$
\large
$$\left[a(d), a^\dag (d')\right] = \fr 1 {\int d^4 x \dot A_\nu d^\nu} \d(d-d^{\dag\prime})$$
$$\left[b(d), b^\dag (d')\right] = \fr 1 {\int d^4 x \dot A_\nu d^\nu} \d(d-d^{\dag\prime})$$
\end{notation}

\begin{figure}[ptbh]
\centering\includegraphics[width=0.6\textwidth ]{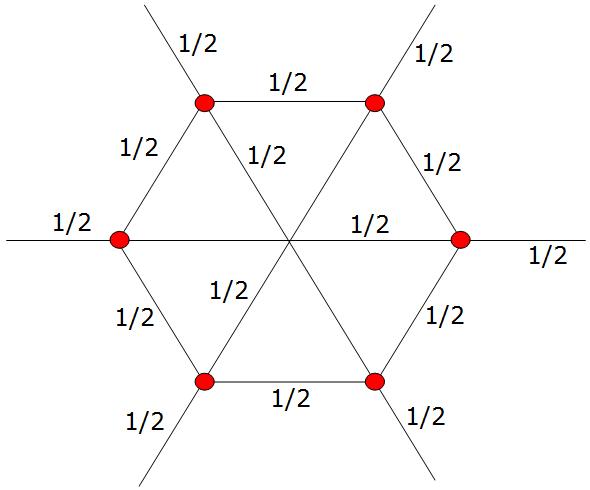}
\caption{A spin network with symmetry $U(6,\mathbf{Y}$). The six vertices are assumed superimposed.}
\label{Spin-network}
\end{figure}
In figure \ref{Spin-network} we see a spin network which defines a
$U(6,\mathbf{Y})$ gauge field $A^{ij}$ with $i,j =1,2,3,4,5,6$. The vertices
are assumed superimposed. The symmetry group is bigger than
$U(1,\mathbf{Y})^6 \sim SO(1,3)^6$ which acts separately on the single vertices.
The group grows in fact to $U(6,\mathbf{Y})$ because we can exchange
the vertices without change the graph. We have the same situation with
open strings: six strings with endpoints on six separated branes define
a state with symmetry $U(1)^6$ but, if the branes are superimposed,
the symmetry becomes $U(6)$.

Generators in $u(6,\mathbf{Y})$ are generators in $u(6,\mathbf{H})$ multiplied
by $1$ or $I$. In turn, generators in $u(6,\mathbf{H})$ can be divided in three families
of generators in $u(6)$, one for every choice of imaginary unit ($i,j$ or $k$).
Note that commutation relations for $U(6)$ are satisfied if and only if

$$U^{ij}U^{jk} = U^{ik},$$
where $U^{ij}$ is the holonomy from $x_i$ to $x_j$. Hence

$$A_\mu = \pa_\mu \Gamma \qquad \,\,\,\text{with}\,\,\,\Gamma\,\,\,\text{scalar.}$$
This means that gauge fields in $U(6)$ could exist without gravity, ie
when $A$ is a pure gauge. Otherwise, an holonomy with $A \neq \pa \Gamma$
exchanges gauge fields between different families.

\begin{theorem}
The actions $tr\,(M^\dag M)$ and $tr\,(M^\dag M M^\dag M)$ are sums of
exponentiated string actions.
\end{theorem}
\begin{proof}
We obtain from theorem \ref{Loop}:

\ba M^{ij} M^{*jk} M^{kl} M^{*li} &=& exp\left(\int_{\pa\Box} A_\mu dx^\mu \right)  \\
                            &=& exp\left(\int_{\Box} F_{\mu\nu} dx^\mu \wedge dx^\nu \right) \nonumber \\
                            &=& exp\left(\int_{\Box} \e^{ab} F_{\mu\nu} X^\mu_{,a} X^\nu_{,b} \, d^2 s\right) \nonumber \ea
This is the exponential of an action for open strings whose worldsheet is a square made
by edges $(ij)$, $(jk)$, $(kl)$, $(li)$. The strings move in a curved background with
antisymmetric metric $F_{\mu\nu} = (d \wedge A)_{\mu\nu}$. In a similar manner

\ba M^{ij} M^{*jk} M^{ki} &=& exp\left(\int_{\bigtriangleup} \e^{ab} F_{\mu\nu} X^\mu_{,a} X^\nu_{,b} \, d^2 s\right) \ea
This is the exponential of an action for open strings whose worldsheet is a triangle.

\ba M^{ij} M^{*ji} &=& exp\left(\int_{O} \e^{ab} F_{\mu\nu} X^\mu_{,a} X^\nu_{,b} \, d^2 s\right) \ea
This is the exponential of an action for open strings whose worldsheet is a circle.

\ba M^{ii} &=& exp\left(\int_{O} \e^{ab} F_{\mu\nu} X^\mu_{,a} X^\nu_{,b} \, d^2 s\right) \ea
The same of above.

\ba M^{ii}M^{jj} &=& exp\left(\int_{Cil} \e^{ab} F_{\mu\nu} X^\mu_{,a} X^\nu_{,b} \, d^2 s\right) \ea
This is the exponential of an action for closed strings whose worldsheet is a cilinder. This concludes the proof.
\end{proof}

\section{Standard model interactions \label{electroweak}}

We suppose that a residual symmetry for $U(6,\mathbf{Y})^{n/6}$ survives.
If we consider the ensembles $\mathcal{U}^a = (x^a_1, x^a_2, x^a_3, x^a_4, x^a_5, x^a_6)$ as the
real physical points, $U(6,\mathbf{Y})^{n/6}$ can be considered as a local $U(6,\mathbf{Y})$.
We have defined $u(6,\mathbf{Y})$ as the complexified Lie algebra of $U(6,\mathbf{H})$,
generated by all matrices in $u(6,\mathbf{H})$ and $Iu(6,\mathbf{H})$.
By exponentiating $u(6,\mathbf{Y})$ we obtain a simple Lie group with complex dimension $78$.
This group is the symplectic group $Sp(12,\mathbf{C})$ and $U(6,\mathbf{H})$ is its real
compact form, sometimes called $Sp(6)$.
We consider the fields $A(x^a_i, x^b_j)$ with $a = b$ (we call them $A(x^a)$). They are
$6 \times 6$ skew adjoint hyperionic matrices $\bar{A}^\dag = -A$. These
matrices form the $Sp(12,\mathbf{C})$ algebra which has $156$ generators $\w$ with $\bar{\w}^\dag = -\w$.

$$\w = \left( \begin{array}[c]{cccccc}
 \vec{y}   & b+\vec{b}  & c+\vec{c}  & d+\vec{d}  & e+\vec{e}  & m+\vec{m}  \\
 -b+\vec{b}& \vec{a}_1  & f+\vec{f}  & g+\vec{g}  & h+\vec{h}  & p+\vec{p}  \\
 -c+\vec{c}& -f+\vec{f} & \vec{a}_2  & s+\vec{s}  & q+\vec{q}  & r+\vec{r}  \\
 -d+\vec{d}& -g+\vec{g} & -s+\vec{s} & \vec{a}_3  & k+\vec{k}  & t+\vec{t}  \\
 -e+\vec{e}& -h+\vec{h} & -q+\vec{q} & -k+\vec{k} & \vec{a}_4  & v+\vec{v}  \\
 -m+\vec{m}& -p+\vec{p} & -r+\vec{r} & -t+\vec{t} & -v+\vec{v} & \vec{a}_5  \\
\end{array} \right) $$

\noindent Consider now the subalgebra of the following form with complex
(not hyperionic) components except for $y$ which remains hyperionic:

$$\w = \left( \begin{array}[c]{cccccc}
 \vec{y}   & 0          & 0          & 0          & 0          & 0          \\
 0         & \vec{a}_1  & f+\vec{f}  & g+\vec{g}  & h+\vec{h}  & p+\vec{p}  \\
 0         & -f+\vec{f} & \vec{a}_2  & s+\vec{s}  & q+\vec{q}  & r+\vec{r}  \\
 0         & -g+\vec{g} & -s+\vec{s} & \vec{a}_3  & k+\vec{k}  & t+\vec{t}  \\
 0         & -h+\vec{h} & -q+\vec{q} & -k+\vec{k} & \vec{a}_4  & v+\vec{v}  \\
 0         & -p+\vec{p} & -r+\vec{r} & -t+\vec{t} & -v+\vec{v} & \vec{a}_5  \\
\end{array} \right) $$

\noindent Moreover we put the additional condition $\vec{a} = \sum_l \vec{a}_l = 0$.
The field $y = tr\,\w$ is the only one which contributes to Ricci scalar.
Conversely, all other fields belong to a $SU(5)$ subgroup, which defines the Georgi - Glashow
grand unification theory. The symmetry breaking in Georgi - Glashow model is induced by Higgs
bosons in representations which contain triplets of color. These color triplet Higgs can mediate
a proton decay that is suppressed by only two powers of GUT scale. However, our mechanism of
symmetry breaking doesn't use such Higgs bosons, but descends from the expectation values of
quadratic terms $AA$, which derive from non trivial minimums of a potential $AA - AAAA$. So we
circumvent the problem.

Restrict now the attention to the $SO(1,3) \otimes SU(2) \otimes U(1) \otimes SU(3)$
generators, that are the generators of standard model plus gravity.

$$\w = \left( \begin{array}[c]{cccccc}
 \vec{y}   & 0          & 0          & 0          & 0          & 0          \\
 0         & \vec{a}_1  & f+\vec{f}  & 0          & 0          & 0          \\
 0         & -f+\vec{f} & \vec{a}_2  & 0          & 0          & 0          \\
 0         & 0          & 0          & \vec{a}_3  & k+\vec{k}  & t+\vec{t}  \\
 0         & 0          & 0          & -k+\vec{k} & \vec{a}_4  & v+\vec{v}  \\
 0         & 0          & 0          & -t+\vec{t} & -v+\vec{v} & \vec{a}_5  \\
\end{array} \right) $$

\noindent We'll show in a moment that all standard model fields transform under this subgroup in the adjoint
representation. In this way themselves are elements of $Sp(12,\mathbf{C})$ algebra, explicitly:

$$\psi = \psi^1 + I\psi^2 = \left( \begin{array}[c]{cccccc}
 0              & e            & -\nu       & d^c_{R}        & d^c_{G}          & d^c_{B}   \\
 -e^*        & 0            & e^c        & -u_{R}         & -u_{G}           & -u_{B}    \\
 \nu^*       & -e^{c*}   & 0          & -d_{R}         & -d_{G}           & -d_{B}    \\
 -d^{c*}_R   & u^*_R     & d^*_R   & 0              & u^c_{B}          & -u^c_{G}  \\
 -d^{c*}_G   & u^*_G     & d^*_G   & -u^{c*}_{B} & 0                & u^c_{R}   \\
 -d^{c*}_B   & u^*_B     & d^*_B   & u^{c*}_{G}  & -u^{c*}_{R}   & 0         \\
\end{array} \right) $$

\noindent We have used the convention of Georgi - Glashow model, where the basic fields of $\psi^1$
are all left and the basic fields of $I\psi^2$ are all right. We have indicated with $^c$ the
charge conjugation. The subscripts $R,G,B$ indicates the color charge for the strong interacting
particles (R=red, G=green, B=blue).

In Georgi - Glashow model the fermionic fields are divided in two families. The first one transforms
in the representation $\bar{5}$ of $SU(5)$ (the fundamental representation). It is exactly the array
$(\w^{1j})$ in the matrix above, with $j = 2,3,4,5,6$. This array transforms in fact in the fundamental
representation for transformations in every $SU(5) \subset Sp(12,\mathbf{C})$ which acts on indices values
$2 \div 6$.

The second family transforms in the representation $10$ of $SU(5)$ (the skew symmetric representation).
Unfortunately it isn't the sub matrix $(\w^{ij})$ with $i,j = 2,3,4,5,6$. This is in fact the skew adjoint
representation of $U(5,\mathbf{Y})$, which is skew hermitian and not skew symmetric.

Do not lose heart. We'll see in a moment that such adjoint representation is a quaternionic combination
of three skew symmetric representations, one for every fermionic family. This concept could appears cumbersome,
but it will be clear along the following calculations.

\begin{theorem} The skew adjoint representation of $U(m,\mathbf{H})$ is a qua\-ternionic combination of
three skew symmetric representations of $U(m) = U(m,\mathbf{C})$ plus a real skew symmetric representation (which is also
skew hermitian).
\end{theorem}

\begin{proof}
Consider a fermionic matrix $\psi$ which transforms in the adjoint representation
of $U(m,\mathbf{H})$:

$$\psi \ra U\psi U^\dag $$
Take then a matrix $\psi'$ with $\psi' k =\psi$. Its transformation law
under $U(m) = U(m,\mathbf{C})$ is easily derived when this group is
constructed by using imaginary unit $i$ or $j$:

$$\psi' k \ra U \psi' k U^\dag = U \psi' U^T k .$$
Here we have used the relation $k \lambda = \lambda^* k$ for $\lambda
\in \mathbf{H}$ without $k$ component. We see that $\psi'$ transforms
in the skew symmetric representation:

$$\psi' \ra U \psi' U^T$$
We obtain a complex matrix $\psi'$ (with $i$ as imaginary unit) when
$\psi$ has the form $Ak+Bj$ with $A,B$ real matrices. Indeed:

$$\psi' = - \psi k = -Akk-Bjk = A - Bi$$
Sending $\psi$ in $\psi^*$ we bring $\psi'$ to $-\psi'$ and so we satisfy the skew symmetry.
Finally we can always write

$$\psi = \psi_0 + \psi_1 k + \psi_2 i + \psi_3 j$$
In this decomposition $\psi_1, \psi_2, \psi_3$ are complex matrices with complex unit
respectively $i, j, k$. Explicitly:

\ba \psi_1 &=& \phi_1 - i\xi_1 \qquad = \phi_1^1 - i\xi_1^1 + I(\phi_1^2 - i\xi_1^2)\nonumber \\
\psi_2 &=& \phi_2 - j\xi_2  \qquad = \phi_2^1 - j\xi_2^1 + I(\phi_2^2 - j\xi_2^2) \nonumber \\
\psi_3 &=& \phi_3 - k\xi_3  \qquad = \phi_3^1 - k\xi_3^1 + I(\phi_3^2 - k\xi_3^2).\nonumber \ea

\noindent Here all $\phi^1$, $\phi^2$ and $\xi^1$, $\xi^2$ are real fields.
In this way $\psi_{1,2,3}$ transform in the skew symmetric representation
of $U(m)$ when we construct this group by using the correspondent imaginary
unit ($i$ for $\psi_1$, $j$ for $\psi_2$ and $k$ for $\psi_3$). Hence they
define the famous three fermionic families, relate each other by $U(1,\mathbf{H})$
transformations. Moreover $\psi_0$ is a real skew symmetric field.
\end{proof}

$\pt$

\noindent Consider the following lagrangian

\ba tr(\psi^{\dag} \na \psi) &=& tr(k^* \psi^{\dag}_1 \na \psi_1 k) + tr(i^* \psi^{\dag}_2 \na \psi_2 i) +tr(j^* \psi^{\dag}_3 \na \psi_3 j) \nonumber \\
&& - tr(i^* \phi_2^\dag \na \xi_3 i) - tr(j^* \phi_3^\dag \na \xi_1 j) - tr(k^* \phi_1^\dag \na \xi_2 k) \nonumber \\
&&- tr(\psi_0^\dag \na \psi_0) \nonumber \ea
\ba &=& tr(\psi^{\dag}_1 \na \psi_1 kk^*) + tr(\psi^{\dag}_2 \na \psi_2 ii^*) +tr(\psi^{\dag}_3 \na \psi_3 jj^*) \nonumber \\
&& - tr(\phi_2^\dag \na \xi_3 ii^*) - tr(\phi_3^\dag \na \xi_1 jj^*) - tr(\phi_1^\dag \na \xi_2 kk^*) \nonumber \\
&& - tr(\psi_0^\dag \na \psi_0) \nonumber \\
&=& tr(\psi^{\dag}_1 \na \psi_1) + tr(\psi^{\dag}_2 \na \psi_2) +tr(\psi^{\dag}_3 \na \psi_3) \nonumber \\
&& - tr(\phi_2^\dag \na \xi_3) - tr(\phi_3^\dag \na \xi_1) - tr(\phi_1^\dag \na \xi_2) \nonumber \\
&&- tr(\psi_0^\dag \na \psi_0) \ea

\noindent In the third last line we have the fermionic terms in Georgi-Glashow model for three
families of fields in representation $10$. In this way we can use the lagrangian $tr(\psi^{\dag} \na \psi)$,
with $\psi$ in the adjoint representation, in place of Georgi-Glashow terms with $\psi_{1,2,3}$ in
the skew symmetric representation. Mixed terms in the second last line give a reason to CKM and PMNS matrices
which appear in standard model. Consider now the equivalence

$$tr(\psi^{\dag} \psi \na ) = tr( \psi \na \psi^{\dag}) =
tr( (-\psi^{\dag}) \na (-\psi) )= tr(\psi^{\dag} \na \psi).$$

\noindent Hence

\be tr(\psi^{\dag} \na \psi ) = \fr 12 tr(\psi^{\dag} \{\na, \psi \}).\label{anticom}\ee

\noindent In this formalism, given $\w \in su(3)\otimes su(2) \otimes u(1)$, the transformation $\d\psi = [\w,\psi]$
corresponds to the usual transformation $\d\psi = \w\psi$ in the standard model formalism. We see that the
only fields which transform correctly under $SO(1,3)$ are $e$, $\nu$ and $d^c$. For now we do not care.

We note rather that, when we restrict the elements of $\w$ from the hyperions to the complex numbers, we
have $3$ possibilities to do it. A complex number is not only in the form $a+ib$, with $a,b \in R$, but also
$a+jb$ and $a+kb$. The same is true for a fixed linear combination $a+(ci+dj+fk)b$, where $c,d,f \in R$ and
$c^2 + d^2 + f^2 =1$. The choice of $j$ in place of $i$ determines another set of ($\w,\psi)$ isomorphic to
the first one. In the same way we obtain a third set choosing $k$. Note that for a $i$-complex left field we
have an $Ii$-complex right field and so on for $j$ and $k$.

The three sets are related by the group $SU(2)$ which rotates an unitary vector in $R^3$ with coordinates
$(c,d,f)$. Its generators are

$$\w = \fr {\vec{y}}{6} \left( \begin{array}[c]{cccccc}
 1         & 0          & 0          & 0          & 0          & 0          \\
 0         & 1          & 0          & 0          & 0          & 0          \\
 0         & 0          & 1          & 0          & 0          & 0          \\
 0         & 0          & 0          & 1          & 0          & 0          \\
 0         & 0          & 0          & 0          & 1          & 0          \\
 0         & 0          & 0          & 0          & 0          & 1          \\
\end{array} \right) .$$

Their diagonal form suggests an identification between this group and the gravitational
group $SU(2)^{\subset SO(1,3)}$. If the two groups coincided, all fields would transform
correctly under $SU(2)^{\subset SO(1,3)}$. By extending this group to the entire $SO(1,3)$,
we see that boosts exchange left fields with right fields.

Note that three families have to exist also for bosonic particles (photon, $W^\pm$, $Z$, gluons)
although they are probably in\-di\-stin\-gui\-sha\-ble. Other interesting thing is that we have no warranty
for the persistence of $Sp(12,\mathbf{C})$ in the entire universe. However we have surely at least
the symmetry $U(1,\mathbf{Y}) = SO(1,3)$, which implies the secure existence of gravity.

\subsection{Fermions from an extended arrangement matrix}
We introduce the following entities:
\begin{itemize}
{\setlength\itemindent{-5mm} \item I-complex grassmannian coordinates $\theta = \theta^1 + I\theta^2$ and $\bar{\theta} = \theta^1 - I\theta^2$;}
{\setlength\itemindent{-5mm} \item Grassmannian derivatives $\pa_g$ and $\bar{\pa}_g$, with $\pa_g \theta = \bar{\pa}_g \bar{\theta} =1$ and $\pa_g \bar{\theta} = \bar{\pa}_g \theta =0$;}
{\setlength\itemindent{-5mm} \item Grassmannian covariant derivatives $\na_g = \pa_g +\psi$ and $\bar{\na}_g^\dag = \bar{\pa}_g+ \bar{\psi}^\dag$.}
\end{itemize}
The fundamental products return
$$ \theta\theta = \theta^1 \theta^1 + \theta^1 I \theta^2 + I \theta^2 \theta^1 - \theta^2\theta^2 = 0 + I \theta^1 \theta^2 - I \theta^1 \theta^2 - 0 = 0 \nonumber $$
\ba \bar{\theta}\bar{\theta} &=& \theta^1 \theta^1 - \theta^1 I \theta^2 - I \theta^2 \theta^1 - \theta^2\theta^2 = 0 - I \theta^1 \theta^2 + I \theta^1 \theta^2 - 0 = 0 \nonumber \\
 \theta\bar{\theta} &=& \theta^1 \theta^1 - \theta^1 I \theta^2 + I \theta^2 \theta^1 + \theta^2\theta^2 = - I \theta^1 \theta^2 - I \theta^1 \theta^2 = -2 I \theta^1 \theta^2 \nonumber \ea

\noindent In the arrangement field formalism, covariant derivatives descend from a grassmanian matrix $M_g$ or $\bar{M}_g^\dag$.
We can consider a unique generalized matrix $M_T = M_g + M$ that, up to a generalized $U(n,\mathbf{Y})$, becomes

\ba M_T &=& \theta \na_g + d^\mu \na_\mu = \theta \pa_g + \theta \psi + d^\mu \na_\mu \nonumber \\
    \bar{M}^\dag_T &=& \bar{\na}^\dag_g \bar{\theta} + \bar{\na}^\dag_\mu \bar{d}^{\dag\mu} = \bar{\pa}_g \bar{\theta} + \bar{\psi}^\dag \bar{\theta}  + \bar{\na}^\dag_\mu  \bar{d}^{\dag\mu}.\label{deriv} \ea

\noindent Expanding $tr\,(\bar{M}^\dag_T M_T)$ we obtain

\ba tr\,(\bar{M}^\dag_T M_T) &=& tr\,\left(d^\nu \bar{d}^{\dag\mu} \bar{\na}^\dag_\mu \na_\nu  \right) = \sum_a \sqrt h R(x^a) \label{ricci2}.\ea
To calculate $tr\,(\bar{M}^\dag_T M_T\bar{M}^\dag_T M_T)$ we write first $\bar{M}_T^{\dag 2}$ and $M_T^{2}$.

\ba M^2_T &=& \theta \pa_g + \theta \psi + \theta d^\mu \{ \na_\mu,  \psi\} + d^\mu \na_\mu d^\nu \na_\nu \nonumber \\
    \bar{M}^{\dag 2}_T &=& \bar{\pa}_g \bar{\theta} + \bar{\psi}^\dag  \bar{\theta} + \{\bar{\psi}^\dag, \bar{\na}^\dag_\a \} \bar{d}^{\dag \a} \bar{\theta} + \bar{\na}^\dag_\a \bar{d}^{\dag \a}\bar{\na}^\dag_\b \bar{d}^{\dag \b} \ea
If $M$ has the form (\ref{deriv}), then $[M_T,\bar{M}^\dag_T] = 0$. This implies

$$tr\,(\bar{M}^\dag_T M_T\bar{M}^\dag_T M_T) = tr\,(M_T^2 \bar{M}_T^{\dag 2}).$$

\noindent We calculate its value starting from the following product

\ba tr\,(\theta d^\mu \{\na_\mu, \psi \}\{\bar{\psi}^\dag , \bar{\na}^\dag_\a \}\bar{d}^{\dag \a}\bar{\theta} ) &=& tr\,(\theta \bar{\theta} d^\mu \{\na_\mu, \psi \}\{\bar{\psi}^\dag , \bar{\na}^\dag_\a \}\bar{d}^{\dag \a}).\nonumber \\ && \ea

\noindent Remember that operator $tr$ acts as a sum over vertices. Now every vertex is labeled by a couple $(\theta, x_i)$ and then

$$ tr\, (\theta \bar{\theta} (***)) = \left( \int d\bar{\theta} d\theta\,\theta \bar{\theta} \right)tr\,(***) = tr\,(***)$$
Hence

\ba tr\,(\theta d^\mu \{\na_\mu, \psi \}\{\bar{\psi}^\dag , \bar{\na}^\dag_\a \}\bar{d}^{\dag \a}\bar{\theta}) &=& tr\,(d^\mu \{\na_\mu, \psi \}\{\bar{\psi}^\dag , \bar{\na}^\dag_\a \}\bar{d}^{\dag \a}) \nonumber \\
       &=& tr\,(\bar{d}^{\dag \a} d^\mu [\na_\mu, \bar{\na}^\dag_\a] \psi \bar{\psi}^\dag)  \nonumber\\
       &=& \sum_a \sqrt h R(x^a) \bar{\psi}^\dag \psi\ea

\noindent In this way

\ba tr\,(\bar{M}^\dag_T M_T\bar{M}^\dag_T M_T) &=& tr\, \left(\bar{\psi}^\dag d^{\mu} \{\na_\mu, \psi\} + \{\bar{\psi}^\dag, \bar{\na}^\dag_\a\} \bar{d}^{\dag \a}\psi \right) + \nonumber \\
&& + \sum_a \sqrt h R(x^a) \bar{\psi}^\dag \psi + S_{GB} \label{ora}\ea

\noindent We have seen that every family distinguishes itself by the choice of complex unity.
Inserting in $\psi$ the definitions of $\psi_{1,2,3}$ we can write

\ba \psi &=& \psi_0^1 + i(\phi_2^1 + \xi_3^1) + j(\phi_3^1 + \xi_1^1) +k(\phi_1^1 + \xi_2^1) + \nonumber \\
        && + I\psi_0^2 + iI(\phi_2^2 + \xi_3^2) + jI(\phi_3^2 + \xi_1^2) + kI(\phi_1^2 + \xi_2^2) \nonumber \ea
Using the correspondences $(1,I,i,j,k,iI,jI,kI) \leftrightarrow \g\g$ and $4 \leftrightarrow tr$, the first term in (\ref{ora}) becomes

$$ 2 \times \fr 14 \times tr\, \left( {\psi}^{lm} \overline{(\g_l\g_m)^\dag} \left(
        \g_0 \g_s e^{\mu s} \overset{G}{\na}_\mu \psi^{np} (\g_n \g_p) + A_\mu \psi \right)\right) $$

\be \downdownarrows === \downdownarrows \ee

$$ \fr 12 \, tr\, \left( \psi^{lm} (\g_m\g_l) \left(
        \g_0 \g_s e^{\mu s} \overset{G}{\na}_\mu, \psi^{np} (\g_n \g_p) + A_\mu \psi_0 + \sum_{q,q'=1}^3  A^q_\mu \psi_{q'} i_{q'} \right)\right) $$

\noindent Here we have deleted the anticommutator by means of (\ref{anticom}).
In the covariant derivative we have included only the gravitational (track) contribution,
while $A_\mu$ is intended to have null track. Moreover $i_1 =k$, $i_2 =i$ and $i_3 =j$.

In the second line we have divided the $75$ generators $A_\mu$ in three families of $35$ generators.
Obviously, only two families are linearly independent. When they act on spinorial fields which belong
to their own family, they behave exactly as the $35$ generators of $SU(6)$ (which comprise the $24$
generators of $SU(5)$). Conversely, when a generator $A^q$ acts on a $q'$-field (with $q \neq q'$), it
mimics the application of some generator $A^{q'}$ followed by a rotation in $SU(2)_{GRAVITY}$ which
sends the family $q'$ in the remaining family $q''$.

We explicit now one entry of $\psi = \psi^1 + I\psi^2$ by exploiting the correspondence with $\g\g$.
We have

$$ \psi = \left( \begin{array}[c]{cc}
\psi_0^1 + i(\phi^1_2+\xi_3^1)&(\phi_3^1+\xi_1^1) +i(\phi^1_1 +\xi^1_2) \\
-(\phi_3^1+\xi_1^1)+i(\phi^1_1 +\xi^1_2)&\psi_0^1-i(\phi^1_2+\xi_3^1) \\
i\psi_0^2-(\phi^2_2+\xi_3^2)&i(\phi_3^2+\xi_1^2)+(\phi^2_1 +\xi^2_2) \\
-i(\phi_3^2+\xi_1^2)+(\phi^2_1 +\xi^2_2)& i\psi_0^2+(\phi^2_2+\xi_3^2) \\
\end{array}\right. $$
$$ \left. \begin{array}[c]{cc}
i\psi_0^2-(\phi^2_2+\xi_3^2) & i(\phi_3^2+\xi_1^2)+(\phi^2_1 +\xi^2_2)   \\
-i(\phi_3^2+\xi_1^2)+(\phi^2_1 +\xi^2_2) &i\psi_0^2+(\phi^2_2+\xi_3^2)    \\
\psi_0^1+i(\phi^1_2+\xi_3^1) &(\phi_3^1+\xi_1^1)+i(\phi^1_1 +\xi^1_2)     \\
-(\phi_3^1+\xi_1^1)+i(\phi^1_1 +\xi^1_2) &\psi_0^1-i(\phi^1_2+\xi_3^1)    \\
\end{array}\right) $$

\noindent If we define the four components spinor

$$ \hat{\psi} = \left(\begin{array}[c]{c} \psi_0^1 + i(\phi^1_2+\xi_3^1) \\
        -(\phi_3^1+\xi_1^1)+i(\phi^1_1 +\xi^1_2) \\ i\psi_0^2-(\phi^2_2+\xi_3^2) \\
        -i(\phi_3^2+\xi_1^2)+(\phi^2_1 +\xi^2_2) \end{array} \right) $$

\noindent the derivative term can be rewritten as

\be 2 \times \hat{\psi}^\dag \, \g_0 \g_s e^{\mu s} \overset{G}{\na}_\mu \hat{\psi} \ee
This is the Dirac action, although with a new interpretation of spinorial components.
Moreover

$$\psi^{AB} = W^{ABC} \hat{\psi}^C \qquad ; \qquad W^{ABC} W^{DBC} = \mathbf{1}_{AD}$$

\LARGE
$$ \left(W^{ABC}\right) = $$
$$\pt$$
\normalsize
$$ \left( \left( \begin{array}[c]{cccc}
1 & 0 & 0 & 0  \\
0 & 1 & 0 & 0  \\
0 & 0 & 1 & 0  \\
0 & 0 & 0 & 1  \\
\end{array}\right),
\left( \begin{array}[c]{cccc}
0 & -\ast &  0 & 0  \\
\ast &  0 &  0 & 0  \\
0 &  0 &  0 & \ast  \\
0 &  0 & -\ast & 0  \\
\end{array}\right),
\left( \begin{array}[c]{cccc}
0 & 0 & 1 & 0  \\
0 & 0 & 0 & 1  \\
1 & 0 & 0 & 0  \\
0 & 1 & 0 & 0  \\
\end{array}\right),
\left( \begin{array}[c]{cccc}
0 & 0 & 0 & \ast  \\
0 & 0 & -\ast & 0  \\
0 & -\ast & 0 & 0  \\
\ast & 0 & 0 & 0  \\
\end{array}\right)
\right) $$
$$ \pt $$
\be \pt \ee
\large

\noindent with $\ast \hat{\psi} = \hat{\psi}^\ast$. Adding the other terms

$$ tr\,(\hat{M}^\dag \hat{M} \hat{M}^\dag \hat{M}) = S_{GB} +$$
$$ + 2 \int \left( \hat{\psi}^\dag\, \g_0 \g_s e^{\mu s} \overset{G}{\na}_\mu \hat{\psi}
 + \hat{\psi} \sum_{q,q'} A^q_\mu \hat{\psi}_{q'} i_{q'} +
 \sqrt h R(x)\sum_q \hat{\psi}_q^\dag \hat{\psi}_q \right) dx  $$

In this way we include all the contents of standard model as elements in the
generalized $Sp(12,\mathbf{C})$ algebra. Terms which mix fa\-mi\-lies can be used to calculate values
in CKM and PMNS matrices. Masses for fermionic fields arise, as usual, from non null expectation
values of $A_\mu(x^a_i, x^b_j)$ with $a\neq b$ in $\na_\mu$.

We obtain a contribute to Hilbert-Einstein action also from term $\int d^4x \sqrt h R \bar{\psi} \psi$,
due to a non null expectation value of $\sum_q  \bar{\psi}_q \psi_q $. It contains in fact the chiral
condensate, whose non null vacuum value breaks the chiral flavour symmetry of QCD Lagrangian.

Note that known fermionic fields fill up a matrix $\psi$ with null track. However, only if $tr\,\psi \neq 0$
our action has an extra invariance under

\ba A_\mu &\ra& d_\mu^{-1} \theta \psi \nonumber \\
    \psi &\ra& \overleftarrow{\pa}_g d^\mu A_\mu .\label{super}\ea
Here $\overleftarrow{\pa}_g$ is a $\pa_g$ which acts backwards.
This means we have the same number of fermions and bosons, so that the vacuum energies erase each other.

Invariance (\ref{super}) predicts the existence of a new colored fermionic sextuplet which sits
on diagonal in $\psi$. Inside it we can include a conjugate neutrino ($\nu^c$), a sterile neutrino
($N$) and a conjugate sterile neutrino ($N^c$). Explicitly

$$\psi = \left( \begin{array}[c]{cccccc}
N & 0 & 0 & 0 & 0 & 0 \\
0 & \nu^c & 0 & 0 & 0 & 0 \\
0 & 0 & \nu^c & 0 & 0 & 0 \\
0 & 0 & 0 & N^c & 0 & 0 \\
0 & 0 & 0 & 0 & N^c & 0 \\
0 & 0 & 0 & 0 & 0 & N^c \end{array}
\right).$$

\noindent This field commutes with any gauge field in $U(1) \otimes SU(2) \otimes SU(3)$
and so it hasn't electromagnetic, weak or strong interactions. Moreover it gives a Dirac mass to
neutrinos via the term

$$tr\,(\bar{\psi}^\dag d^\mu A_\mu \psi) = \bar{\psi}^{\dag ij} d^\mu A_\mu^{kl} \psi^{mn} f^{(ij)(kl)(mn)}.$$
Here $f^{(ij)(kl)(mn)}$ are structure constants for $SU(6)$ and masses for neutrinos are eigenvalues of $<d^\mu A_\mu >$.

\subsection{The vector superfield}

The invariance (\ref{super}) suggests a connection with super-symmetric theories.
We redefine the supersymmetry algebra as follows:

\ba Q = \pa_g - d^\mu \bar{\theta} \widetilde{\na}_\mu \qquad\quad\!\! &;&\quad\qquad\!\! [Q, \widetilde{\na}_\nu] = - d^\mu \bar{\theta} \widetilde{R}_{\mu\nu} \nonumber \\
\bar{Q} = \bar{\pa}_g - \theta \bar{d}^{\dag \nu} \widetilde{\na}_\nu \qquad\quad\!\! &;&\quad\qquad\!\!
\{Q, \bar Q \} = -2d_H^\nu \widetilde{\na}_\nu - {\Sigma}^{\mu\nu} \widetilde{R}_{\mu\nu}\nonumber \\
2 d_H^\nu &=& d^\nu + \bar{d}^{\dag\nu} \ea

\noindent Here $\widetilde{\na}$ is a compatible covariant derivative which acts as a skew-adjoint operator.
It is a functional of $d^\mu$ with $[\widetilde{\na}_\nu, d^\mu] = [\widetilde{\na}_\nu, \bar{d}^{\dag\mu}] = 0$.
Note that off-shell we have $\widetilde{\na}_\nu \neq \overset{G}{\na}_\nu$. Moreover

$${\Sigma}^{\mu\nu} = d^\mu \bar{d}^{\dag\nu} \theta \bar{\theta} $$
$\widetilde{R}_{\mu\nu}$ is the curvature tensor made with $\widetilde{\na}$, ie $\widetilde{R}_{\mu\nu} = [\widetilde{\na}_\mu, \widetilde{\na}_\nu]$.
Consider that locally we can find a coordinate system where $\widetilde{R}_{\mu\nu} = 0$, recovering the usual SUSY algebra with $-i\widetilde{\na}_\mu$ in place of $P_\mu$.
The vector superfield assumes the simple form\footnote{As usually in this work, we absorb an $i$ in the fields to make them skew-hermitian.}

\ba V &=& \theta \psi + d^\mu A_\mu \nonumber \\
 V_{AB} &=& \theta_{AD} W^{D}_{\pt\pt BC} \hat{\psi}^C + (d^\mu  A_\mu)_{AB} \ea

\noindent with $\theta_{AD} = \mathbf{1}_{AD}\,\theta^1 + \g_{5\,AD} \,\theta^2$.
Note that $M_T = V+ \theta \pa_g + d^\mu \pa_\mu$ and then
the usual kinetic term for $V$ includes the same terms we have found in $tr\,(\bar{M}_T^\dag M_T) -
tr\,(\bar{M}^\dag_T M_T\bar{M}^\dag_T M_T)$. It's remarkable that all the known fermionic
fields take the role of gauginos for all the known bosonic fields. In this way the right up
quarks are gauginos for gluons, while right electrons are gauginos for $W$ bosons.
This is permitted because both fermions and bosons in AFT transform in the adjoint
representation of $Sp(12,\mathbf{C})$. In this way our theory includes SUSY $N = 1$ with no need
for new unknown particles.

\section{Inflation}

Our final action is
\label{inflation}
$$ S = tr\,\left(\fr {\bar{M}^\dag M }{16\pi G}- \bar{M}^\dag M \bar{M}^\dag M\right)$$
This is also an action for an $U(n,\mathbf{Y})$ gauge theory with coupling constant $1/G$ in a
mono-vertex space-time. In these theories the scaling of coupling constant can be calculated
exactly in the limit of large $n$. In several cases the coupling constant changes its sign for
big values of scale: this has considerable consequences for the first times after Big Bang, when
a measurement of $G$ has sense only at very high energies (very small distances). What said suggests
that such measurement can return a negative value of $G$, which implies a repulsive force of gravity.
In turn, repulsive gravity implies an accelerate expansion for the universe.

Because the entries of $M$ are probability amplitudes, we would be it was
dimensionless. However, when we pass from $M$ to $\na$, we need a scale
$\Delta$ to define the matrix $\pa$. This justify the inclusion of $\Delta^{-1}$
inside $M$. If we extract this factor, the Hilbert Einstein action becomes

$$\fr {\Delta^4}{16\pi G \Delta^2}tr\,(\bar{M}^\dag M) = \fr {\Delta^2}{16\pi G}tr\,(\bar{M}^\dag M)$$

\noindent where we have also added the correct volume form $\Delta^4$. This
seems a more natural formulation when $M$ represents probability amplitudes.
In this way we can take $\Delta$ very small but not zero. The most natural
choice is $\Delta^2 \sim G$.

In this case, what does it mean that $G$ is negative? Negative $G$ implies
negative $\Delta^2 = ds^2$. In lorentzian spaces $\Delta^2 = dt^2 -ds^2 <0$.
For purely temporal intervals we'll have $dt^2 < 0$, so the time becomes
imaginary. An imaginary time is indistinguishable from space. This hypothesis
of a \lq\lq spatial'' time had already been espoused by Hawking as a solution
for eliminate the singularity in the Big Bang \cite{Hawking}.

\section{Classical solutions}
\label{classical}
We rewrite our action in the form

$$S = \fr 12 tr\, (\bar{M}^\dag M) - \fr 1{4g} tr\,(\bar{M}^\dag M \bar{M}^\dag M)$$
where we have defined $g = \fr {\Delta^2} {32\pi G}$. We diagonalize $M$ with a
transformation in $U(n,\mathbf{Y})$ and define $M^{ii} \equiv \f(x_i)$, $\f(x)
= a(x)+\vec{b}(x)$. The lagrangian becomes:

$$L = \fr 12 \left[ a(x_i)^2 + |\vec{b}(x_i)|^2\right]  -\fr 1{4g} \left[ a(x_i)^4
+ |\vec{b}(x_i)|^2 + 2a(x_i)^2 |\vec{b}(x_i)|^2 \right]$$

\noindent The motion equations are

$$ g a(x) - a(x)^3 - a(x)|\vec{b}(x)|^2 = 0$$

$$ g \vec{b}(x) - \vec{b}(x)|\vec{b}(x)|^2 - a(x)^2\vec{b}(x) = 0$$

\noindent There are two solutions:

$$(1) \qquad a(x) = \vec{b}(x) = 0$$
$$(2) \qquad a(x)^2 + |\vec{b}(x)|^2 = \bar{M}^\dag M = g$$

\noindent The first one corresponds to the vacuum (all non-gravitational fields equal
to zero) plus a solution of Einstein equations in the vacuum:

$$\psi = A_\mu = 0 \qquad R(x) = 0$$

The solution $\bar{M}^\dag M = g$ corresponds to a vacuum expectation value for $\bar{M}^\dag M$ equal
to $g$. $M$ contains a factor $A$, so that an expectation value for $\bar{M}^\dag M$ corresponds
to an expectation value for $AA$. This means that

$$AAAA = <AA>AA + \text{quantum perturbations}$$

\noindent $<AA>$ gives a mass for $A$.

More precisely, for $A \in U(n,\mathbf{Y})/U(m,\mathbf{Y})^{n/m}$,

$$m_A^2 \sim \fr {<\bar{M}^\dag M>}{\Delta^2} = \fr g {\Delta^2} = \fr 1 {32\pi G}$$
So the fields $A \in U(n,\mathbf{Y})/U(m,\mathbf{Y})^{n/m}$ have a mass in the order of
Planck mass $m_P$. Moreover, in the primordial universe, when $k_B T \approx m_p$, all
the fields behave like null mass fields. In that time the symmetry was then $U(n,\mathbf{Y})$
and no arrangement exists. Our conclusion is that Quantum Gravity cannot be treated as a
quantum field theory in an ordinary space. In what follows we explain how overcome this
trouble.

\section{Quantum theory}
\label{quantize}
\begin{spacing}{1}
Quantum theory is defined via the following path integral:

\ba && \!\!\!\!\!\!\!\!\int D[M(x,y)]D[\bar M^* (x,y)] \nonumber \\
&& \qquad\qquad\quad \!\! Oe^{\int M(x,y)\bar{M}^*(x,y) dx dy - \int M(x,y)\bar M^*(x,y')M(x',y')\bar M^*(x',y) dx dy dx' dy'}\nonumber \ea
\begin{center}
with
\end{center}
\ba && \!\!\!\!\!\!\!\! Oe^{\int F(x,y)dx dy} = \nonumber \\
&& = 1+ \int F(x,y)dx dy +
\fr 12 \int F(x,y)F(x^1,y^1)dx dy dx^1 dy^1 + \nonumber \\
&&\qquad + \ldots + \fr 1 {n!} \int F(x,y)F(x^1,y^1)\ldots F(x^{n-1},y^{n-1})
dx dy dx^1 dy^1 \ldots \nonumber \\
&& \qquad\qquad\qquad\qquad\qquad\qquad\qquad\qquad \ldots dx^{n-1} dy^{n-1} \nonumber \ea
\end{spacing}

\ba Oe^{\int F(x,x',y,y')dx dx' dy dy'} &=& 1+ \int F(x,x',y,y')dx dy dx' dy'+ \nonumber \\
&&\!\!\!\!\!\!\!\!\!\!\!\!\!\!\!\!\!\!\!\!\!\!\!\!\!\!\!\!\!\!\!\!\!\!\!\!\!\!\!\!\!\!\!\!\!\!\!\!\!\!\!\!\!\! +
\fr 12 \int F(x,x',y,y')F(x^1,{x'}^1,{y'}^1,y^1)dx dy dx' dy' dx^1 dy^1 d{x'}^1 d{y'}^1 + \nonumber \\
&& \!\!\!\!\!\!\!\!\!\!\!\!\!\!\!\!\!\!\!\!\!\!\!\!\!\!\!\!\!\!\!\!\!\!\!\!\!\!\!\!\!\!\!\!\!\!\!
+ \ldots + \fr 1 {n!} \int F(x,x',y,y')F(x^1,{x'}^1,y^1,{y'}^1)\ldots \nonumber \\
&& \!\!\!\!\!\!\!\!\!\!\!\!\!\!\!\!\!\!\!\!\!\!\!\!\!\!\!\!\!\!\!\!\!\!\!\!\!\!
\ldots F(x^{n-1},{x'}^{n-1},y^{n-1},{y'}^{n-1}) dx dy dx' dy' dx^1 dy^1 d{x'}^1 d{y'}^1\ldots \nonumber \\
&& \qquad \qquad\qquad \quad\ldots dx^{n-1} dy^{n-1} d{x'}^{n-1} d{y'}^{n-1} \nonumber \\
&& \quad \overset{!}{=} \fr 1 {n!} F^n \ea
\begin{spacing}{2}
\noindent The integration of $F^n$ is very simple and gives

$$\fr 1 {n!} \int D^2[M] e^{\int M^2 dx dy} \,\,F^n = \fr {(4n)!}{n! 2^{2n}(2n)!} = \fr 1 {n!} P(4n)$$

\noindent Here $P(4n)$ gives the number of different ways to connect in couples $4n$ points.

It's clear that any universe configuration corresponds to an $F^k$ inside which
some connections have been fixed and the corresponding integrations have been
removed. For example:
\end{spacing}
\newpage
\begin{figure}[ptbh]
\centering\includegraphics[width=1\textwidth ]{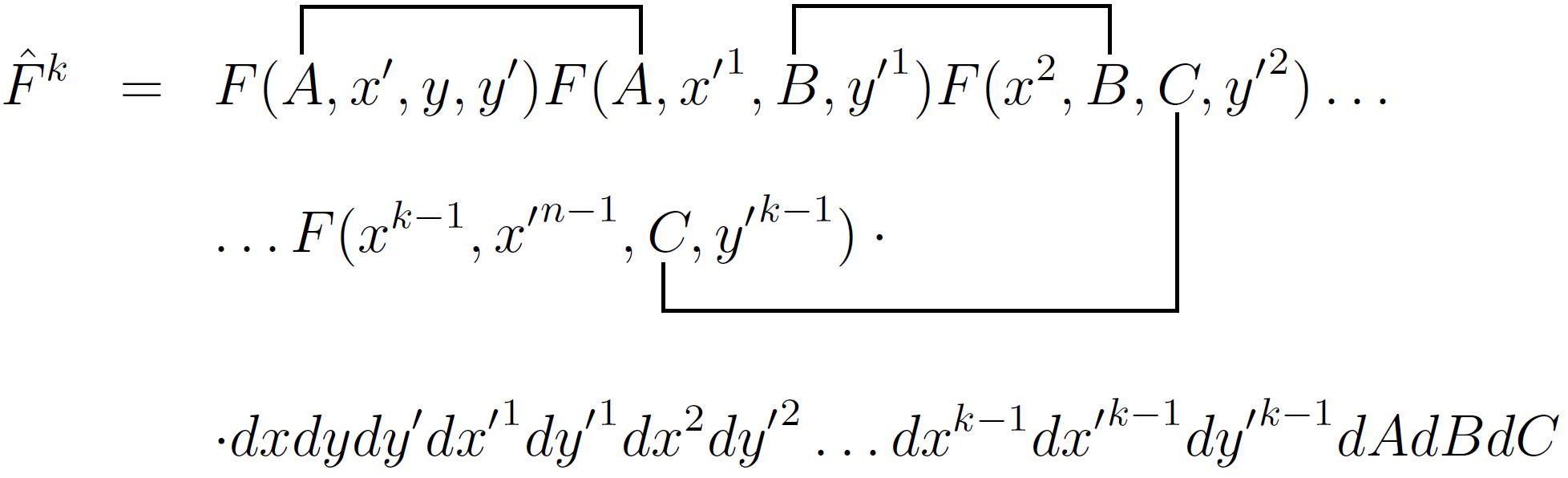}
\end{figure}

\noindent If the fixed connections are $m$, then

$$<\hat F^k > = \fr{\sum_n \fr 1 {n!} P(4(n+k)- 2m)}{\sum_n \fr 1 {n!} P(4n)}$$

\begin{notation}
\noindent At relatively low energies we can tract $\overset{G}{A}$ as
an ordinary gauge field. The arrangement field theory is then
approximated with a common quantum theory on a curved background,
determined by $e^{\mu a}$.
\end{notation}

\section{Quantum Entanglement and Dark Matter\label{entanglement}}

The elements of $M$ which do not reside in or near the diagonal, describe connections between
points that are not necessarily adjacent to each other, in the common sense.
These connections construct discontinuous paths as in figure \ref{cammini-entanglement}
and can be considered as quantum perturbations of the ordered space-time.

Such components permit us to describe the quantum entanglement effect, as it could be shown
in detail in a complete coverage that goes beyond the purpose of the present paper.

It is remarkable that in this framework the discontinuity of paths is only apparent, and it is a
consequence of imposing an arrangement to the space-time points. These discontinuous paths
can be con\-si\-de\-red as continuous paths which cross wormholes. The trait of path inside a wormhole
is described by a component of $M$ far away from diagonal. The information seems to travel
faster than light, but in reality it only takes a byway.

\begin{figure}[ptbh]
\centering\includegraphics[width=0.6\textwidth ]{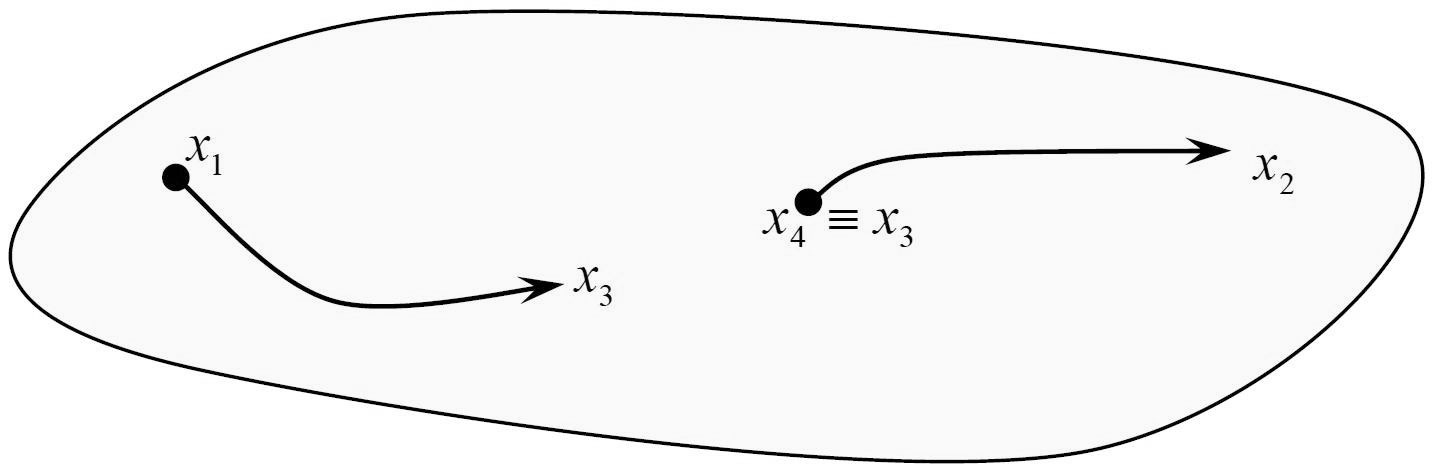}\caption{Discontinuous paths.
The connection between $x_3$ and $x_4$ is done by a component of $M$ far away from diagonal.}
\label{cammini-entanglement}
\end{figure}

Imagine now a gravitational source with mass $M_S$ which emits some gravitons
with energy $\sim E_{PLANCK}$, directed to an orbiting body with mass $M_B$ at distance $r$.
In this case (respect such gravitons) the fields $M(x^a, x^b)$ with $a \neq b$ would behave
as they had null mass. This implies the probable existence of connections (practicable by such
gravitons) between every couple of vertices in the path from the source to the orbiting body.
This means that if $r = \Delta j$, $j \in \mathbf{N}$, the graviton could reach the
orbiting body by traveling a shorter path $\Delta j'$, $j>j' \in \mathbf{N}$.
The question is: what is the average gravitational force perceived by the orbiting body?

The probability for a graviton to reach a distance $r$ passing through $m$ vertices is

$$P_m = (1-a)^{m-1}a \qquad with\,\,\,\sum_{m=1}^\infty P_m = 1$$
where $a =  1/j$.
These are the probabilities for extracting one determined object
in a box with $j$ objects at the $m$-th attempt.
In this way the effective length traveled by the graviton will be $\Delta m$.

We use these probabilities to compute the average gravitational force in a semiclassical approximation.

\ba <F> &=& G\fr{M_{B}M_{S}}{\Delta^2} \fr a {1-a} \int_1^\infty \fr {(1-a)^m}{m^2} dm\nonumber \\
        &=& G\fr{M_{B}M_{S}}{\Delta^2} \fr a {1-a} [log(1-a)] \int_{log(1-a)}^{-\infty} \fr {e^x}{x^2} dx \nonumber \\
        && \label{dark}\ea
The last integral gives

$$\int_{log(1-a)}^{-\infty} \fr{e^x}{x^2}dx = -Ei(log(1-a)) + \fr {1-a}{log(1-a)}$$

We expand $\langle F \rangle$ near $a=0$ (which implies $j >> 1$), obtaining

$$\fr a {(1-a)} [log(1-a)]\int_{log(1-a)}^{-\infty} \fr{e^x}{x^2}dx \approx a + a^2(log(a) + \g) + O(a^3).$$

Here $\g$ is the Eulero-Mascheroni constant. The dominant contribution is then

\ba <F> &\approx& G\fr{M_{B}M_{S}}{\Delta^2} \cdot a \cdot (1+a\,log(a) + a\g) \nonumber \\
     &\approx& \fr{G}{\Delta}\fr{M_{B}M_{S}}{r}\left( 1 - \fr \Delta r \left( log\left( \fr r \Delta \right) - \g \right) \right)\ea

If the massive object orbits at a fix distance $r$, its centrifugal force has to be equal to the
gravitational force. This gives

$$<F> \approx \fr{G}\Delta\fr{M_{B}M_{S}}{r}\left( 1 - \fr \Delta r \left( log\left( \fr r \Delta \right) - \g \right) \right) = \fr {M_B v^2} r$$

$$v^2 = \fr{G M_{B}M_{S}}{\Delta}\left( 1 - \fr \Delta r \left( log\left( \fr r \Delta \right) - \g \right) \right)$$
We see that, varying the radius, the velocity remains more or less constant (It increases slightly with $r$).
Can this explain the rotation curves of galaxies without introducing dark matter?

Surely not all gravitons have energy $> E_{PLANCK}$; at the same time we have to consider that
$G$ scales for small distances (hence for small $m$ in (\ref{dark})).
It's possible that these factors reduces the extremely high value of $r/\Delta$.

\section{Conclusion}

In this work we have applied the framework developed in \cite{Arrangement} to
describe the contents of our universe, ie forces and matter.

Doing this, we have discovered an unexpected road toward unification, which shows
similarities with Loop Gravity, String Theory and Georgi - Glashow model.
For the first time a natural symmetry justifies the existence of three particles
families, not one more, not one less. Moreover a new version of supersymmetry
seems to couple gauge fields with all known fermions, without necessity of
imagining new particles never seen by experiments.

Clearly this fact closes the door to dark matter. To compensate this big absence,
AFT proposes an explanation to galaxy rotation curves which doesn't make use of
dark matter.

Another considerable implication of AFT regards tangent space, which has symmetry $SO(1,3)$
only when gravity decouples from other forces. At that point also the real space-time
can obtain the same symmetry. This fact is coherent with \emph{no-go theorem} of Coleman-Mandula\cite{nogo},
under which \lq\lq $S$-matrix is Lorentz invariant if and only if the action symmetry is
$SO(1,3) \otimes internal\,\,symmetries$''.

We don't say that this theory is exact. However there are several good signals
which must be taken into account. We hope that a future teamwork can verify
this theory in detail, deepening all its implications.

\chapter{Antigravity in AFT}

\section{Introduction}
Arrangement field theory is a quantum theory defined by means of probabilistic
spin-networks. These are spin-networks where the existence of an edges is
regulated by a quantum amplitude. AFT is a proposal for an unifying
theory which joins gravity with gauge fields. See \cite{Arrangement} and \cite{Arrangement2}
for details. The unifying group is $Sp(12,\mathbf{C})$ for the lorentzian theory
and its compact real form $Sp(6)$ for the euclidean theory.
The unifying group contains three indistinguishable copies of gauge fields,
mixed by gravitational field. Moreover, commutators between gravitational
and gauge fields are non null and give new terms for the Einstein equations.
In what follows we focus on the term which mixes gravity with electromagnetism,
showing that its contribution to Einstein equations could generate antigravity.
In the end we verify that new interactions don't affect the making of nucleus
and nucleons.

\section{Antigravity}
\label{formalism2}
\begin{spacing}{1.5}
The term which mixes gravity with electromagnetism is given by space-time
integration of the following expression:
\end{spacing}
\normalsize
$$ -\fr 14 f^{(G)(EM1)(EM2)} A^{(G)}_\mu A^{(EM1)}_\nu \bigg(
F^{(EM2)\mu\nu} + \a f^{(EM3)(EM1)(EM2)} A^{(EM3)\mu} A^{(EM1)\nu}
\bigg) $$
\large
\be \pt \label{prima}\ee
Remember that AFT includes three indistinguishable e\-lec\-tro-ma\-gne\-tic
fields, with non-trivial commutators. In this way $A^{(G)}$ is the
gravitational gauge field, $A^{(EMn)}$ is the n-th electromagnetic
field and $\a$ is the fine structure constant. In the realistic case of
null torsion, the gravitational gauge field can be rewritten in
function of the tetrad field:

$$A_\mu^{(G)bc} = \fr 12 e^{\nu [b} \pa_{[\mu} e^{c]}_{\nu]} + \fr 14
e_{\mu d} e^{\nu b} e^{\sigma c} \pa_{[\sigma} e^d_{\nu]}$$

\noindent From now we take a low energy limit so defined: $e_{ii} = 1$ with $i=1,2,3$,
$e_{00} = \theta(x)$ and $\pa_0 \theta(x) =0$. Varying with respect to $e$
we obtain:

$$\fr {\d A_\mu^{(G)bc}}{\d e^s_\tau} = \fr 12 e^{\nu [b} \d^{c]}_s
\d^\tau_{[\nu} \pa_{\mu]} + \fr 14 e_{\mu s} e^{\nu b} e^{\sigma c}
\d^\tau_{[\nu} \pa_{\sigma]}$$

$$\fr {\d A_\mu^{(G)bc}}{\d g_{\w\tau}} = 2e^{\w s} \fr {\d A_\mu^{(G)bc}}{\d e^s_\tau}
= e^{\w [c} e^{b] \nu} \d^\tau_{[\nu} \pa_{\mu]} + \fr 12 \d^\w_{\mu} e^{\nu b} e^{\sigma c}
\d^\tau_{[\nu} \pa_{\sigma]}$$
The component with $c = \w =\tau = 0$ and $b \neq 0$ results:

$$\fr {\d A_\mu^{(G)b0}}{\d g_{00}} = -\theta^{-1} \d^0_\mu \pa_b
- \fr 12 \theta^{-1} \d^0_\mu \pa_b = -\fr 3{2\theta} \d^0_\mu \pa_b$$

$$A^{(EM)\rho}A^{(EM)}_\rho A^{(EM)\mu}\fr {\d A_\mu^{(G)b0}}{\d g_{00}} = \fr 3{2\theta} \pa_b A^{(EM)0} A^{(EM)\rho}A^{(EM)}_\rho$$
The minus sign has disappeared because we have reversed the derivative.
The variation of quartic term in (\ref{prima}) with respect to $\d g_{00}$
is then given by

$$-\fr \a 4 \cdot \fr 3{2\theta} \pa_b f^b A^{(EM)0} A^{(EM)\rho}A^{(EM)}_\rho =
-\pa_b f^b \fr {3\a} {8\theta} V(\theta^2 V^2 - A^2)$$
$$f^b = \sum_{cade} f^{(bo)ca} f^{dea} \approx 4\fr {x^b}{r}.$$
Here we have indicated with $V$ the electric potential and with $A$ the
magnetic vector potential. The sum inside $f$ is over the three electromagnetic
fields.

It's so clear that varying the complete action with respect to $g_{\mu\nu}$
we obtain a new term for Einstein equations. In the Newtonian limit we can
substitute $g_{00} = -(1-2\phi)$ and $R_{00} - (1/2)Rg_{00} = \na^2 \phi$
where $\phi$ is the newtonian potential. Hence:

\ba 2 \na^2 \phi &\approx& 8\pi T^{00} = 8\pi \fr{-2}{\sqrt {-g}} \fr {\d \sqrt {-g} L_{matt}}{\d g_{00}} \nonumber \\
            &\approx& \pa_b \fr {x^b}{r} 24\pi \a V(\theta V^2 - \theta^{-1} A^2) \ea
For radial potential we have

$$\pa_b \phi = \fr {x^b}{r} \pa_r \phi .$$
In such case

$$C_G = \pa_r \phi \approx 12\pi\a V(\theta V^2 -\theta^{-1} A^2)$$
Now we insert the appropriate universal constants and approximate $\theta$  with $1$:

\be C_G \approx 12\pi \a \fr{(G\e_0)^{3/2}}{c^4 L_p} V(V^2 - c^2 A^2) = k V(V^2 - c^2 A^2) \label{ultima} \ee
Here $L_p$ is the Planck length, equal to $\sqrt {\hbar G/c^3}$. The multiplicative constant is

$$ k = \fr{12\pi}{137}\cdot \fr{(6,67\cdot 10^{-11}\cdot 8,85\cdot 10^{-12})^{3/2}}{(3\cdot 10^8)^4 \cdot(1,62\cdot 10^{-35})} = 30,27\cdot 10^{-33} \,\left(\fr{C^3 s^4}{Kg^3 m^5}\right).$$
This means that for having a weight variation (on Earth) of about $10\%$
($\Delta C_G =1$) we need an electrical potential of $10^{11}$ Volts.
These are $100$ billions of Volts. For $V = Q/r$ and $A=0$ we have:

$$C_G = \fr {k}{(4\pi\e_0)^3}\cdot \fr {Q^3}{r^3} = 2,198 \cdot 10^{-2} \left(\fr {m^4}{s^2 C^3}\right)\fr {Q^3}{r^3}$$
Note that the sign of $C_G$ is the sign of $Q$ and then we obtain
antigravity for negative $Q$. We associate to this interaction an
equivalent mass $m$, substituting $C_G = Gm/{r^2}$. We have

$$m = \fr k G V^3 r^2 = \fr {k}{G(4\pi \e_0)^3}\fr {Q^3} r = 3,293\cdot 10^8 \left(\fr {Kg\, m}{C^3}\right) \fr {Q^3}{r}$$
which is a negative mass for negative $Q$. Negative mass implies negative
energy via the relation $E =mc^2$.
Intuitively, if we search a similar relation for gravi-magnetic field
(which is $\na \times (g^{0i})$, $i=1,2,3$), we should find the same
formula (\ref{ultima}) with an exchange between $V$ and $cA$.

We calculate now at what distance the gravitational attraction between two
protons is equal to their electromagnetic repulsion.

$$G\fr {m^2}{r^2} = \fr {k^2}{G^2 (4\pi\e_0)^6} \fr {Q_p^6}{r^4} = \fr 1 {4\pi \e_0} \fr {Q_p^2}{r^2}$$

$$\fr {k^2 Q_p^4}{G^2 (4\pi\e_0)^5} = r^2$$

$$\Longrightarrow r^2 = 79,49 \cdot 10^{-70} m^2 \Longrightarrow r = 8,916 \cdot 10^{-35} m = 5,516\, L_p$$
Note that we are $20$ orders of magnitude under the range of strong force
and $23$ orders of magnitude under the range of weak force. In this way
the gravitational force doesn't affect the making of nucleus and nucleons.

\section{Conclusion}

We have seen that a potential of $10^{11}$ Volts can induce relevant
gra\-vi\-ta\-tio\-nal effects. They are too many for notice variations in the
experiments with particles accelerators. However they sit at the border
of our technological capabilities. The possibility to rule gravitation is
very attractive and constitutes a good reason for try experiments
with high electric potentials. Such experiments can be connected to the
work of Nikola Tesla and can also be a good test for the arrangement
field theory.

\end{document}